\RequirePackage{amsmath}
\documentclass[twocolumn]{svjour3}         

\smartqed  						
\usepackage{mathptmx}      		
\usepackage{graphicx}     			
\usepackage{multirow}            	
\usepackage{amsmath}				
\usepackage{amssymb}                
\usepackage{bm}			
\usepackage{upgreek}       

\usepackage{mathtools}		

\usepackage{xspace}            
\usepackage[caption=false]{subfig}            

\usepackage{enumitem}               

\usepackage{color}	
\usepackage{colortbl}             

\usepackage{tabularx}		

\usepackage{rotating}		

\usepackage{stmaryrd}	
\usepackage{soul}		

\usepackage[usenames,table,dvipsnames]{xcolor}

\usepackage{hyperref}   
\hypersetup{                                                            
	pdfpagemode=UseNone,               
    pdfstartview=Fit,
    pdfpagelayout=SinglePage,       
    bookmarks=false,
    bookmarksnumbered=true,
    bookmarksopen=false,             
    pdftex=true,
    unicode,
    colorlinks=true,       
    linkcolor=blue,          
    citecolor=blue,  
    filecolor=magenta,      
    urlcolor=cyan           
}

\usepackage[all]{hypcap}

\makeatletter
\let\c@proposition\c@theorem
\let\c@corollary\c@theorem
\let\c@lemma\c@theorem
\let\c@definition\c@theorem
\let\c@example\c@theorem
\makeatother

\spnewtheorem{questionW}{Question}{\it}{}
\spnewtheorem{resultW}{Result}{\bfseries}{\itshape}

\let\orgautoref\autoref          	
\providecommand{\Autoref}[1]{
    \def\equationautorefname{Equation}
    \def\figureautorefname{Figure}%
	\def\subfigureautorefname{Figure}%
    \def\lemmaautorefname{Lemma}%
    \def\conjectureautorefname{Conjecture}%
    \def\remarkautorefname{Remark}%
    \def\propositionautorefname{Proposition}%
    \def\corollaryautorefname{Corollary}%
    \def\definitionautorefname{Definition}%
    \def\sectionautorefname{Section}%
    \def\subsectionautorefname{Section}%
    \def\subsubsectionautorefname{Section}%
    \def\exampleautorefname{Example}%
    \def\lessonautorefname{Result}%
    \def\resultautorefname{Result}%
    \orgautoref{#1}%
}
\renewcommand{\autoref}[1]{
    \def\figureautorefname{Fig.\!}%
    \def\subfigureautorefname{Fig.\!}
    \def\lemmaautorefname{Lemma}%
    \def\conjectureautorefname{Conjecture}%
    \def\remarkautorefname{Remark}%
    \def\propositionautorefname{Prop.}%
    \def\algorithmautorefname{Algorithm}%
    \def\sectionautorefname{Section}%
    \def\subsectionautorefname{Section}%
    \def\lessonautorefname{Result}%
    \def\resultautorefname{Result}%
    \orgautoref{#1}%
}

\newcommand{\specificref}[2]{\hyperref[#2]{#1~\ref*{#2}}}	
\newcommand{\specialref}[3]{\hyperref[#2]{#1~\ref*{#2}~#3}}	

\newcommand{\figref}[2]{\hyperref[#1]{Fig.~\ref*{#1}#2}}	
\newcommand{\Figref}[2]{\hyperref[#1]{Figure~\ref*{#1}#2}}

\usepackage{tcolorbox}		
\tcbuselibrary{breakable,skins}		
\tcbset{
	enhanced jigsaw,		
	colback=gray!20,
	colframe=gray!40,
	arc=0mm,
	boxrule=1pt,
	left skip=-1mm,
	right skip=-1mm,	
	left=0mm,
	topsep at break=1mm,			
	right=0mm,
	top=-2mm,
	bottom=0mm,		
	breakable,		
	parbox = false		
	}

\tcolorboxenvironment{example}{}

\usepackage[algo2e, boxed, vlined]{algorithm2e}		
\DontPrintSemicolon     		
\SetNlSty{}{}{}					
\SetAlgoCaptionSeparator{\,}	
\SetAlgoSkip{smallskip}

\SetAlgoCaptionLayout{myCap}					

\SetInd{1.2mm}{1.2mm}	

\SetKwFor{ForAll}{forall}{do}{endfch}	
\LinesNumbered

\renewcommand{\(}{\left(}                   
\renewcommand{\)}{\right)}

\newcommand{\PP}[1]{{{\mathbb{P}}} [ #1 ]}       
\newcommand{\PPP}{{{\mathbb{P}}}}       
\newcommand{\ADom}{{\mathit{ADom}}}

\newcommand{\R}{\mathbb{R}}             
\renewcommand{\epsilon}{\varepsilon}    

\renewcommand{\[}{[\![}

\newcommand{\smallsection}[1]{\vspace{5mm}\noindent\textit{#1.}}	
\newcommand{\introparagraph}[1]{\textit{#1.}}        

\definecolor{gray}{rgb}{0.5,0.5,0.5}
\definecolor{niceblue}{rgb}{.8,.85,1}
\newcommand{\graycell}{\cellcolor[gray]{0.9}}

\newcommand{\sql}[1]{\textup{\textsf{\small#1}}}

\newcommand{\datarule}{{\,:\!\!-\,}}			
\newcommand{\Var}{\textup{\texttt{Var}}}		
\newcommand{\HVar}{\textup{\texttt{HVar}}} 		
\newcommand{\JVar}{\textup{\texttt{JVar}}}

\newcommand{\SepVar}{\textup{\texttt{SepVar}}}  	
\newcommand{\SepPVar}{\textup{\texttt{PSepVar}}}

\newcommand{\EVar}{\textup{\texttt{EVar}}}  	
\newcommand{\Lineage}{\textup{\texttt{Lin}}}  	
\newcommand{\score}{\textit{score}}  			
  
\newcommand{\MinCuts}{\textup{\texttt{MinCuts}}}  
\newcommand{\MinPCuts}{\textup{\texttt{MinPCuts}}}

\newcommand{\PDMS}{\textsc{PDB}\xspace}
\newcommand{\PDMSs}{\textsc{PDB}s\xspace}

\newcommand{\at}{\textup{\texttt{at}}}  		

\newcommand{\joinp}[2]{\Join^p\!\! \big[#2\big]} 
		
\newcommand{\joind}[2]{\Join_{#1}\!\! \big[#2\big]}			
\newcommand{\projp}[1]{\pi^p_{\!-\!#1}}						
\newcommand{\projpd}[1]{\pi^p_{#1}}						 	
\newcommand{\projd}[1]{\pi_{\!-\!#1}}						
							
\newcommand{\minp}[1]{\min \! \left[ #1 \right]}

\newcommand{\avg}{\textrm{avg}}			

\renewcommand{\vec}[1]{\boldsymbol{\mathbf{#1}}}

\usepackage{stackengine,scalerel}
\newcommand{\dotcirc}{\mathop{\ThisStyle{
  \ensurestackMath{\stackinset{c}{0\LMpt}{c}{0\LMpt}{\SavedStyle\cdot}{\SavedStyle\circ}}}}}

\newcommand\xqed[1]{%
  \leavevmode\unskip\penalty9999 \hbox{}\nobreak\hfill
  \quad\hbox{#1}}
\newcommand\markend{\xqed{$\blacksquare$}} 		
\newif\ifqed 
\def\qedhere{\tag*{\hbox{\rlap{$\sqcap$}$\sqcup$}}\global\qedfalse}

\usepackage{scalefnt}	
\makeatletter       	
    \def\url@leostyle{
        \@ifundefined{selectfont}{\def\UrlFont{\sf}}{\def\UrlFont{\scalefont{0.9}\ttfamily}}}
\makeatother
\newcommand{\set}[1]{\{#1\}}

\newcommand{\makeop}[2]                         
  {\ifx#2.\def\next##1{}\else\escapechar=-1     
  \def\next##1{\escapechar=92\def#2{#1}}        
  \expandafter\next\expandafter{\string#2}      
  \let\next\makeop\fi\next{#1}}                 

\makeop{{\cal #1}}                              
  \W  .             
\def \var(#1){{\bf #1}}

\newcommand{\silentreminder}[1]{}

\newcommand{\eat}[1]{}

\newcommand{\define}{\coloneqq}

\newcommand{\bibpath}{propagation} 		
\graphicspath{{figs/}}

\makeatletter
\setbool{@fleqn}{false}
\makeatother

\journalname{VLDBJ}

\begin{document}

\title{Dissociation and Propagation for Approximate Lifted Inference with 
Standard Relational Database Management Systems}

\author{Wolfgang Gatterbauer \and Dan Suciu}

\institute{
	Wolfgang Gatterbauer
	\at
	Tepper School of Business\\
	Carnegie Mellon University\\
	\email{gatt@cmu.edu}           
\and
	Dan Suciu
	\at
	Department of Computer Science \& Engineering\\
	University of Washington
}

\date{}	

\maketitle

\setlength{\emergencystretch}{5pt}		

\begin{abstract}
Probabilistic inference over large data sets is a challenging data management problem since exact inference is generally \#P-hard and is most often solved approximately with sampling-based methods today.
This paper proposes an alternative approach for approximate evaluation of conjunctive queries with standard relational databases: In our approach, every query is evaluated entirely in the database engine by \emph{evaluating a fixed number of query plans}, each providing an upper bound on the true probability, then taking their minimum. 
We provide an algorithm that takes into account important schema information to enumerate only the minimal necessary plans among all possible plans.
Importantly, this algorithm 
is a \emph{strict generalization of all known PTIME self-join-free conjunctive queries}: A query is in PTIME if and only if our algorithm returns one single plan.
Furthermore, our approach is a generalization of a family of efficient ranking methods from graphs to hypergraphs.
We also adapt three relational query optimization techniques to evaluate all necessary plans 
very fast.
We give a detailed experimental evaluation of our approach and, in the process,
provide a new way of thinking about the value of probabilistic methods over non-probabilistic methods for ranking query answers.
We also note that the techniques developed in this paper apply immediately to \emph{lifted inference} from statistical relational models since lifted inference corresponds to PTIME plans in probabilistic databases.

\end{abstract}

\newpage
\section{Introduction}\label{sec:intro}

Probabilistic inference over large data sets is becoming a central
data management problem. Recent large knowledge bases, such as
Yago~\cite{DBLP:journals/ai/HoffartSBW13},
Nell~\cite{DBLP:conf/aaai/CarlsonBKSHM10}, DeepDive~\cite{deepdive},
or Google's Knowledge Vault~\cite{knoweldge-vault-kdd-2014}, have
millions to billions of uncertain tuples.  Data sets with missing
values are often ``completed'' using inference in graphical
models~\cite{DBLP:conf/sigmod/ChenW14,DBLP:journals/jiis/RaghunathanDK14,DBLP:conf/icde/StoyanovichDMT11}
or sophisticated low rank matrix factorization
techniques~\cite{bouchard-uai2014,DBLP:conf/kdd/SinghG08} that
ultimately result in a large probabilistic database. 
Data sets that result from crowdsourcing~\cite{DBLP:conf/dasfaa/AmarilliAM14}
or that are inferred from unstructured information~\cite{DBLP:journals/tois/Cohen00}
are also
uncertain,
and
probabilistic databases have been applied to
bootstrapping over samples of data~\cite{DBLP:conf/sigmod/ZengGMZ14}.

However, probabilistic inference is known to be \#P-hard in the size of the
database, even for some very simple
queries~\cite{DBLP:journals/vldb/DalviS07}.  Today's state of the art
inference engines use either sampling-based methods or are based on
some variant of the DPLL algorithm for Weighted Model Counting~\cite{DBLP:journals/jacm/DavisP60}.  
For example,
Tuffy~\cite{DBLP:journals/pvldb/NiuRDS11}, a popular implementation of
Markov Logic Networks (MLN) over relational databases, uses Markov Chain Monte Carlo methods (MCMC).
Gibbs sampling can be significantly improved by adapting some
classical relational optimization
techniques~\cite{DBLP:conf/sigmod/ZhangR13}.  For another example,
MayBMS~\cite{DBLP:conf/icde/AntovaKO07a} and its successor
Sprout~\cite{OlteanuHK2010:ICDE}
use query plans to guide a
DPLL-based algorithm for Weigh\-ted Model
Counting~\cite{DBLP:series/faia/GomesSS09}.  While both approaches
deploy some advanced relational optimization techniques, at their core
they are based on general purpose probabilistic inference techniques,
which either run in exponential time (DPLL-based algorithms have been
proven recently to take exponential time even for queries computable
in polynomial time~\cite{DBLP:conf/icdt/BeameLRS14}), 
or require many
iterations until convergence.

In this paper, we propose a different approach to query evaluation with
probabilistic databases (PDBs).  In our approach, \emph{every query is
evaluated entirely in the database engine}.  Probability computation is
done at query time, using simple arithmetic operations and aggregates.
Thus, probabilistic inference is entirely reduced to a standard query
evaluation problem with aggregates. There are no iterations and no
exponential blowups. All benefits of relational engines (such as
cost-based optimizations, multi-core query processing, shared-nothing
parallelization) are directly available to queries over
probabilistic databases.  

To achieve this, we compute approximate rather than exact probabilities, with a one-sided
guarantee: The probabilities are guaranteed to be upper bounds to the
true probabilities, which we show is \emph{sufficient to rank the top query
answers with high precision}.
Our approach consists of approximating the true \#P-hard query probability by
evaluating a fixed number of PTIME queries (the number depends on the query), each providing an upper
bound on the true probability, then taking their minimum.  
Another way to put this is that we replace the standard semantics based on
\emph{reliability}, with a related but much more efficient semantics based on
\emph{propagation}, and which is guaranteed to be an upper bound on reliability. We explain this alternative semantics next.

The semantics of a query over a \PDMS is based on the possible world semantics, which is equivalent to 
``\emph{query reliability}''~\cite{DBLP:conf/pods/GradelGH98}.
Among its roots are {\em network reliability} \cite{Colbourn:1987fk} which is defined as
the probability that a source node $s$ remains connected to a target node $t$
in a directed graph if edges fail independently with known probabilities. However, computing network reliability is \#P-hard. Hence, many applications where an exact probabilistic semantics is not critical (especially for ranking alternative answers)  have replaced network reliability with another semantics based on a ``\emph{propagation scheme}.'' We illustrate with an example.

\begin{figure}[t]
    \centering
	\subfloat[]
		{\begin{minipage}[b]{39mm}
		\includegraphics[scale=0.72]{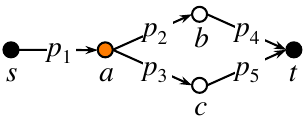}
		\label{Fig_IntroductionExample_a}
	\end{minipage}}
	\hspace{0mm}
	\subfloat[]
		{{\begin{minipage}[b]{43mm}
		\vspace{-3mm}
		\begin{minipage}[t]{42mm}
		\small
		$\quad\quad q \datarule R(s, x), S(x,y), T(y,t)$ \\[-2mm]
		
		\hspace{7mm}
		\setlength{\tabcolsep}{0.4mm}
			\mbox{
					\begin{tabular}[t]{ >{$}c<{$} | >{$}c<{$} >{$}c<{$} }
		 			R	& C		& A	\\
					\hline
					p_1	& s		& a			
					\end{tabular}			
			}
			\hspace{-2mm}
			\mbox{
					\begin{tabular}[t]{ >{$}c<{$} | >{$}c<{$} >{$}c<{$} >{$}c<{$} >{$}c<{$}}
		 			S	& A		& B\\
					\hline
					p_2	& a		& b		\\
					p_3	& a		& c		
					\end{tabular}
			}
			\hspace{-2mm}
			\mbox{
					\begin{tabular}[t]{ >{$}c<{$} | >{$}c<{$} >{$}c<{$}}
		 			T	& B		& C	\\
					\hline
					p_4	& b		& t\\	
					p_5	& c		& t
					\end{tabular}			
			}
		\end{minipage}	
		\vspace{-1mm}
		\end{minipage}	
		}\label{Fig_IntroductionExample_a_table}}
	\hspace{0mm}
	\subfloat[]
		{\begin{minipage}[b]{39mm}
		\includegraphics[scale=0.72]{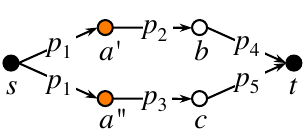}
		\label{Fig_IntroductionExample_b}
		\end{minipage}}
	\hspace{0mm}
	\subfloat[]{
			{\begin{minipage}[b]{42.0mm}
			\vspace{-3mm}
			\begin{minipage}[t]{41.9mm}
			\small
			\hspace{2mm}
			$q^{\Delta_1} \datarule R^y(s, x, y), S(x,y), T(y,t)$ \\[-1.8mm]

			\hspace{4mm}
			\setlength{\tabcolsep}{0.4mm}
				\mbox{
						\begin{tabular}[b]{ >{$}c<{$} | >{$}c<{$}	>{$}c<{$}	>{$}c<{$}   }
			 			R^y	& C		& A		& B		\\
						\hline
						p_1	& s		& a		& b	\\
						p_1	& s		& a		& c					
						\end{tabular}			
				}
				\hspace{-2mm}
				\mbox{
						\begin{tabular}[b]{ >{$}c<{$} | >{$}c<{$} >{$}c<{$} >{$}c<{$} >{$}c<{$}}
			 			S	& A		& B\\
						\hline
						p_2	& a		& b		\\
						p_3	& a		& c		
						\end{tabular}
				}
				\hspace{-2mm}
				\mbox{
					\begin{tabular}[b]{ >{$}c<{$} | >{$}c<{$} >{$}c<{$}}
					T	& B		& C	\\
					\hline
					p_4	& b		& t\\	
					p_5	& c		& t
					\end{tabular}			
				}
			\end{minipage}
			\vspace{-1mm}
			\end{minipage}
			}\label{Fig_IntroductionExample_b_table}}		
	\caption{\specificref{Example}{ex:1}. The \emph{propagation score} $\rho(t)$ in graph (a) corresponds to the \emph{reliability score} $r(t)$ in graph (c) with node $a$ dissociated into two. 
	(b,d): Corresponding chain queries with respective databases.}\label{Fig_IntroductionExample}
\end{figure}

\begin{example}[Propagation in $k$-partite digraphs]\label{ex:1}
Consider the 4-partite graph in \autoref{Fig_IntroductionExample_a}.
Intuitively, let's call a node $x$ ``active'' if there exists a directed path from the source node to $x$.  
Then, the ``\emph{reliability score}'' $r(x)$ of a node $x$  is the probability that $x$ is active if every edge $e$ is included in the graph independently with probability $p_e$. 
The score of interest is the reliability of a target node $t$:
$r(t) 
= p_1 (p_2 p_4 \otimes p_3 p_5)
= p_1 (1\!-\!(1\!-\!p_2 p_4)(1\!-\!p_3 p_5))$
where ``$\otimes$'' stands for the ``\emph{independent-or}'' 
in infix or prefix notation,
which combines probabilities as if calculating the disjunction between independent events:
$\bigotimes_{i} p_i \define
 1 - \prod_{i} (1-p_i)$.
While reliability can be computed efficiently for series-parallel
graphs as the one in \autoref{Fig_IntroductionExample_a}, it is 
\#P-hard in general, even on 4-partite networks~\cite{Colbourn:1987fk}.  The probability of a
query over a \PDMS corresponds precisely to network reliability.  
For
example, in the case of a 4-partite graph, reliability is
given by the probability of the 3-chain query $q \datarule
R(s,x),S(x,y),T(y,t)$ over the \PDMS shown in
\autoref{Fig_IntroductionExample_a_table} 
(here $s$ and $t$ stand for constants).
Notice that the reliability of a node is a combinatorial or ``\emph{global property}'' of the entire graph:
it is defined as a weighted average over all possible worlds and can generally not be calculated easily.

In contrast, the {\em propagation score} $\rho(x)$ of a node $x$ is a value that recursively depends on the scores of its neighbors and the probabilities of the connecting edges: 
\begin{equation}\label{eq:dissociationFormula}
	\rho(x) \leftarrow \bigotimes_{e} p_{e} \cdot \rho(u_e)
\end{equation}
where $e$ ranges over all incoming edges $(u_e,x)$.
By
  definition, $\rho(s) = 1$.  
  In \autoref{Fig_IntroductionExample_a}, the propagation score of
  the target node $t$ is
$\rho(t) 
= p_4 \rho(b) \otimes p_5 \rho(c) 
= p_1 p_2 p_4 \otimes p_1 p_3 p_5 = 1 \!- (1 \!- p_1 p_2 p_4)(1 \!- p_1 p_3 p_5)$. 
Notice that the propagation score of a node is a recursive or ``\emph{local property}'' 
since it can be calculated from the scores of its neighbors.
\markend
\end{example}

\vspace{-1mm}

With ``propagation'', we refer to a family of techniques for calculating the relative importance of nodes in networks with \emph{iterative models of computation}: 
``relevance'' is \emph{propagated} across edges from node to node while ignoring past dependencies (see~\autoref{Fig_Propagation_PageRank}).
Thus, unlike reliability, propagation scores can always be computed efficiently, even on
very large graphs. 
Variants of propagation have been successfully used in a range of
applications for calculating relevance where exact probabilities are not
necessary. Examples include similarity ranking of proteins~\cite{Weston:2004fk},
integrating and ranking uncertain scientific data~\cite{DetwilerGLST2009:ICDE},
models of human comprehension~\cite{Quillian:68a},
activation in feedforward networks~\cite{Rumelhart:1986:LIR:104279.104293},
search in associative networks~\cite{DBLP:journals/air/Crestani97},
trust propagation~\cite{DBLP:conf/www/GuhaKRT04} and 
influence propagation~\cite{DBLP:conf/wsdm/GoyalBL10} in social networks,
keyword search in databases~\cite{DBLP:conf/icde/BhalotiaHNCS02}, 
the noisy-or gate \cite[Sect.~4.3.2]{Pearl:1988pb},
computing web page reputation with PageRank~\cite{DBLP:journals/cn/BrinP98},
belief propagation in graphical models~\cite{Pearl:1988pb},
linearized belief propagation for node labeling~\cite{DBLP:journals/pvldb/GatterbauerGKF15},
or finding true facts from a large amount of conflicting information~\cite{DBLP:journals/tkde/YinHY08}.\footnote{Also see \cite{PasternackR2010:FactFinder} for a related discussion of fact finding algorithms, in which the approach of \cite{DBLP:journals/tkde/YinHY08} and its use of the iterative propagation  \autoref{eq:dissociationFormula} is referred to as ``\emph{pseudoprobabilistic}''.}
Note that the resulting relevance scores commonly do not have an exact probabilistic semantics, 
and may be used as a heuristics instead. 
For example, the {PageRank} of a web page does not have to be smaller than $1$ 
(see \autoref{Fig_Propagation_PageRank} for a comparison of the update equations).
However, these variants have in common that the score of a node is recursively defined
only \emph {in terms of the scores of its neighbors}, and not in terms of
the entire topology of the graph.

\begin{figure}[t]
\centering
\includegraphics[scale=0.41]{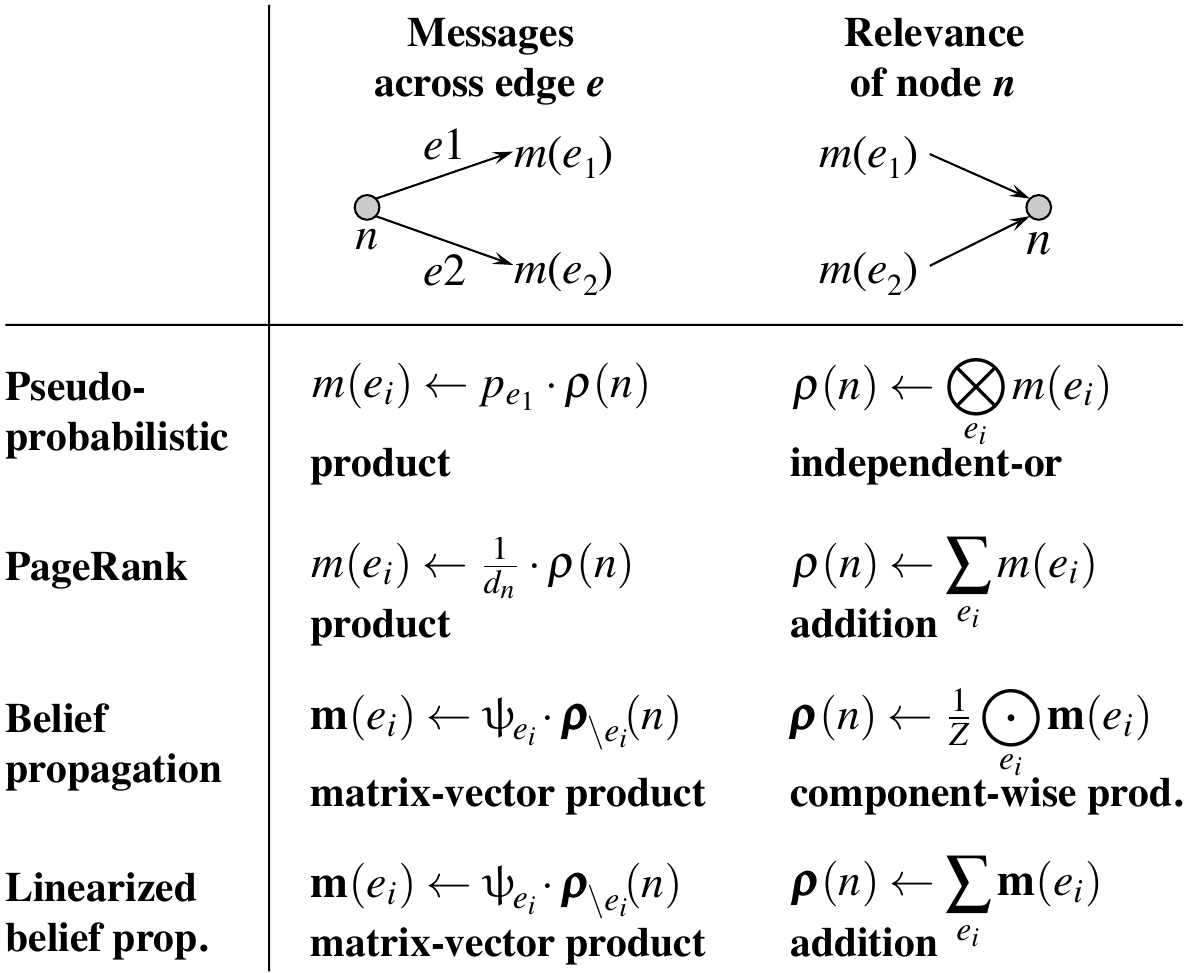}
\caption{``\emph{Relevance propagation}'' in graphs works by iteratively calculating messages $m(e)$ across edges $e$ and relevance scores $\rho(n)$ of nodes $n$. 	
The propagation method we consider is \emph{pseudoprobabilistic} in that the two operators are ``\emph{independent-and}'' or product ($\cdot$), and ``\emph{independent-or}'' ($\otimes$).
PageRank and related methods from semi-supervised learning 
replace the probability $p_e$ of an edge with a weight (here $d_n$ stands for the out-degree of node $n$) and the independent-or with addition or sum ($\sum$).
Belief Propagation propagates not just one message across an edge but a vector $\vec m(e)$ of messages, 
scales this message vector with a matrix $\bm{\uppsi}_e$ (also called ``edge potential''),
and replaces the independent-or with a component-wise product ($\odot$), followed by a normalization 
(here $Z$ stands for a normalizer).
Linearized Belief Propagation uses again addition as second operator and requires no normalization.
Intuitively, the method developed in this paper generalizes pseudoprobabilistic relevance propagation to hypergraphs.
}\label{Fig_Propagation_PageRank}
\end{figure}

While \autoref{ex:1} shows how the propagation score can be defined on graphs, queries are not represented by graphs but \emph{hypergraphs}, in general. To the best of our knowledge, no definition of a propagation score on hypergraphs exists, and it is not obvious how to define such a score.
Also, the propagation score between two nodes depends on the directionality of the
graph, which can be best illustrated with our example of $k$-partite
graphs: In \autoref{Fig_IntroductionExample_a} the propagation score
from $s$ to $t$ is different from the one from $t$ to $s$ (in fact, the
latter coincides with the reliability score).  It is not immediately clear what this
directionality corresponds to for a relational query whose lineage defines a hypergraph.

With this paper, we introduce a propagation score for queries over
\PDMSs, describe the connection to the reliability score, 
and give a method to efficiently compute the propagation score for \emph{any self-join-free conjunctive query}
with a standard relational database engine.
While the propagation score differs from the
reliability score, we prove several properties showing that it is a
reasonable substitute: 
(\textit{i})~propagation and reliability are guaranteed to coincide for all known PTIME queries: our score are thus \emph{strict generalization} of efficient evaluation methods from PTIME to \#P-hard queries;
(\textit{ii}) propagation is in PTIME and can be evaluated with a \emph{standard relational DBMS} without any changes to the underlying relational query engine;
(\textit{iii}) propagation is inspired by the above listed number of successful ranking schemes on graphs: yet our score extends the underlying idea of propagation on graphs to \emph{propagation on hypergraphs};
(\textit{iv}) the propagation score is always an upper bound to the reliability score: it can thus be applied as efficient filter; and
(\textit{v}) the ranking given by the propagation score is very close to the ranking given by the reliability score in our experimental validation.

\begin{example}[\autoref{ex:1} cont.]	
We have seen that the propagation score differs from the reliability
score on the DAG (Directed Acyclic Graph) in \autoref{Fig_IntroductionExample_a}.  By
inspecting the expressions of the two scores, one can see that they
differ in the way they treat $p_1$: reliability treats it as a single
event, while propagation treats it as two independent events.  In fact, the propagation
score is precisely the reliability score of the DAG in
\autoref{Fig_IntroductionExample_b}, which has two copies of $p_1$.  We call this DAG the
``\emph{dissociation}'' of the DAG in
\autoref{Fig_IntroductionExample_a}.  At the level of the database, 
dissociation can be obtained by adding a new attribute $B$
to the first relation $R$ (\autoref{Fig_IntroductionExample_b_table}).
The dissociated query is $q^{\Delta_1} \!\datarule \! R^y(s,x, y), S(x,y),T(y,t)$, 
where the exponent $^y$ in $R^y$ indicates the new attribute, and
its probability is indeed the same as the
propagation score for the graph in
\autoref{Fig_IntroductionExample_a}.  The important observation here
is that, while the evaluation problem for $q$ is \#P-hard in general, 
the query $q^{\Delta_1}$ is ``hierarchical''~\cite{DBLP:journals/vldb/DalviS07} and
can therefore be computed efficiently.
A query $q$ usually has more than one dissociation: $q$ has a second
dissociation $q^{\Delta_2} \,\,\datarule\, R(s,x),S(x,y), T^x(x,y,t)$ obtained by adding
the attribute $A$ to $T$ (not shown in the figure).  Its probability
corresponds to the propagation score from $t$ to $s$, i.e.\ from
right to left.  And $q^{\Delta_3} \datarule
R^y(s,x, y),S(x,y), T^x( x,y,t)$ is a third dissociation. We prove that each dissociation
step can only increase the probability (e.g., $r(q) \leq r(q^{\Delta_1}) \leq
r(q^{\Delta_3})$).
 We define the
propagation score of $q$ as the smallest probability of these three
dissociations.  The database system has to compute $r(q^{\Delta_1})$ and
$r(q^{\Delta_2})$ and return the smallest score: on the graph in
\autoref{Fig_IntroductionExample_a}, this is $r(q^{\Delta_2})$, since $r(q) =
r(q^{\Delta_2})$.
\markend
\end{example}

\introparagraph{Contributions and outline} 
(1) We derive ``\emph{query dissociation}'' as a generalization of relevance propagation from graphs to hypergraphs and
define the propagation score for
any self-join-free conjunctive query in terms of 
dissociations~(\autoref{sec:dissociation}).  A
query dissociation is a rewriting of both the data and the query.  On the
data, a dissociation is obtained by making multiple, independent
copies of some of the tuples in the database. Technically, this is
achieved by extending the relational schema with additional
attributes.  On a query, a dissociation extends atoms with
additional variables.  We prove that a dissociation can only increase the probability of a query,
and define \emph{the propagation score of a query as the minimum
reliability of all dissociated queries that are ``hierarchical''}.
This is justified by the fact that, in a $k$-partite graph, the propagation score is precisely the probability of one dissociated hierarchical query.
Thus, in our definition, choosing a direction for the network in order to
define the propagation score corresponds to choosing a particular dissociation
that makes the query hierarchical.

(2) 
We show how the propagation score can be evaluated with the help of a query-dependent number of query plans (\autoref{sec:plans}).
We achieve this by establishing a one-to-one correspondence between hierarchical dissociations and traditional query plans and 
showing that \emph{every query plan computes a probability that is an upper bound of query reliability}.
Moreover, we describe a 
natural partial order on the probabilities of query plans
Thus \emph{every self-join-free conjunctive query can be approximated by a fixed number of query plans}, 
and 
it suffices to iterate over all minimal plans, compute their probabilities, then take the minimum.
We give an intuitive 
system R-style algorithm~\cite{DBLP:conf/sigmod/SelingerACLP79} that enumerates all minimal plans for a given query $q$.

(3) 
We generalize the algorithm to take into consideration schema knowledge on deterministic relations and functional dependencies~(\autoref{sec:optimizationsWithSchema}).
In particular, we give a \emph{unified treatment and generalization of all previously known PTIME self-join free conjunctive queries}, i.e.\ those that can be evaluated with a query plan in polynomial time in the size of the database,
and show that our approach naturally generalizes all known PTIME queries: 
for every query that is PTIME (whether due to 
key constraints, or the presence of deterministic tables), reliability and propagation scores always coincide; for every query that is \#P-hard, our approach still returns a unique, well-defined score in polynomial time. 
 
(4) We give a set of targeted \emph{multi-query optimization techniques} that considerably speed up the time needed to evaluate the propagation score~(\autoref{sec:optimizations}). 
Evaluating some queries may require a large number of plans (e.g., an
8-chain query requires 429 plans).
Evaluating all plans sequentially 
would still be
prohibitively expensive.  
Instead, we tailor three relational query optimization techniques to dissociation:
($i$) combining all minimal plans into \emph{one single query}, 
($ii$) reusing \emph{common subexpressions} with views, and 
($iii$) performing \emph{deterministic semi-join reductions}.
 
(5)
We conduct a set of very extensive experiments in which we compare the quality of ranking and scalability of various alternative methods (probabilistic and not) against exact probabilistic inference on TPC-H data~\cite{tpc-h}. 
We devise a setup that measures the additional benefit of probabilistic inference for ranking over alternative methods, showing that our technique has high precision for ranking query answers based on their output
probabilities. 
We also show that, with all our optimizations enabled, computing hard queries over probabilistic databases incurs only a modest penalty over
computing the same query on a deterministic database: For example, the
8-chain query (with 429 query plans) runs only a factor of $<10$ slower than on a deterministic
database.

\introparagraph{Prior publications}
In recent work~\cite{DBLP:journals/tods/GatterbauerS14}, we apply the idea of dissociation to both upper \emph{and} lower bound  the probability of Boolean functions, but discuss the connection to query evaluation only in passing.
Parts of \specificref{Section}{sec:dissociation} and \autoref{sec:plans} are based on a workshop paper~\cite{GatterbauerJS2010:MUD}. The remainder is based upon \cite{DBLP:journals/pvldb/GatterbauerS15}. We added the connection to propagation on graphs, more detailed experiments, slightly changed the formalisms, and included extensive illustrating examples throughout.
Due to space restrictions, some of our proofs had to be included in an online appendix on ArXiv~\cite{arxivDissociation:2013}.

\section{Technical background}\label{sec:background}

\subsection{Probabilistic databases and self-join-free conjunctive queries}\label{sec:backgroundPDBs}
A \emph{tuple-independent probabilistic database} (TI-PDB) is a database $D$ 
plus a function $p(t) \in [0,1]$ 
associating an independent probability to
each tuple $t \in D$.  
We fix a relational vocabulary $\sigma = (R_1, \ldots, R_m)$ and denote with $D$ the database, i.e.\ the collection of tuples and their probabilities. 
A \emph{possible world} is then a subset of $D$ generated by independently
including each tuple $t$ in the world with probability $p(t)$.
We use bold notation (e.g., $\vec x$) to denote both sets or tuples.
A \emph{self-join-free
  conjunctive query} (sj-free CQ) is a first-order formula 
  $q(\vec z) = \exists
x_1\ldots \exists x_k.(a_1 \wedge \ldots \wedge a_m)$ where each atom $a_i$
represents a relation $R_i(\vec x_i)$, the
variables $x_1, \ldots, x_k$ are called {\em existential variables},
and $\vec z$ are called the \emph{head variables} (or free
  variables).\footnote{W.l.o.g.\ we assume $\vec x_i$ to be a tuple of only variables and don't write the constants.
 Selections can always be directly pushed into the database before executing the query.}

The term ``self-join-free'' means that the atoms refer
to distinct relational symbols. We assume therefore w.l.o.g.\ that
every relational symbol $R_1, \ldots, R_m$ occurs exactly once in the
query.  Unless otherwise stated, a ``query'' in this paper always denotes
a sj-free CQ.
As usual, we abbreviate a query
by $q(\vec z) \datarule a_1, \ldots, a_m$,
and write $\HVar(q) = \vec z$, $\EVar(q) = \set{x_1, \ldots, x_k}$
and $\Var(q) = \HVar(q) \cup \EVar(q)$ for the set of head variables,
existential variables, and all variables of $q$. If $\HVar(q) =
\emptyset$, then $q$ is called a \emph{Boolean} query and $\EVar(q) = \Var(q)$.  We also write
$\Var(a_i)$ for the variables in atom $a_i$ and $\at(x_j)$ for the set of atoms that
  contain variable $x_j$. 
The \emph{active domain} of a variable $x_j$ is denoted $\ADom_{x_j}$,\footnote{Defined formally as
  $\ADom_{x_j} = \bigcup_{i: x_j \in \Var(R_i)} \pi_{x_j}(R_i)$.}
 and the active domain of the entire database is $\ADom =
\bigcup_j \ADom_{x_j}$.  The focus of probabilistic query evaluation is
to compute $\PP{q}$, i.e.\ the probability that the query is true
in a randomly chosen world. We will refer to this probability as the ``\emph{query reliability}''
$r(q)$~\cite{DBLP:conf/pods/GradelGH98}.

It is known that the data complexity~\cite{DBLP:conf/stoc/Vardi82} of
any query $q$ is either in PTIME or \#P-hard~\cite{DBLP:journals/jacm/DalviS12}.
The PTIME queries are also called ``\emph{safe queries}'' and, for the case of sj-free CQs, are characterized precisely by a syntactic property called  \emph{hierarchical queries}~\cite{DBLP:journals/vldb/DalviS07}.
We briefly review these results:

\begin{definition}[Hierarchical query]\label{def:hierq}
A query $q$ is called \emph{hierarchical} iff for any 
two existential variables $x, y \in \EVar(q)$,
one of the following three conditions
holds: 
	$\at(x) \subseteq \at(y)$, 
  	$\at(x) \supseteq \at(y)$, or
	$\at(x) \cap \at(y) = \emptyset$.	
\end{definition}

\noindent
For example, the query $q_1\datarule R(x,y), S(y,z), T(y,z,u)$ is
hierarchical, while $q_2 \datarule R(x,y), S(y,z), T(z,u)$ is not as
neither of the three conditions holds for the variables $y$ and $z$.  

\begin{theorem}[Hierarchy dichotomy~\cite{DBLP:journals/vldb/DalviS07}]\label{th:dichotomy}
  If $q$ is hierarchical, then
  $\PP{q}$ can be computed in PTIME in the size of $D$.
  Otherwise, computing $\PP{q}$ is \#P-hard
  in the size of $D$.
\end{theorem}

We next give an equivalent, recursive characterization of hierarchical
queries, for which we need a few definitions.  
We write $\SepVar(q)$ for the set of existential variables that appear in every atom
(called ``\emph{separator variables}'').
A \emph{connected component} of $q$ (or short, ``\emph{query component}'') is a subset of atoms that are connected via existential variables.
A query $q$ is \emph{disconnected} if its atoms can be partitioned into two
non-empty sets that do not share any existential variables 
(e.g., $q \datarule R(x,y),S(z,u),T(u,v)$ is disconnected and has
two query components: ``$R(x,y)$'' and ``$S(z,u),T(u,v)$'').
For every set of
variables $\vec x$, denote $q - {\vec x}$ the query obtained by
removing all variables $\vec x$ and decreasing the arities of the
relational symbols that contain variables from $\vec x$.
Any query can become disconnected by removing a set of variables.

\begin{lemma}[Hierarchical queries] \label{lemma:hierarchical}
A query  $q$ is ``hierarchical'' iff either: (1) $q$ has a single atom; (2) $q$ has $k
  \geq 2$ query components all of which are hierarchical; or (3)
  $q$ has a separator variable $x$, and $q-\{x\}$ is hierarchical.
\end{lemma}

\noindent
Every hierarchical query can be computed in PTIME, but non-hierarchical queries are \#P-hard, in general.\footnote{Non-hierarchical queries can be in PTIME when considering functional dependencies or deterministic tables \cite{DBLP:journals/vldb/DalviS07,DBLP:conf/icde/OlteanuHK09}  (see \autoref{sec:optimizationsWithSchema}).}

\subsection{Probabilistic query plans}
Unless otherwise stated, a ``query plan'' in this paper always denotes a \emph{probabilistic query plan}.

\begin{definition}[Query plans]\label{def:queryPlans}
  A \emph{query plan} $P$ is given by the grammar
$   P 	::=\, R_i(\vec x)
   \,\,\,|\,\, 	\projpd{\vec x} P
   \,\,\,|\,\,\! 	\joinp{}{P_1, \ldots, P_k} 
$
 where $R_i(\vec x)$ is a relational atom containing the variables
 $\vec x$, $\projpd{\vec x}$ is the 
 \emph{probabilistic project operator with
 duplicate elimination} (or short ``projection''), 
 and $\joinp{}{\ldots}$ is the \emph{probabilistic natural join} (or short ``join'') in prefix notation, 
 which we  allow to be $k$-ary ($k \geq 2$).  We require that joins and
 projections alternate in a plan and do not distinguish between join
 orders.
\markend
\end{definition}

We write $\Var(P)$ for all variables in a plan $P$
and $\HVar(P)$ for its \emph{head variables},
which are recursively defined as follows:
(1) if $P = R_i(\vec x)$, then $\HVar(P) = \vec x$;
(2) if $P = \projpd{\vec x}(P')$, then $\HVar = \vec x$; and
(3) if $P = \joinp{}{P_1, \ldots, P_k}$,
then $\HVar(P) = \bigcup_{i=1}^k \HVar(P_i)$.
The \emph{existential variables} $\EVar(P)$ are then defined as $\Var(P) - \HVar(P)$.

Every plan $P$ represents a query $q_P$
defined by taking all atoms mentioned in $P$ as the body and setting
$\HVar(q_P) = \HVar(P)$.
A plan is called Boolean if $\HVar(P) = \emptyset$.
We assume the usual sanity conditions on plans to be satisfied: for a
projection $\projpd{\vec x}P$ we assume $\vec x \subseteq
\HVar(P)$, and each variable $y$ is projected away at most once in a
plan, i.e.\ there exists at most one operator $\projpd{\vec x}P$
s.t. $y \in \HVar(P) - \vec x$.
For notational convenience, we also use the ``\emph{project-away operator}''
$\projp{\vec y}P$ instead of $\projpd{\vec x}P$, where
$\vec y$ are the variables being projected away, 
i.e. 
$\vec x = \HVar(\projp{\vec y}P) = \HVar(P) - \vec y$.

Each subplan $P$ returns an
intermediate relation of arity $|\HVar(P)|+1$. 
The extra
\emph{probability attribute} stores a $\score(t)$ for each output tuple
$t \in P(D)$. 
Given a probabilistic database $D$ and a plan $P$, 
$\score(t)$ is defined
inductively on the structure of $P$ as follows: 
(1) If $t \in R_i(\vec
x)$, then $\score(t)= \PP{t}$, i.e.\ its probability in $D$;
(2) if $t \in \,
\joinp{}{P_1(D), \ldots, P_k(D)}$ where $t = \joinp{}{t_1, \ldots,
  t_k}$, then $\score(t) = \prod_{i=1}^k\score(t_i)$; and
(3) if $t \in
\projpd{\vec x}P(D)$, and $t_1, \ldots, t_n \in P(D)$ are all the
tuples that project onto $t$, then 
$\score(t) = \bigotimes_{i=1}^n
\score(t_i)$, where ``$\otimes$'' stands for the independent-or.
In other words, $\score$ computes a probability by
assuming that all tuples joined by $\Join^p$ and
all duplicates eliminated by $\projpd{}$ are \emph{independent}. 
Only if
these conditions hold, then $\score$ is the correct query probability (also called ``query reliability'' \cite{DBLP:conf/pods/GradelGH98}),
but in general it is not. 
Therefore, $\score$ is also called an
\emph{extensional semantics}~\cite{DBLP:journals/tois/FuhrR97,Pearl:1988pb,DBLP:journals/debu/ReDS06} and is, in general,
not equal to the query probability, which is defined in terms of possible worlds: 
$\textit{score}(P) \neq \PP{q_p}$.\footnote{
  \emph{Extensional approaches} compute the probability of any formula
  as a function of the probabilities of its subformulas according to
  syntactic rules, regardless of how those were
  derived. \emph{Intensional approaches} reason in terms of possible
  worlds and keep track of dependencies~\cite{Pearl:1988pb}.}
For a Boolean plan $P$, we get one single score,
which we denote $\textit{score}(P)$.

The requirement that joins and projections alternate is w.l.o.g.\
because nested joins,
such as $\joinp{}{\joinp{}{R_1, R_2}, R_3}$ 
or $\joinp{}{R_1, \joinp{}{R_2, R_3}}$,
can
be rewritten into $\joinp{}{R_1, R_2, R_3}$ while keeping the
same score, 
e.g., $(p_1 p_2) p_3 = p_1 (p_2 p_3)$.  
For the same reason we do not distinguish
between different permutations in the joins,
called join orders~\cite{Moerkotte:BuildingQueryCompilers}. 
We do not focus on query optimization in this paper until \autoref{sec:optimizations}.

\begin{definition}[Safe plan]\label{def:safePlan}
A plan $P$ is called ``\emph{safe}'' iff, for each join  $\joinp{}{P_1, \ldots, P_k}$, 
the head variables of each subplan $P_i$ contain the same existential variables of plan $P$:
$\HVar(P_i) \cap \EVar(P) = \HVar(P_j) \cap \EVar(P), \forall 1 \leq i, j \leq k$.
\end{definition}

The recursive definition of \autoref{lemma:hierarchical} gives us
immediately a safe plan for a hierarchical query.  Conversely, every
safe plan defines a hierarchical query.  
We next illustrate this.

\begin{example}[Hierarchical queries and safe plans]\label{ex:IntroExample}
Consider
$q(x) \datarule  R(x,y), S(x, y,z), T(y,z,u)$, depicted in
\Autoref{Fig_IntroExampleIncidenceMatrix_b}
with its ``\emph{augmented incidence matrix}''.\footnotemark~The \emph{incidence matrix} $I(q)$ of a Boolean sj-free CQ $q$ with $m$ atoms and $k$ variables is a $m \times k$-dimensional 01-matrix with $I(i,j)= 1$ iff $x_j \in \Var(a_i)$.
We ``augment'' it in three ways: 
(1) We replace 1-entries with circles ($\circ$) and ignore 0-entries: 
this merely cosmetic change makes it is easier to recognize patterns;
(2) We separate columns for $\HVar(q)$ to the left and $\EVar(q)$ to the right: 
recall that query components and safety are only determined by $\EVar(q)$. 
For our example we have $\HVar(q) = \{x\}$ and $\EVar(q) = \{y, z, u\}$;
(3)~If the query is hierarchical, then we emphasize the hierarchy between $\EVar(q)$ with gray background. 
For our example we have 
$\at(u)\subseteq \at(z) \subseteq \at(y)$.
\Autoref{Fig_IntroExampleIncidenceMatrix_c}
shows the corresponding safe plan of $q$ where the hierarchy is reflected in the order in which variables are projected away:
first $u$, then $z$, finally the separator variable $y$.
\Autoref{Fig_IntroExampleIncidenceMatrix_d} shows the translation into SQL assuming 
$R(A,B), S(A, B, C), T(B, C, D)$ as schema 
and each table having one additional attribute \sql{P} for the probability of a tuple. 
Here \sql{IOR(X)} is a user-defined aggregate (UDA) that 
calculates the independent-or for the probabilities of grouped tuples, i.e.\
\sql{IOR}$(p_1,\ldots,p_n) = \bigotimes_{i=1}^n p_i$. 
See \cite{DBLP:journals/tods/GatterbauerS14} for the complete UDA definition in PostgreSQL.

Next consider $q'\datarule  R(x,y), S(x, y,z), T(y,z,u)$, i.e.\ a variant of $q$ where
$x \in \EVar(q')$. 
Now $q'$ is not hierarchical anymore since 
$\at(x)\not \subseteq \at(z)$,
$\at(x)\not \supseteq \at(z)$,
and $\at(x) \cap \at(z) \not = \emptyset$.
Starting with $P$ from \autoref{Fig_IntroExampleIncidenceMatrix_c} 
and replacing the final projection $\projp{y}$ with $\projp{x,y}$, 
the plan $P'$ is now unsafe: 
the join 
$\joinp{}{S(x,y,z), \projp{u}T(y,z,u)}$,
has ($i$)
$\HVar(S(x,y,z)) = \{x,y,z\}$,
but 
($ii$)
$\HVar(\projp{u}T(y,z,u)) = \{y,z\}$:
their intersections with $\EVar(P')=\{x,y,z,u\}$ are now different.
\markend
\end{example}
\footnotetext{
\emph{Incidence matrices} allow us to compactly reason about two types of relationships between variables and relations of sf-free CQs simultaneously:
($i$) in a column: a variable that is shared across relations, 
and ($ii$) in a row: relations that are joined by a variable. 
They thus allow us to reason about both the ``query hypergraph'' and the ``dual query hypergraph''
at the same time,
which is helpful also for other types of problems involving sf-free CQs (see, e.g.~\cite{DBLP:journals/pvldb/FreireGIM15}).
}

\begin{figure}[t]
\vspace{-3mm}	
\centering
\renewcommand{\tabcolsep}{0.9mm}
\renewcommand{\arraystretch}{0.95}
\subfloat[Datalog notation]{
	\label{Fig_IntroExampleIncidenceMatrix_a}
	\begin{minipage}[t]{41mm}
		\vspace{6mm}
		$q(x) \datarule R(x,y), S(x,y,z), T(y,z,u)$
		\vspace{2mm}
	\end{minipage}
}
\hspace{0mm}
\subfloat[Augmented incidence matrix]{
	\label{Fig_IntroExampleIncidenceMatrix_b}
	\begin{minipage}[t]{38mm}	
		\hspace{9mm}
	\begin{tabular}[t]{@{\hspace{1pt}} >{$}c<{$}|>{$}c<{$} | >{$}c<{$} >{$}c<{$}
		 								>{$}c<{$} >{$}c<{$} @{\hspace{1pt}}}
			& x    	& y	 			& z 			& u		\\
			\hline
		R	&\circ	&\graycell\circ	 						\\
		S	&\circ	&\graycell\circ &\graycell\circ 		\\
		T	&		&\graycell\circ &\graycell\circ &\graycell\circ
	\end{tabular}
	\end{minipage}
	}
\hspace{0mm}
\subfloat[Unique safe plan]{
	\label{Fig_IntroExampleIncidenceMatrix_c}
	\begin{minipage}[t]{\linewidth}
		\vspace{1mm}
		\centering
		$P = \projp{y} \joinp{}{R(x,y), \projp{z} \joinp{}{S(x,y,z), \projp{u}T(y,z,u)}}$
	\end{minipage}
}
\hspace{0mm}
\subfloat[Safe plan in SQL]{
	\label{Fig_IntroExampleIncidenceMatrix_d}
	\begin{minipage}[t]{\linewidth}	
	\vspace{-3mm}
	{\sql{	
	\begin{tabbing}
	\hspace{0.3cm}\=\hspace{0.9cm}\=\hspace{0.9cm}\=\hspace{0.9cm}\=
	\hspace{0.9cm}\=\hspace{0.9cm}\=\hspace{0cm}\=\kill
	\>select X4.A, IOR(P) as P\\
	\>from	\>(select R.A, R.B, R.P * X3.P as P		\\
	\>		\>from	\>(select	\> X2.A, X2.B, IOR(P) as P		\\
	\>		\>		\>from 		\>	(select	\> S.A, S.B, S.C, S.P * X1.P as P  	\\
	\>		\>		\>			\>	from	\> (select	\> T.B, T.C, IOR(P) as P 	\\
	\>		\>		\>			\>			\> from 	\> T			\\
	\>		\>		\>			\>			\> group by	B, C) as X1, S	\\
	\>		\>		\>			\>	where 	\> S.B=X1.B and S.C=X1.C) as X2	\\
	\>		\>		\>group by  X2.A, X2.B) as X3, R	 			\\
	\>		\>where R.A=X3.A and R.B=X3.B) as X4	\\
	\>group by X4.A 
	\end{tabbing}
	}}
	\end{minipage}
	}	
\caption{\autoref{ex:IntroExample}. 
A query $q$ in Datalog (a), its augmented incidence matrix (b),
its unique safe plan in our plan notation (c), and in SQL (d).
}
\label{Fig_IntroExampleIncidenceMatrix}
\end{figure}

The following proposition summarizes our discussion:

\begin{proposition}[Safety \cite{DBLP:journals/vldb/DalviS07}]\label{prop:uniqueSafePlan}
  (1) Let $P$ be a plan for query $q$. Then $\textit{score}(P) =
  \PP{q}$ for any probabilistic database iff $P$ is safe. 
  (2)   Assuming \#P$\neq$PTIME, a query $q$ is safe (i.e.\ $\PP{q}$ has
  PTIME data complexity) iff it has a safe plan $P$; in that case the
  safe plan is unique (up to permutation in the join orders), and $\PP{q} = \score(P)$.
\end{proposition}

\subsection{Boolean Formulas} 
Consider a set of Boolean variables $\mathbf{X}
= \set{X_1, X_2, \ldots}$ and a probability function $p : \mathbf{X}
\rightarrow [0,1]$.  Given a Boolean formula $F$, denote $\PP{F}$ the
probability that $F$ is true if each variable $X_i$ is independently true
with probability $p(X_i)$. 
In general, computing
$\PP{F}$ is \#P-hard in the number of variables $\mathbf{X}$. 

If $D$
is a probabilistic database then we interpret every tuple $t \in D$ as
a Boolean variable and denote the lineage of a Boolean query $q \datarule
a_1,\ldots,a_m$ on $D$ as the Boolean DNF formula $F_{q,D} =
\bigvee_{\theta: \theta \models q}\theta(a_1)\wedge\cdots\wedge
\theta(a_m)$, where $\theta$ 
ranges over all assignments of $\EVar(q)$ 
to constants in the active domain
that satisfy $q$ on $D$.
It is well known that $\PP{q} = \PP{F_{q,D}}$. In other
words the probability of a Boolean query is the same as the
probability of its lineage formula.

\begin{example}[Lineage]\label{ex:simple} 
If $F= X Y_1 Z_1 \vee X Y_2 Z_2$, then 
$\PP{F} = p(X) \big( p(Y_1) p(Z_1) \otimes p(Y_2) p(Z_2)\big)$. 
Next consider a query $q \datarule R(x), S(x,y), T(y)$ over the
database $D$ from \autoref{Fig_DissociationExample_a}.
Then the lineage formula is 
$F_{q,D} = \big( R(a)\wedge S(a,b)\wedge T(b)  \big)
	\vee \big( R(a) \wedge S(a,c) \wedge T(c) \big)$, 
i.e.\ the same as $F$ up to variable renaming. It is now easy to
see that $\PP{q} = \PP{F_{q,D}}$.
\end{example}

A key technique that we use in this paper is the following result
from~\cite{DBLP:journals/tods/GatterbauerS14}:  Let $F$ and $F'$ be two
Boolean functions 
with sets of variables $\mathbf{X}$ and $\mathbf{X}'$,
respectively.  We say that $F'$ is a ``\emph{dissociation}'' of $F$ if
there exists a substitution $\theta: \mathbf{X}' \rightarrow
\mathbf{X}$ such that $F'[\theta]=F$.  If $\theta^{-1}(X) =
\set{X',X'', \ldots}$ then we say that the variable $X$ {\em
  dissociates into} $X', X'', \ldots$; if $|\theta^{-1}(X)| = 1$ then
we assume w.l.o.g.\ that $\theta^{-1}(X) = X$ (up to variable renaming)
and we say that $X$ does not dissociate.  Given a probability function
$p : \mathbf{X} \rightarrow [0,1]$, we extend it to a probability
function $p' : \mathbf{X}' \rightarrow [0,1]$ by setting $p'(X') =
p(\theta(X'))$.  Then, we have previously shown:

\begin{theorem}[Oblivious DNF bounds~\cite{DBLP:journals/tods/GatterbauerS14}] \label{th:bool:dissoc} Let
  $F'$ be a monotone DNF formula that is a dissociation of $F$ through the
  substitution $\theta$.  Assume that for any variable $X$, no two
  distinct dissociations $X', X''$ of $X$ occur in the same prime
  implicant of $F'$.  
  Then:
  (1) $\PP{F} \leq \PP{F'}$, and 
  (2) if every
  dissociated variable $X \in \mathbf{X}$ is deterministic 
  (i.e.\
  $p(X)=0$ or $p(X)=1$), then $\PP{F}=\PP{F'}$.
\end{theorem}

Intuitively, a dissociation $F'$ is obtained from a formula $F$ by
replacing different occurrences of a variable $X$ with
fresh variables $X', X'', \ldots$; 
by doing this, $\PP{F'}$ gives us an upper bound for $\PP{F}$
and may be easier to compute.

\begin{example}[Dissociation]
$F' = X'Y \vee X''Z$ is a dissociation of $F = XY \vee XZ$, and its probability is 
$\PP{F'} = p(X)p(Y) \otimes p(X) p(Z)$.
Here, only the variable $X$
dissociates into $X', X''$.  It is easy to see that $\PP{F} \leq
\PP{F'}$. Moreover, if $p = 0$ or 1, then $\PP{F}=\PP{F'}$.
The condition that no two dissociations of the same variable occur in
a common prime implicant is necessary: for example, $F'= X'X''$ is a
dissociation of $F = X$ as $X = X X$.
However, $\PP{F} = p(X)$, $\PP{F'} = p(X)^2$,  and
thus $\PP{F} \not \leq \PP{F'}$.
\end{example}

\section{Dissociation and propagation for unsafe queries}\label{sec:dissociation}
This section defines our technique of ``\emph{query dissociation}'' and defines the ``\emph{propagation score}'' of a query.  
Our motivation comes from \autoref{th:dichotomy}: hierarchical queries are safe (i.e.\ in PTIME), while non-hierarchical queries are unsafe (i.e.\ \#P-hard). 
At its very core, our approach will approximate the probability of a non-hierarchical query with a set of related hierarchical queries.
We first define our approach
(\autoref{sec:queryDissociation}),
then draw the connection to propagation in graphs
(\autoref{sec:graphConnection}),
and finally derive a partial order between a set of hierarchical queries
(\autoref{sec:partialDissociationOrder}).

\begin{figure}[t]
\vspace{-3mm}	
\centering
	\renewcommand{\tabcolsep}{0.9mm}
	\renewcommand{\arraystretch}{0.95}
	\subfloat[$q$]{
		\begin{minipage}[t]{30mm}
		\hspace{10mm}
		\begin{tabular}[t]{@{\hspace{1pt}} >{$}c<{$}|>{$}c<{$} >{$}c<{$} @{\hspace{1pt}}}
				& x	& y  \\
			\hline
			R	& \circ	 \\
			S	& \circ & \circ \\
			T	& & \circ 
		\end{tabular}\label{Fig_DissociationExampleIncidenceMatrix_a}
		\end{minipage}	}
	\hspace{8mm}	
	\subfloat[$q^{\Delta}$]{
		\begin{minipage}[t]{35mm}
		\hspace{12mm}
		\begin{tabular}[t]{@{\hspace{1pt}} >{$}c<{$}|>{$}c<{$} >{$}c<{$} @{\hspace{1pt}}}
				& x	& y  \\
			\hline
			R	& \circ	& \bullet 	\\
			S	& \circ & \circ 	\\
			T	& & \circ 
		\end{tabular}\label{Fig_DissociationExampleIncidenceMatrix_b}
		\end{minipage}	}	
	\hspace{2mm}
	\renewcommand{\tabcolsep}{0.6mm}
	\subfloat[$D$]
		{{\begin{minipage}[t]{30mm}
			\mbox{
					\begin{tabular}[t]{ >{$}c<{$} |  >{$}c<{$} }
		 			R	& A	\\
					\hline
					p_1	& a			
					\end{tabular}	}
			\hspace{-1mm}
			\mbox{
					\begin{tabular}[t]{ >{$}c<{$} | >{$}c<{$} >{$}c<{$} >{$}c<{$} >{$}c<{$}}
		 			S	& A		& B\\
					\hline
					p_2	& a		& b		\\
					p_3	& a		& c		
					\end{tabular}		}
			\hspace{-1mm}
			\mbox{
					\begin{tabular}[t]{ >{$}c<{$} | >{$}c<{$} }
		 			T	& B			\\
					\hline
					p_4	& b		\\	
					p_5	& c		
					\end{tabular}				}
		\end{minipage}	
		}\label{Fig_DissociationExample_a}}
	\hspace{8mm}
	\subfloat[$D^{\Delta}$]
	   {{\begin{minipage}[t]{34mm}
       	\mbox{
	   			\begin{tabular}[t]{ >{$}c<{$} | >{$}c<{$}	>{$}c<{$}   }
	    		R^y	&  A	& B		\\
	   			\hline
	   			p_1	&  a		& b	\\
	   			p_1	&  a		& c					
	   			\end{tabular}		   	}
	   	\hspace{-1mm}
	   	\mbox{
	   			\begin{tabular}[t]{ >{$}c<{$} | >{$}c<{$} >{$}c<{$} >{$}c<{$} >{$}c<{$}}
	    		S	& A		& B\\
	   			\hline
	   			p_2	& a		& b		\\
	   			p_3	& a		& c		
	   			\end{tabular}	   	}
	   	\hspace{-1mm}
	   	\mbox{
	   		\begin{tabular}[t]{ >{$}c<{$} | >{$}c<{$} }
	   		T	& B			\\
	   		\hline
	   		p_4	& b		\\	
	   		p_5	& c		
	   		\end{tabular}	   	}
	   \end{minipage}
	   }\label{Fig_DissociationExample_b}}				
\caption{\autoref{ex:queryDissociation}: 
Incidence matrices of $q \datarule R(x), S(x,y), T(y)$ and dissociation $q^{\Delta} \datarule \R^y(x,y), S(x,y), T(y)$.
Original database $D$ and new database $D^{\Delta}$ with table $R$ dissociated on variable $y$.}
\label{Fig_DissociationExampleIncidenceMatrix}
\end{figure}

\subsection{Query dissociation}\label{sec:queryDissociation}

\begin{definition}[Query dissociation]\label{def:Dissociation}
Given a probabilistic database $D$ and a 
query $q(\vec z) \datarule R_1(\vec x_1), \ldots, R_m(\vec x_m)$.
  Let $\Delta = (\vec y_1, \ldots, \vec y_m)$ be a collection of sets
  of variables with $\vec y_i \subseteq \EVar(q)- \vec x_i$ for every relation $R_i$.
The ``\emph{query dissociation}'' defined by $\Delta$ has then two components:

\begin{enumerate}[nolistsep,label=(\arabic*)]

\item the ``\emph{dissociated query}'':
\begin{align*}
	q^{\Delta}(\vec z) \datarule R_1^{\vec y_1}(\vec x_1, \vec y_1), \ldots, R_m^{\vec y_m}(\vec x_m, \vec y_m)
\end{align*}
where each $R_i^{\vec y_i}(\vec x_i, \vec
y_i)$ is a new relation of arity $|\vec x_i| + |\vec y_i|$.

\item the ``\emph{dissociated database}'' $D^\Delta$ consisting of
  the tables over the vocabulary 
  $\sigma^\Delta = (R_1^{\vec y_1},\ldots, R_m^{\vec y_m})$
  obtained by 
	replacing each table $R_i$ with the $k_i$-fold Cartesian product
  $R_i \times \ADom_{y_{i1}} \times \cdots \times \ADom_{y_{ik}}$
  where $\vec y_i = (y_{i1}, \ldots, y_{ik_i})$. 
  For each new tuple $t' \in R_i^{\vec y_i}$, its probability is
  $p'(t') = p(\pi_{\vec x_i}\,t')$, i.e.\ 
  the probability of $t$ in the database $D$.
\markend
\end{enumerate}
\end{definition}

Thus, conceptually, we define the semantics of ``query dissociation'' as follows:  Add some existential
variables to some atoms in the query; this results in a \emph{dissociated query} over a new schema. 
Transform the probabilistic database by replicating some of their tuples and by adding new attributes to match the new schema; this is the \emph{dissociated database}.  
Finally, compute the probability of the dissociated query on the dissociated database. 
Recall that each tuple in
the original table represents an independent probabilistic event. The dissociated table now contains multiple copies of each tuple, all with the same probability, yet considered to represent \emph{independent} events. Thus, the dissociated table has a different probabilistic interpretation than the original table.
Notice that this is the semantics of a dissociated query, and not the way we actually evaluate queries
(in later sections we describe methods that evaluate the dissociated query without modifying the tables in the database).

\begin{example}[\autoref{ex:simple} cont.]\label{ex:queryDissociation}
We illustrate with the query $q \datarule R(x), S(x,y), T(y)$ and the database shown in 
\autoref{Fig_DissociationExample_a}
where a variable $p_i$ stands for the independent probability of a tuple with index $i$.
Then $\Delta = (\{y\}, \emptyset, \emptyset )$ defines the following dissociation:
$q^{\Delta} \datarule R^{y}(x,y), S(x,y), T(y)$.
Notice we write here and later $R^y$ instead of $R^{\{y\}}$ to simplify our notation.
The 
active domain $\ADom_y$ is $\{b,c\}$, and 
\autoref{Fig_DissociationExample_b} shows the new database with table $R^y$ as the original table $R$ dissociated on variable $y$.
Notice that the original tuple $R(a)$ got dissociated into two tuples $R^y(a,b)$ and $R^y(a,c)$ with the same probability $p_1$.
\Autoref{Fig_DissociationExampleIncidenceMatrix_b} shows $q^{\Delta}$ with the help of an incidence matrix that is \emph{augmented} in a 4th way:
while an empty circle ($\circ$) still indicates that the original relation contains a variable,
a full circle ($\bullet$) now indicates that a relation is dissociated on a variable.
Notice that the lineage of the dissociated query $q^\Delta$ is
$F_{q^{\Delta},D^{\Delta}} = R^y(a,b),S(a,b), T(b) \vee R^y(a,c),S(a,c), T(c)$ and
is the same (up to variable renaming) as the dissociation of the
lineage of query $q$: $F'= X' Y_1 Z_1 \vee X''Y_2 Z_2$.  
Also notice the deliberate similarity 
with \autoref{ex:1} and \autoref{Fig_IntroductionExample} from the introduction.
\markend
\end{example}

This example generalizes and allows us to prove
our first major technical result that query dissociation can only increase the probability:

\begin{theorem}[Upper query bounds]\label{th:upperBounds}
For every database $D$ and every dissociation $\Delta$ of a query $q$:
$\PP{q^{\Delta}} \geq \PP{q}$.
\end{theorem}

\begin{proof}[\autoref{th:upperBounds}]
This follows immediately from \autoref{th:bool:dissoc} by
  noting that the lineage $F_{q^\Delta,D^\Delta}$ is a dissociation of
  the lineage $F_{q,D}$ through the substitution $\theta : D^\Delta
  \rightarrow D$ defined as follows: for every tuple $t' \in R_i^{\vec
    y_i}$, $\theta(t') = \pi_{\vec x_i}(t')$.
\end{proof}

By \autoref{th:dichotomy}, the probability of a dissociation can be evaluated in PTIME iff $q^\Delta$ is hierarchical.
Hence, amongst all dissociations, we are interested in those that are easy to evaluate
and use them as a technique to approximate the probabilities of queries
that are
hard to compute:

\begin{definition}[Hierarchical dissociation]
A dissociation $\Delta$ of a query $q$ is called ``\emph{hierarchical}'' 
if the dissociated query ${q^{\Delta}}$ is hierarchical.
\end{definition}

\noindent
The idea now is simple: 
Find a hierarchical dissociation $\Delta$, compute $\PP{q^{\Delta}}$, and thereby obtain an upper bound on $\PP{q}$.  In fact,
we will consider \emph{all} hierarchical dissociations and take the
minimum of their probabilities, since this gives an even better upper bound
on $\PP{q}$ than that by a single dissociation.  We call this
quantity the ``\emph{propagation score}'' of the query $q$
because of 
  similarities with efficient relevance propagation algorithms on graphs.
\Autoref{fig:NetworkConnection} and the next subsection explain in more detail how query dissociation generalizes 
``\emph{relevance propagation}'' from graphs to hypergraphs.

\begin{definition}[Propagation]
  The ``\emph{propagation score}'' $\rho(q)$ for a query $q$ is the minimum probability of all \emph{hierarchical dissociations}, 
  i.e.\ $\rho(q) = \min_{\Delta} \PP{q^{\Delta}}$ with $\Delta$ ranging over all hierarchical dissociations.
\end{definition}

\definecolor{bad}{rgb}{0.96, 0.85, 0.85}

\definecolor{good}{rgb}{0.84, 0.91, 0.80}

\begin{figure}[t]
\centering
\small
\renewcommand{\tabcolsep}{0.7mm}
\renewcommand{\arraystretch}{1.20}
\begin{tabular}[b]{@{\hspace{0pt}} c | c @{\hspace{0pt}}}
\textbf{\normalsize Networks}	& \textbf{\normalsize Conjunctive queries} \\
\textbf{\normalsize (Graphs)}	& \textbf{\normalsize (Hypergraphs)} \\
\hline
\multicolumn{1}{@{\hspace{0pt}}p{41mm}|}{
\textbf{Network reliability}:
Probability that two nodes are connected.
\emph{Independent} of edge direction.
\mbox{\colorbox{good}{\hspace{10mm}($+$) undirected \hspace{16mm}}}
\mbox{\colorbox{bad}{\hspace{10mm}($-$) \#P-hard\hspace{30mm}}}
}
	&  \multicolumn{1}{p{40.5mm}@{\hspace{0pt}}}{
	\textbf{Query reliability}:
	Probability that query is true in a random world.
	\emph{Independent} of query plan.
	\mbox{\colorbox{good}{\hspace{10mm}($+$) undirected \hspace{10mm}}}
	\mbox{\colorbox{bad}{\hspace{10mm}($-$) \#P-hard \hspace{13mm}}}
	}		\\
\hline
\multicolumn{1}{@{\hspace{0pt}}p{41mm}|}{
\textbf{Propagation score}:
`Relatedness' propagates from source to target.
\emph{Dependent} on edge direction.
\emph{Upper bound} to reliability.
\hspace{10mm}
\mbox{\colorbox{bad}{\hspace{10mm}($-$) directed\hspace{20mm}}}
\mbox{\colorbox{good}{\hspace{10mm}($+$) PTIME \hspace{20mm}}}
}
	&  \multicolumn{1}{p{40.5mm}@{\hspace{0pt}}}{
	\textbf{Dissociation score}:
	A query plan evaluates from leafs to root.
	\emph{Dependent} on choice of dissociation.
	\emph{Upper bound} to reliability.
	\mbox{\colorbox{bad}{\hspace{10mm}($-$) directed\hspace{13.4mm}}}
	\mbox{\colorbox{good}{\hspace{10mm}($+$) PTIME \hspace{13.4mm}}}
	}		\\
\hline
	&  \multicolumn{1}{p{40.5mm}@{\hspace{0pt}}}{
	\textbf{Propagation score}:
	\emph{Minimum} over all hierarchical dissociations.
	\emph{Unique} for given query.
	\mbox{\colorbox{good}{\hspace{10mm}($+$) undirected \hspace{10mm}}}
	\mbox{\colorbox{good}{\hspace{10mm}($+$) PTIME \hspace{13.6mm}}}
	}		\\
\end{tabular}
\caption{Connection between \emph{reliability} and \emph{propagation} in networks and conjunctive queries (CQs):
In contrast to networks, the propagation score for CQs is the minimum over all possible hierarchical dissociations, and is therefore unique for every query and database.
($+$) and ($-$) denote positive or negative properties (best seen in color).}\label{fig:NetworkConnection}
\end{figure}

We propose to adopt the \emph{propagation score as an alternative semantics for ranking query results over probabilistic databases}.  
While the data complexity of computing the reliability $r(q)$ is \#P-hard in general, computing the propagation score $\rho(q)$ is always in PTIME in the size of the database.  Furthermore, $\rho(q) \geq r(q)$ and, if $q$ is safe, then $\rho(q) = r(q)$. Both claims follow immediately from \autoref{th:upperBounds}. 
Hence, the propagation score is a \emph{natural generalization of reliability from safe queries to all queries}: If the query is safe, both scores coincide; if the query is unsafe, propagation still allows to evaluate the query in PTIME (in addition, the next two sections will show how to evaluate the propagation very efficiently without first dissociating the tables).

\subsection{Dissociation and the relation to propagation on graphs}\label{sec:graphConnection}

Recall that our original motivation was to develop for queries a concept that is analogous to propagation on directed networks.
Queries have no concept of direction, and we suggest that the \emph{choice of direction} in a graph corresponds to a particular \emph{choice of hierarchical dissociation} of a query.
We now justify our definitions of query dissociation and propagation by drawing the connection to network reliability and propagation:
When a digraph is $k\!+\!1$-partite, then its two terminal reliability can be expressed by a conjunctive $k$-chain query.\footnote{A \emph{conjunctive $k$-chain query} is a query $q$ without self-joins in which each relation is binary, all relations are joined together, and there is no single variable common to more than two relations. Furthermore, the first and last variable are head variables and can be replaced by constants: $q(x_1, x_{k+1}) \datarule R_1(x_1, x_2), R_2(x_2, x_3), \ldots, R_k(x_k, x_{k+1})$. 
The fact that relations are binary entails that the query hypergraph is actually a standard graph. 
Similarly, the fact that a variable is not common to more than two relations also entails the ``dual hypergraph'' to be a graph as well.
The expression \emph{chain query} derives from the observation that both its hypergraph and dual hypergraph resemble a simple chain.}
Further, the propagation score over this network corresponds to one of several possible dissociations of this query $q$,
some of which khave no natural correspondence to propagation on graphs.
 Thus, \emph{query dissociation is a strict generalization of network propagation on $k$-partite graphs},
and we define query propagation as the minimum reliability of a set of query dissociations (see~\autoref{fig:NetworkConnection}).
Notice, however, that dissociation admits a natural interpretation as network propagation only on $k$-partite graphs, and says nothing about graphs that are not $k$-partite.

In the following, we use $[k]$ to denote the set $\{1, \ldots, k\}$, and $\vec x_{[i,j]}$ as short form for 
$(x_i, x_{i+1}, \ldots, x_j)$.

\begin{proposition}[Connection to networks]\label{prop:connectionPropagationScore}
  Let $G = (V, E)$ be a $k\!+\!1$-partite digraph with a source node $s$ and
  a target node $t$, where each edge has a probability.  The nodes are
  partitioned into $V = \set{s} \cup V_2 \cup \ldots \cup V_{k} \cup
  \set{t}$, and the edges are $E = \bigcup_i R_i$, where $R_i$ denotes
  the set of edges from $V_i$ to $V_{i+1}$ with $i \in [k]$.  Then:
\begin{enumerate}[label=\textup{(\alph*)}, itemsep=0pt, parsep=1pt, topsep = 1pt]

	\item The ($s,t$)-network reliability of $G$ is $\PP{q}$ with:
	\begin{align*}
		q & \datarule  R_1(s,x_2), R_2(x_2,x_3), \ldots, R_k(x_{k}, t)
	\end{align*}

	\item The directed propagation score from $s$ to $t$ (as defined in \autoref{ex:1}) is $\PP{q^\Delta}$ with:
	\begin{align*}
		q^\Delta  \datarule & R_1^{\vec x_{[3,k]}}(s, \vec x_{[2, k]}),
		R_2^{\vec x_{[4,k]}}(\vec x_{[2,k]}), \ldots, R_k^{\emptyset}(x_{k}, t)
	\end{align*}
	\end{enumerate}
\end{proposition}

\subsection{Partial dissociation order}\label{sec:partialDissociationOrder}

The difficulty in computing $\rho(q)$ is that the total number of
dissociations is large even for relatively small queries: 
the number corresponds to the cardinality of the power set of variables that can be added to atoms.
Thus, if $q$ has
$k$ existential variables and $m$ atoms, then $q$ has $2^{|K|}$
possible dissociations with $K = \sum_{i=1}^m \big( k- |
\texttt{Var}(a_i) | \big)$ forming a \emph{partial order} in the shape of a
\emph{power set lattice}:

\begin{definition}[Partial dissociation order]
We define the partial order on the dissociations of a query as:
\begin{align*}
	\Delta \preceq \Delta'   \,\,\Leftrightarrow\,\,  \forall i: \vec y_i \subseteq \vec y_i'
\end{align*}
\end{definition}

Whenever $\Delta \preceq \Delta'$, then $q^{\Delta'}, D^{\Delta'}$ is
a dissociation of $q^{\Delta}, D^{\Delta}$ (given by $\Delta'' =
\Delta' - \Delta$).  Therefore, we obtain immediately:
If $\Delta \preceq \Delta'$ then $\PP{q^\Delta} \leq
\PP{q^{\Delta'}}$.
However, the statement holds in both directions:

\begin{theorem}[Partial dissociation order]\label{th:partialDissociationOrder}
	For every two dissociations $\Delta$ and $\Delta'$ of a query $q$, the following holds over every database:
\begin{equation*}
		\Delta \preceq \Delta'   \,\,\Leftrightarrow\,\,  \PP{q^{\Delta}} \leq \PP{q^{\Delta'}}
\end{equation*}
\end{theorem}

\begin{example}[Partial dissociation order]\label{ex:partialDissociation}
	Consider the query $q \datarule R(x), S(x), T(x,y), U(y)$.
It is unsafe and allows $2^3 = 8$ dissociations which are shown in \autoref{Fig_PartialDissociationOrderExample_a} with the help of augmented incidence matrices.
Among the 8 dissociations, 5 are hierarchical,
and 2 among those 5 are minimal:
\begin{align*}
	q^{\Delta_3}	& \datarule  R(x), S(x), T(x,y), U^x(x,y) \\
	q^{\Delta_4}	& \datarule  R^y(x, y), S^y(x, y), T(x,y), U(y)
\end{align*}
The propagation score is the minimum score of all \emph{minimal hierarchical dissociations}:
$\rho(q) = \min_{i \in \{3,4\}} \PP{q^{\Delta_i}}$.
To illustrate that these dissociations are upper bounds, consider a database with 
$R=T=U =\{(1),(2)\}$, 
$S=\{(1,1),$ $(1,2),(2,2)\}$, 
and the probability of all tuples being $\frac{1}{2}$. 
Then $\PP{q} = \frac{83}{2^9} \approx 0.161$, 
while $\PP{q^{\Delta_3}}=\frac{169}{2^{10}} \approx 0.165$
and $\PP{q^{\Delta_4}}=\frac{353}{2^{11}} \approx 0.172$.
Both dissociations give upper bounds, and the propagation score is their minimum ($\approx 0.165$).
\Autoref{Fig_PartialDissociationOrderExample_b} is explained later in \Autoref{ex:queryPlans}.
\markend
\end{example}

\begin{figure}[t]
	\vspace{-3mm}
    \centering
	\subfloat[]{
		\label{Fig_PartialDissociationOrderExample_a}
		\hspace{-4mm}
		\includegraphics[scale=0.43]{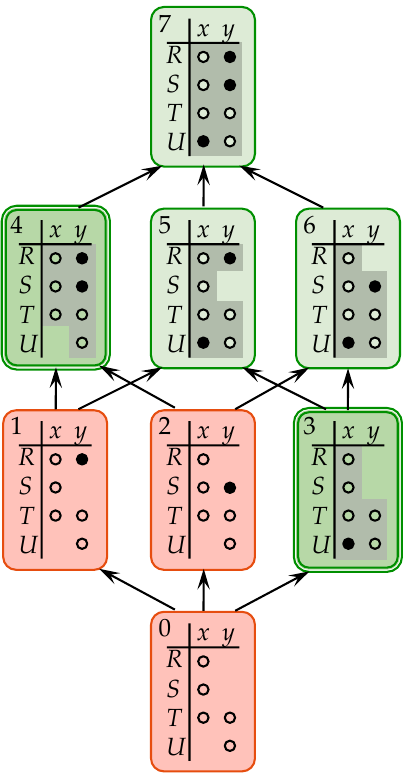}
		\hspace{-3mm}}
	\hspace{4mm}	
	\subfloat[]{
		\includegraphics[scale=0.53]{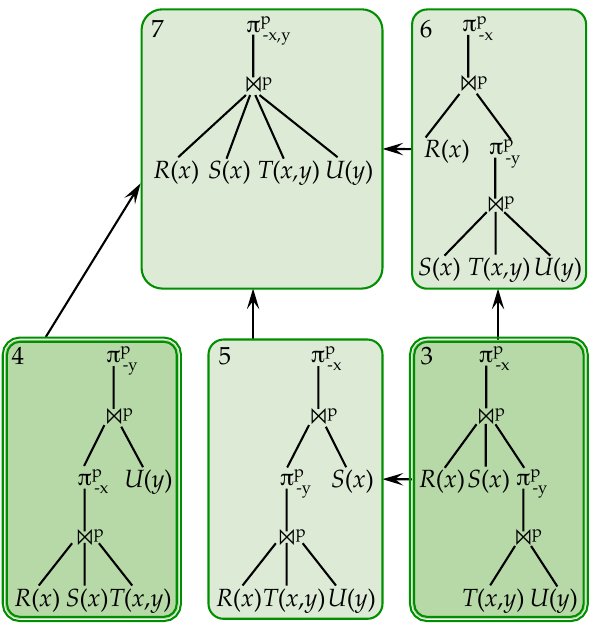}
		\hspace{-2mm}
		\label{Fig_PartialDissociationOrderExample_b}}
	\hspace{0mm}
	\caption{
		\autoref{ex:partialDissociation} and \autoref{ex:queryPlans}:
		(a):
		Partial dissociation order for $q \datarule R(x),S(x),$ $T(x,y), U(y)$.
		``\emph{Hierarchical dissociations}'' are green and have the hierarchies between variables shown in their augmented incidence matrices (3 to 7), 
		``\emph{minimal hierarchical dissociations}'' are dark green and double-lined (3 and 4). 
		(b):
		All 5 query plans, their correspondence to hierarchical dissociations, and their partial dissociation order (best viewed in color).
		}
	\label{Fig_PartialDissociationOrderExample}
	\vspace{-1mm}
\end{figure}

The smallest element in the lattice of dissociations is $\Delta_\bot = (\emptyset, \ldots, \emptyset)$ with
$q^{\Delta_\bot}=q$,
and the largest element is $\Delta_\top =
(\Var(q)-\Var(a_1), \ldots, \Var(q)-\Var(a_m))$. 
$q^{\Delta_\top}$ is always
hierarchical as every atom contains all variables.  As we move up in the
lattice the probability increases, but the 
hierarchy
status may
toggle arbitrarily from hierarchical to non-hierarchical and back.  
For example, the query 
$q \datarule R(x), S(y), T(x,y,z)$ is non-hierarchical, 
its dissociation
$q' \datarule R(x), S^y(x,y), T(x,y,z)$ is hierarchical, 
its dissociation 
$q'' \datarule R^z(x, z), S^y(x,y), T(x,y,z)$ is non-hierarchical again.

This suggests the following naive algorithm for computing $\rho(q)$:
Enumerate all dissociations $\Delta_1, \Delta_2, \ldots$ by traversing
the lattice breadth-first, bottom up (i.e.\ whenever $\Delta_i \prec
\Delta_j$ then $i < j$).  For each dissociation $\Delta_i$, check if
$q^{\Delta_i}$ is safe. If so, then first update $\rho \leftarrow \min(\rho,
\PP{q^{\Delta_i})}$, then remove from the list all dissociations
$\Delta_j \succ \Delta_i$.  However, this algorithm is inefficient for
practical purposes for two reasons: ($i$) we need to iterate over
many dissociations in order to discover those that are safe; and
($ii$) computing
$\PP{q^{\Delta_i}}$ requires computing a new database 
$D^{\Delta_i}$ for each hierarchical dissociation $\Delta_i$.
In the next two sections we show how to evaluate the propagation score very efficiently.

\section{Dissociations and minimal query plans}\label{sec:plans}
So far, in order to compute the propagation score of a query $q$, we need to dissociate its tables and compute several dissociated hierarchical queries $q^{\Delta}$. In practice, we will not apply naively query dissociation (\autoref{def:Dissociation}) because creating the dissociated database 
is very inefficient. 
We next show a 1-to-1 correspondence between hierarchical dissociations and query plans
which allows us to calculate $\PP{q^{\Delta}}$ on the original database
(\autoref{sec:queryPlans}), and then present an efficient algorithm for enumerating a minimum number of query plans we need to evaluate (\autoref{sec:enumeratingMinimalQueryPlans}).

\subsection{Hierarchical dissociations and query plans}\label{sec:queryPlans}
We next show ($i$) how to efficiently find hierarchical dissociations (by iterating over
query plans instead of all dissociations), and 
($ii$) how to compute $\PP{q^{\Delta}}$ without
having to materialize the dissociated database $D^{\Delta}$.

\begin{theorem}[Hierarchical dissociation]\label{th:safeDissociation}
For every sf-free CQ,
there is an isomorphism
between the set of query plans
and the set of hierarchical dissociations.
Moreover, the probability of a hierarchical dissociated query $q^{\Delta}$ is equal to the score of the corresponding plan: $\PP{q^{\Delta}} = \score(P_{\Delta})$.
\end{theorem}

We next describe the mappings:
(1) $\Delta \mapsto P_{\Delta}$:
Consider a hierarchical query dissociation
$q^{\Delta}$ and denote its corresponding unique safe plan
$P_{\Delta}$. This plan uses dissociated relations, hence each relation
$R^{\vec y_i}_i(\vec x_i, \vec y_i)$ has extraneous variables
$\vec y_i$.  Drop all variables $\vec y_i$ from the relations and
all operators using them.
Since we only remove existential variables from subgoals, the usual sanity conditions for projections are satisfied and each variable is still projected away in at most one project operator.
This transforms $P_{\Delta}$ into a
regular, generally unsafe plan $P$ for $q$.  For a trivial example,
the plan corresponding to the top dissociation $\Delta_\top$
of a query $q$
(i.e.\ the dissociation at the top of the partial dissociation order with $\vec y_i = \EVar(q) - \vec x_i$)
is $\projp{\EVar(q)} \joinp{}{P_1, \ldots, P_k}$: It
first joins all tables, then projects away all existential variables.

(2) $P \mapsto \Delta_P$: Conversely, consider any plan $P$ for $q$.  
We define its corresponding hierarchical dissociation $\Delta_P$ as follows: 
For each join
operation $\joinp{}{P_1,\ldots,P_k}$, let its \emph{join variables}
$\JVar$ be the union of the head variables of all subplans: $\JVar =
\bigcup_j \HVar(P_j)$.
We go from a plan $P$ to a hierarchical dissociation $\Delta = g(P)$ by recursively dissociating each relation $R_i$ occurring in a subplan $P_j$ of a join operation $\joinp{}{P_1,\ldots,P_k}$ on the missing \emph{existential variables} 
$\big(\JVar - \HVar(P_j)\big) \cap \EVar(P)$. 
Then, recursively and for every relation $R_i$ occurring in $P_j$,
add those variables to $\vec y_i$.\footnote{
Notice that dissociating a table on any head variable has no implication on the probability of a query result as it does not change its lineage.
We therefore only focus on dissociating existential variables.}

\begin{example}[\autoref{ex:partialDissociation} continued]\label{ex:queryPlans}
We saw in \autoref{ex:partialDissociation} that the query $q \datarule R(x), S(x), T(x,y), U(y)$ has 8 dissociations  depicted in \autoref{Fig_PartialDissociationOrderExample_a}. Among those, 5 are hierarchical, and \autoref{Fig_PartialDissociationOrderExample_b} shows their query plans:
\begin{align*}
	P_{\Delta_3}	& = \projp{x} \joinp{}{R(x), S(x), \projp{y} \joinp{}{T(x,y), U(y) }  }  \\
	P_{\Delta_4}	& = \projp{y} \joinp{}{U(y), \projp{x} \joinp{}{R(x),S(x),T(x,y)}   }	\\
	P_{\Delta_5}	& = \projp{x} \joinp{}{S(x), \projp{y} \joinp{}{R(x),T(x,y),U(y)}   }	\\
	P_{\Delta_6}	& = \projp{x} \joinp{}{R(x), \projp{y} \joinp{}{S(x), T(x,y), U(y) }  }  \\	
	P_{\Delta_7}	& = \projp{x,y} \joinp{}{R(x), S(x), T(x,y), U(y)  }  
\end{align*}
As every plan
corresponds to one hierarchical dissociation, the partial dissociation order
carries over to a partial order on all query plans.
The propagation score is thus the minimum of the scores of the two minimal plans:
$\rho(q) = \min_{i \in \{3,4\}}  \big[ \score\big(P_{\Delta_i}\big)  \big]$.
Next consider the subplan $\joinp{}{R(x), T(x,y), U(y)}$ 
in $P_{\Delta_5}$.
Here, $\JVar \cap \EVar(q)=
\set{x,y}$ and the corresponding hierarchical dissociation of this subplan is
$q^\Delta(x,y) \datarule$ $R^y(x, y),T(x,y), U^x(x,y)$.
\markend
\end{example}

Notice that a hierarchical dissociation is different from and does not imply a
safe plan for the original query.  It merely states that the
dissociated query $q^{\Delta}$ allowed a safe plan $P$ assuming
all tuples in its relations were independent. 
Further notice that while
there is a 1-to-1 mapping between hierarchical dissociations and query
plans, non-hierarchical dissociations do not correspond to plans and are still hard
(e.g., dissociations 0, 1, and 2 in \autoref{ex:partialDissociation} and
\autoref{Fig_PartialDissociationOrderExample_a}).

Recall from \autoref{sec:background}
that the extensional semantics of an unsafe plan $P$ differs from the query probability: 
$\score(P) \neq \PP{q}$, in general. 
Since we have previously shown that
$\score(P) = \PP{q^\Delta}$ for some dissociation $\Delta$, we derive the following rather surprising result:

\begin{corollary}[Query plans are upper bounds]\label{cor:queryPlanBounds}
  Let $P$ be any plan for a Boolean query $q$.  Then $\PP{q} \leq
  \score(P)$.
\end{corollary}

\noindent
The proof follows immediately from $\PP{q} \leq \PP{q^{\Delta_P}}$
(\autoref{th:upperBounds}) and $\PP{q^{\Delta_P}} = \score(P)$
(\autoref{th:safeDissociation}).
In other words, any query plan for
$q$ as defined in~\autoref{def:queryPlans}
computes a probability score that is guaranteed to be an upper
bound on the correct  probability $\PP{q}$.

\subsection{Enumerating minimal query plans}\label{sec:enumeratingMinimalQueryPlans}

\autoref{th:safeDissociation} suggests the following improved
algorithm for computing the propagation score $\rho(q)$ of a query:
Iterate over all plans $P$, compute their scores, and retain the
minimum score $\min_P [\score(P)]$. Each plan $P$ is evaluated
directly on the original probabilistic database, and there is no need to
materialize the dissociated database.
However, this approach
is still inefficient because it computes several plans that
correspond to non-minimal dissociations (e.g., plans 5, 6, 7 in \autoref{ex:queryPlans} are ``\emph{dominated}'' by plan 3 since plan 3 is lower in the partial dissociation order).
It thus suffices to evaluate only the \emph{minimal query plans}, i.e.\
those for which the corresponding dissociation is minimal (i.e., not dominated) among all
hierarchical dissociations: in our \autoref{ex:partialDissociation}, these are plans 3 and 4.
We next describe
our third technical result,
the recursive \autoref{alg:basicAlgorithm}  that enumerates only the ``\emph{minimal query plans}''
(i.e.\ those that correspond to minimal hierarchical dissociations) 
and thus the minimum necessary number of query plans to evaluate $\rho(q)$.

We require some additional notation: 
Call a plan $P$ \emph{minimal} if $\Delta_P$ is minimal in the set of
hierarchical dissociations.  For example, in \autoref{ex:partialDissociation},
the minimal plans are $P_{\Delta_3}$ and $P_{\Delta_4}$.  
The propagation score is thus the
minimum of the scores of these two plans: $\rho(q) = \min_{i \in
  \{3,4\}} \big[ \score\big(P_{\Delta_i}\big) \big]$.  Our improved
algorithm will iterate only over minimal plans, by relying on a
connection between plans and sets of variables that disconnect a query:
A ``\emph{cut}'' is a set of existential
variables $\vec x \in \EVar(q)$ s.t.\ $q - {\vec x}$ is disconnected.\footnote{Recall that we say a query is \emph{connected} if all subgoals are connected by
considering only existential variables $\EVar(q)$. 
In other words, when computing query components we remove head variables from the query: $q -
\HVar(q)$.
An
  alternative way to write this is to first substitute all head
  variables by constants $q' = q[\vec a / \vec x]$ (here $q[\vec a /
  \vec x]$ denotes the query obtained by substituting each head
  variable $x_i \in \vec x$ with the constant $a_i \in \vec a$), then
  to let $q_1, \ldots, q_k$ be the components of $q'$ connected by any
  variable. The query is connected if $k=1$, otherwise it is
  disconnected, and $\forall i \neq j: \Var(q_i) \cap \Var(q_j)
  \subseteq \HVar(q)$.}
A ``\emph{min-cut}'' (for \emph{minimal cut}) is a cut for which no strict subset is a cut, i.e.\
no proper subset
$\vec y ' \subset \vec y \in \MinCuts(q)$ can disconnect the query
where $\MinCuts(q)$ denotes the set of all min-cuts.
Note that $\MinCuts(q) = \emptyset$ iff $q$ is
disconnected.

The connection between $\MinCuts(q)$ and query plans is given by two
observations:
(1)  Let $P$ be any plan for $q$. If $q$ is connected, then
the last operator in $P$ is a projection, i.e.\  $P = \projp{\vec x}
\joinp{}{P_1, \ldots, P_k}$, and the variables $\vec x$ projected away 
are the intersection of the join variables 
$\JVar = \bigcup_i \HVar(P_i)$ with existential variables,
as we must project away all existential variables.
We claim that $\vec x$ is a cut
for $q$ and that $q- {\vec x}$ has $k$ query
components corresponding to $P_1, \ldots, P_k$.
Indeed, if $P_i, P_j$
share any common variable $y$, then they must join on $y$, hence $y
\in \JVar$. Thus,
\emph{cuts are
in 1-to-1 correspondence with the top-most project-away operator} of a
plan.
(2) Next suppose that $P$ corresponds to a hierarchical dissociation
$\Delta_P$, and let 
$P' = \projp{\vec x} \joinp{}{P_1', \ldots, P_k'}$ 
be its unique safe plan.
Then
${\vec x} = \SepVar(q^{\Delta_P})$; i.e.\ the top-most project
operator removes all separator variables.\footnote{This
  follows from the recursive definition of the unique safe plan of a query in \autoref{lemma:hierarchical}: the top-most projection consists precisely of its separator variables.}
Furthermore, if $\Delta
\succeq \Delta_P$ is a larger hierarchical dissociation, then
$\SepVar(q^{\Delta}) \supseteq \SepVar(q^{\Delta_P}) $
(because any
separator variable of a query continues to be a separator variable in
any dissociation of that query).  Thus,
\emph{minimal plans correspond to min-cuts}; in other words,
$\MinCuts(q)$ is in 1-to-1 correspondence with the top-most projection
operator of {\em minimal plans}.

\begin{algorithm2e}[t]
\caption{enumerates all minimal query plans for a query $q$.}\label{alg:basicAlgorithm}
\SetKwInput{Algorithm}{Recursive algorithm}
\SetKwFunction{MP}{MP}
\SetKwFor{ForAll}{forall}{do}{endfch}	
\Algorithm{\FuncSty{MP (EnumerateMinimalPlans)}}
\KwIn{Query $q(\vec z) \datarule R_1(\vec x_1), \ldots, R_m(\vec x_m)$}
\KwOut{Set of all minimal query plans $\mathcal{P}$}
\BlankLine
\lIf{$m=1$}{$\mathcal{P} \leftarrow \{\projpd{\vec z} R_1(\vec x_1)\}$}\label{alg1:line1}
\Else{
	Set $\mathcal{P} \leftarrow \emptyset$ \;
	\uIf{$q$ is disconnected}{\label{alg1:disconnected}
		Let $q = q_1, \ldots, q_k$ be the query components  of $q$\;
		\lForEach{$q_i$}{Let $\HVar(q_i) \leftarrow \vec z \cap \Var(q_i)$}
		\ForEach{$(P_1, \ldots, P_k) \in \MP(q_1) \times \cdots \times \MP(q_k)$} {
			$\mathcal{P} \leftarrow \mathcal{P} \cup \{\joinp{}{P_1, \ldots, P_k}\}$ \;
		}
	}
	\label{alg1:connected}\Else{
		\ForEach{$\vec y \in \MinCuts(q)$ \label{alg1:line10}}{
			Let $q' \leftarrow q$ with $\HVar(q') \leftarrow \vec z \cup \vec y$ \;
			\lForEach{$P \in \MP(q')$}{
				$\mathcal{P} \leftarrow \mathcal{P} \cup \{\projp{\vec y} \, P\}$ \label{alg1:line12}
			}
		}
	}
}
\end{algorithm2e}

Our discussion leads immediately to \autoref{alg:basicAlgorithm} for computing the
propagation score $\rho(q)$.
The algorithm proceeds recursively: 
If $q$ is a single atom (\autoref{alg1:line1}), then it is safe and we only need to project on the head variables.\footnote{Note that if there are no existential variables ($\vec z = \vec x_i$), then there is no need for the projection operator and one could instead simplify to $\mathcal{P} \leftarrow \{R_i(\vec z)\}$, instead of $\mathcal{P} \leftarrow \{\projpd{\vec z} R_i(\vec x_i)\}$.}
If the query has more than one atom, then we consider two cases depending on whether the query is
connected.  
If the query is disconnected (\autoref{alg1:disconnected}), then the algorithm recursively computes the minimal subplans for each query component, then creates a query plan for each combination of those subplans. 
If the query is connected (\autoref{alg1:connected}), the algorithm creates a separate plan for each min-cut $\vec y \in MinCuts(q)$ by moving $\vec y$ from the existential variables $\EVar(q)$ to the head variables $\HVar(q)$, thereby disconnecting the query.
Notice that recursive calls of the algorithm will alternate between
these two cases, until they reach a single atom.

\begin{figure}[t]
\centering
\renewcommand{\tabcolsep}{0.9mm}
\renewcommand{\arraystretch}{0.95}
\subfloat[$q$]{
	\label{Fig_ExampleHierarchies_a}
	\begin{tabular}[t]{@{\hspace{1pt}} >{$}c<{$}|>{$}c<{$} | >{$}c<{$} >{$}c<{$} >{$}c<{$} >{$}c<{$} @{\hspace{1pt}}}
			& v    	& x		& y 	& z 	& u 	 		\\
					\hline
		R	&		&\circ	& \circ & 	 	& \circ	 	\\
		S	&\circ	&  		& \circ & \circ & \circ  	\\
		T	&\circ	& 		& 		& \circ &
	\end{tabular}}
\hspace{7mm}
\subfloat[$q^{\Delta_1}$]{
	\label{Fig_ExampleHierarchies_b}
	\begin{tabular}[t]{@{\hspace{1pt}} >{$}c<{$}|>{$}c<{$} | >{$}c<{$} >{$}c<{$} >{$}c<{$} >{$}c<{$} @{\hspace{1pt}}}
			& v    	& z	 			& y 		& u 		& x 		 		\\
			\hline
		R	&		&\graycell\bullet	&\graycell\circ &\graycell\circ	&\graycell\circ	 	\\
		S	&\circ	&\graycell\circ 	&\graycell\circ &\graycell\circ &  			 	\\
		T	&\circ	&\graycell\circ 	& 			& 			&
	\end{tabular}}
\hspace{7mm}
\subfloat[$q^{\Delta_2}$]{
	\label{Fig_ExampleHierarchies_c}
	\begin{tabular}[t]{@{\hspace{1pt}} >{$}c<{$}|>{$}c<{$} | >{$}c<{$} >{$}c<{$} >{$}c<{$} >{$}c<{$} @{\hspace{1pt}}}
			& v    	& y 			& u				& x	 		& z  		 	\\
			\hline
		R	&		&\graycell\circ 	&\graycell\circ	 	&\graycell\circ	&  			  	\\
		S	&\circ	&\graycell\circ 	&\graycell\circ  	&  			&\graycell\circ   	\\
		T	&\circ	&\graycell\bullet	&\graycell\bullet	& 			&\graycell\circ
	\end{tabular}}
\caption{\autoref{ex:topVariables}. Query $q$ and its two minimal hierarchical dissociations. 
Notice the hierarchies between $\EVar(q)$ for both dissociations.}
\label{Fig_ExampleHierarchies}
\end{figure}

\begin{figure*}[t]
    \centering
	\vspace{-2mm}
	\subfloat[Probability $\PPP$, dissociation $\rho$]
       	{\includegraphics[scale=0.41]{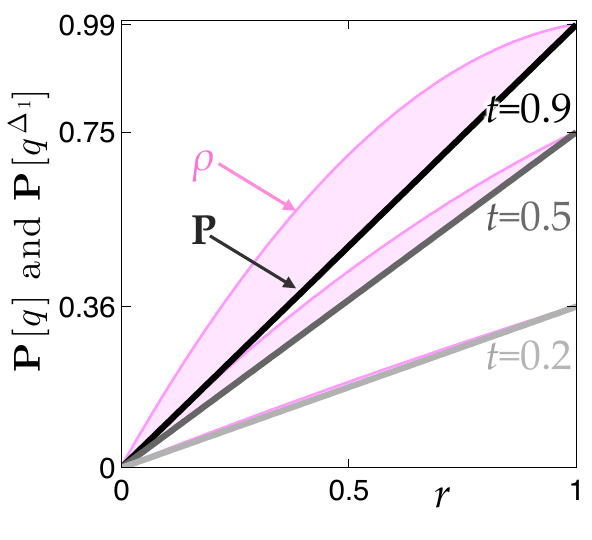}
    	\label{Fig_VLDBJ_SyntheticDissociationBounds1}}
	\hspace{3mm}
	\subfloat[Relative ratio $\rho/\PPP$]
		{\includegraphics[scale=0.41]{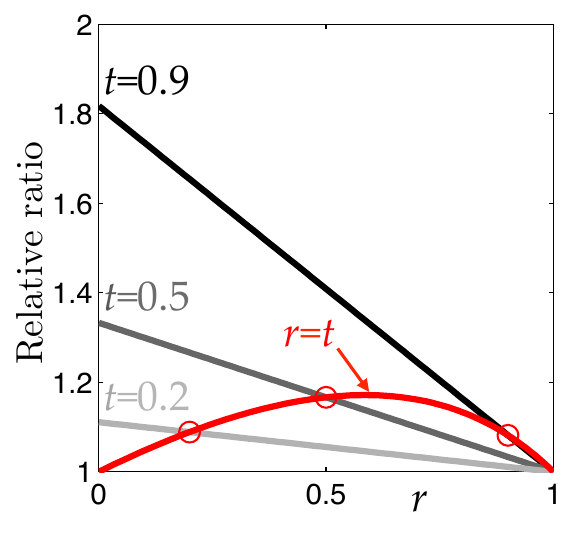}
		\label{Fig_VLDBJ_SyntheticDissociationBoundsRatio1}}
	\hspace{3mm}
	\subfloat[Relative not-ratio $\bar \PPP /\bar \rho$]
		{\includegraphics[scale=0.41]{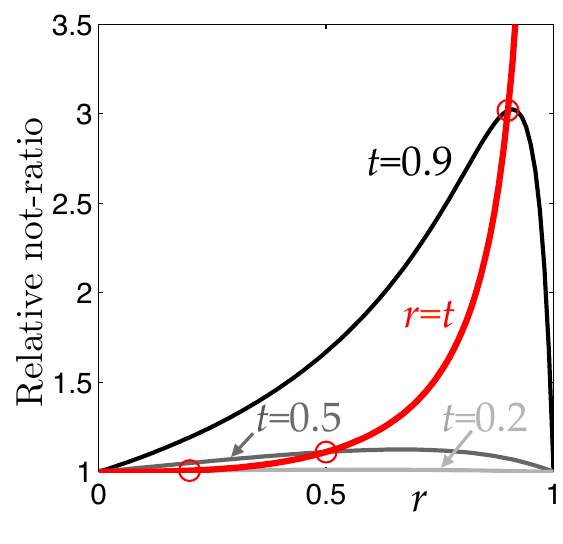}
		\label{Fig_VLDBJ_SyntheticDissociationBoundsRatio2}}
	\hspace{3mm}
	\subfloat[Odds ratio 
	$(\rho \bar \PPP) / (\bar \rho \PPP)$
	]
		{\includegraphics[scale=0.41]{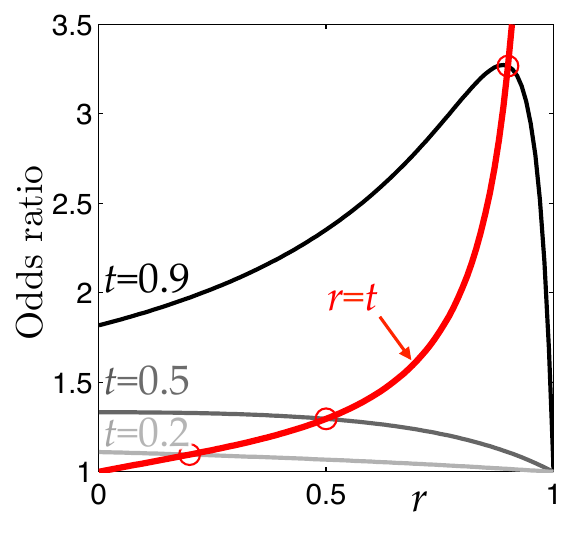}
		\label{Fig_VLDBJ_SyntheticDissociationBoundsRatio3}}
	\caption{\autoref{ex:smallProbabilities}: Comparing 
	$\PPP=\PP{q}=\PP{r_1 t_1 \vee r_1 t_2}$ 
	and 
	$\rho=\PP{q^{\Delta}}=\PP{r_1 t_1 \vee r_1' t_2}$ 
	for varying input probabilities
	$p = \PP{r_1}$ and $q = \PP{t_i}$.
	}
    \label{Fig_DissociationBounds}
\end{figure*}

\begin{theorem}[\Autoref{alg:basicAlgorithm}]\label{prop:algorithmSound}
\Autoref{alg:basicAlgorithm} returns a sound and complete enumeration of minimal query plans. 
\end{theorem}

\noindent
\Autoref{alg:basicAlgorithm} is sound in that only plans are generated which are not dominated by any other plan. It is complete in that the minimum score of all generated plans is equal to the propagation score of the query.

\begin{example}[Enumerate minimal query plans]\label{ex:topVariables}
Consider the non-Boolean query $q(v) \datarule R(x,y,u), S(y,z,u,v), T(z,v)$. 
\Autoref{Fig_ExampleHierarchies_a} shows its incidence matrix with its head variables $\HVar(q) = \{v\}$ and existential variables $\EVar(q) = \{x,y,z,u\}$ shown separately. 
The query is connected, and among $2^4$ different subsets of $\EVar(q)$, there are 2 ``\emph{min-cuts}'', i.e.\  minimum sets of variables for which removing them disconnects the query: $\MinCuts(q) = \{ \{ z\}, \{y,u \} \}$. 
Projecting away the first min-cut $\{z\}$ separates the query into $q_1(z,v) \datarule R(x,y,u),$ $S(y,z,u,v)$ and $q_2(z,v) \datarule T(z,v)$. 
Notice that $q_1$ and $q_2$ share no existential variables (they only share head variables $z$ and $v$). 
Projecting away the second min-cut $\{y,u \}$ separates $q$ into $q_3(y,u) \datarule R(x,y,u)$ and 
$q_4(y,u,v) \datarule $ $S(y,z,u,v), T(z,v)$. Recursive evaluation of $q_1$ to $q_4$ shows that they are all hierarchical, from which follows that $q$ has 2 minimal query plans:
\begin{align*}
	P_{\Delta_1}	& = \projp{z} \joinp{}{
				\projp{y,u} \joinp{}{
				\projp{x} R(x,y,u), S(y,z,u,v)}, T(z,v)}  \\
	P_{\Delta_2}	& = \projp{y,u} \joinp{}{
				\projp{x} R(x,y,u),
				\projp{z} \joinp{}{S(y,z,u,v), T(z,v)}}
\end{align*}
\Autoref{Fig_ExampleHierarchies_b} and \autoref{Fig_ExampleHierarchies_c} show their respective hierarchical dissociations, with existential variables re-ordered as to show the hierarchy implied by the query plans.
\markend
\end{example}

\subsection{Other observations}\label{sec:4:other}

\introparagraph{(1) Conservation} 
Some probabilistic database systems first check
if a query $q$ is safe, and in that case compute the exact probability
using the safe plan, otherwise use some approximation technique.
We show that \autoref{alg:basicAlgorithm}
is \emph{conservative}, in the sense that, if $q$ is safe,
then $\rho(q) = \PP{q}$.
Indeed, in that
case $\MP(q)$ returns a single plan, namely the safe $P$ for $q$,
because the empty dissociation, $\Delta_\bot = (\emptyset, \ldots,
\emptyset)$, is safe, and it is the bottom of the dissociation
lattice, making it the unique minimal hierarchical dissociation.

\smallsection{(2) Approximation quality}
We next show that 
the relative error of approximating query reliability with dissociation improves when the input probabilities decrease.   
As a consequence, the quality of the ranking also increases.
As a practical consequence, the rankings returned by dissociation are better if the input probabilities are small.

\begin{proposition}[Small probabilities]\label{prop:smallProbabilities}
Given a query $q$ and database $D$.
Consider the operation of scaling down the probabilities of all tuples in $D$ with a positive factor $f<1$. Then the relative error of approximation of the query probability $\PP{q}$ by the propagation score $\rho(q)$ decreases as $f$ goes to 0:
$\lim_{f \rightarrow 0^+} \frac{\rho(q) - \PP{q}}{\PP{q}} \rightarrow 0$.
\end{proposition}

In the following analytic example, we illustrate \autoref{prop:smallProbabilities} by calculating the relative ratio between propagation and reliability for changing input probabilities.

\begin{example}[Small probabilities]\label{ex:smallProbabilities}
We consider the Boolean query
$q \datarule R(x), S^d(x,y), T(y,z)$ and the dissociation
$q^{\Delta} \datarule $ $R^y(x, y), S^d(x,y), T(y,z)$
over the database $r_1 = R(a)$, $s_1=S^d(a,b)$, $s_2=S^d(a,c)$, $t_1 = T(b)$, and $t_2 = T(c)$. With deterministic relation $S^d$, the lineages of $q$ and $q^{\Delta}$ are
$\Lineage(q) = r_1 t_1 \vee r_1 t_2$ and $\Lineage(q^{\Delta}) = r_1 t_1 \vee r_1' t_2$, respectively.
Assuming $\PP{r_1}=r$ and $\PP{t_1}=\PP{t_2}=t$, the respective probabilities become
$\PP{q} = r(t\otimes t) = rt(2-t)$ and
$\rho(q):=\PP{q^{\Delta}} = rt \otimes rt = rt(2-rt)$
with $r_1$ dissociated.

There are four reasonable metrics to measure the approximation quality of dissociation
$\rho$
with regard to a probability $\PPP$:
(1) their \emph{absolute difference} $\rho-\PPP$,
which is not meaningful when both are too close to either 0 or 1;
(2)~their \emph{relative ratio} 
$\rho / \PPP$,
which is not meaningful close to 1;
(3)~their \emph{relative not-ratio} 
$\bar \PPP / \bar \rho$
with $\bar x \define 1-x$,
which is not meaningful close to 0; and
(4) the \emph{odds ratio} 
$(\rho / \bar \rho) / (\PPP / \bar \PPP)$,
which is the product of the former two ratios and which is meaningful everywhere in $[0,1]$.
Notice that all four metrics are defined so they are $\geq 1$.
\Autoref{Fig_VLDBJ_SyntheticDissociationBounds1}
shows the original probabilities $\PPP$ (full lines) and those
of their dissociations $\rho$ (border of shaded areas) for various values of $r$ and $t$.
The horizontal axis varies the probability of the dissociated tuple $x$ within $[0,1]$, and the different lines keep the non-dissociated tuples $y_1$, $y_2$ at the same probability either $0.2$, $0.5$, or $0.9$.
\autoref{Fig_VLDBJ_SyntheticDissociationBoundsRatio1},
\autoref{Fig_VLDBJ_SyntheticDissociationBoundsRatio2}, and
\autoref{Fig_VLDBJ_SyntheticDissociationBoundsRatio3}
show the approximation quality in terms of our three previously defined ratios.
Notice that the red line varies both $r$ and $t$ at the same time by keeping $r=t$.
We see that the approximation is good when both input probabilities are small, but get increasingly worse when the probability of the \emph{non-dissociated} variables $t_1, t_2$ gets close to 1.

Finally notice that the \emph{relative error} 
$(\rho-\PPP)/\PPP 
= {\rho}/{\PPP} - 1 
= \frac{t(1-r)}{2-t}$,
which clearly tends towards 0 as $r,t \rightarrow 0$.
\markend
\end{example}

\begin{figure}[t]
\centering
\small
\renewcommand{\tabcolsep}{0.9mm}
\renewcommand{\arraystretch}{0.95}
	\begin{tabular}[t]{@{\hspace{1pt}} >{$}c<{$} | >{$}r<{$} >{$}r<{$} >{$}r<{$} ||
		>{$}r<{$} | >{$}r<{$} >{$}r<{$} >{$}r<{$} @{\hspace{1pt}}}
		 	\multicolumn{4}{c || }{$k$-star query}	&\multicolumn{4}{c }{$k$-chain query}\\
		k 	&\multicolumn{1}{c}{\#MP}	&\multicolumn{1}{c}{\#P}	&\multicolumn{1}{c||}{$\#\Delta$}
		&k	&\multicolumn{1}{c}{\#MP}	&\multicolumn{1}{c}{\#P}	&\multicolumn{1}{c}{$\#\Delta$}	\\
		\hline
		1	&1 		&1		 &1			&2  &1  		&1		&1		   \\
		2	&2 		&3		 &4			&3  &2  		&3		&4		   \\
		3	&6 		&13		 &64		&4  &5  		&11		&64	       \\
		4	&24 	&75		 &4096		&5  &14  		&45		&4096	   \\
		5	&120	&541	 &>10^6		&6  &42  		&197	&>10^6	   \\
		6	&720	&4683	 &>10^9		&7  &132  		&903	&>10^9	   \\
		7	&5040	&47293	 &>10^{12}	&8  &429  		&4279	&>10^{12}   \\
		\hline
		\textrm{seq}	&k!		&\href{http://oeis.org/classic/A000670}{A000670} 	 &2^{k(k\!-\!1)}
		& \textrm{seq}
			&\href{http://oeis.org/classic/A000108}{A000108}
			&\href{http://oeis.org/classic/A001003}{A001003}	&2^{(k\!+\!1)k}	   \\
	\end{tabular}
\caption{Number of minimal plans, total plans, and total dissociations for star and chain queries (A are  OEIS sequence numbers~\protect\cite{oeis}).}
\label{table:numberMinimalQueryPlans}
\end{figure}

\smallsection{(3) Number of minimal query plans}
We end this section by commenting on the number of minimal hierarchical dissociations.  Not surprisingly, this number is exponential in the size of the query.  To see this, consider a Boolean $k$-\emph{star query}\footnote{A Boolean \emph{conjunctive $k$-star query} is a query with $k$ unary relations and one $k$-ary relation: $q \datarule R_1(x_1), \ldots, R_k(x_k), U(x_1, \ldots, x_k)$. The fact that each variable appears in exactly two relations implies that the dual query hypergraph is actually a standard graph 
(the dual hypergraph of a query is defined by the relations as vertices and variables as the hyperedges).
The expression \emph{star query} derives from the observation that the query's dual (hyper)graph resembles a star with the table $U$ connected to all other relations.}
$q \datarule R_1(x_1), \ldots, R_k(x_k), U(x_1, \ldots, x_k)$. There are exactly $k!$ minimal hierarchical dissociations: Take any consistent preorder $\preceq$ on the variables.  It must be a total preorder, i.e.\ for any $i, j$, either $x_i \preceq x_j$ or $x_j \preceq x_i$, because $x_i, x_j$ occur together in $U$.  Since it is minimal, $\preceq$ must be an order, i.e.\ we can't have both $x_i \preceq x_j$ and $x_j \preceq x_i$ for $i \neq j$.  Therefore, $\preceq$ is a total order, and there are $k!$ such.
Note that while the number of hierarchical dissociations is exponential in the size of the query, the number of query plans is independent of the size of the database, and hence our approach has PTIME data complexity~\cite{DBLP:conf/stoc/Vardi82} for all queries.
\Autoref{table:numberMinimalQueryPlans} gives an overview of the number of minimal query plans, total query plans, and dissociations for star and chain queries. 
Recall that in our definition of query plans, we \emph{do not consider permutations in the joins} (called join orderings~\cite{Moerkotte:BuildingQueryCompilers}). Also, our problem differs from the standard problem of optimal join enumeration in relational database engines. For example, every safe query has only one single minimal query plan, whereas any relational database engine compares several query plans.
Later \autoref{sec:optimizations} gives optimizations that allow us to evaluate a large number of plans efficiently.

\vspace{5mm}
In summary, our approach allows to rank answers to both safe and unsafe queries in polynomial time in the size of the database, and is conservative w.r.t.\ the ranking according to exact probabilistic inference for both safe queries and for data-safe queries~\cite{JhaOS2010:EDBT}. 
The latter follows easily from the point that if a query over a particular database allows one single safe plan, then this plan must be among the minimal plans in the partial dissociation order.

\section{Optimizations with schema knowledge}\label{sec:optimizationsWithSchema}
In this section, we show how \emph{deterministic relations} (i.e.\ all tuples in a relation have
probability $1$), and \emph{functional dependencies} (e.g., keys) can reduce the number of plans needed to
calculate the propagation score.

\subsection{Deterministic relations (DRs)}\label{sec:deterministicOptimization}
In the following, we denote deterministic relations (DRs) with an exponent ``$d$'', i.e.\ a relation $R$ is probabilistic, and a relation $R^d$ is deterministic. 
First notice that we can treat DRs just like probabilistic
relations, and \autoref{cor:queryPlanBounds} with $\PP{q} \leq score(P)$ still holds for any
plan $P$. Just as before, our goal is to find a minimum number of plans that compute the
minimal score of \emph{all plans}: $\rho(q) = \min_P score(P)$. It is known that a non-hierarchical query
$q$ can become safe (i.e., $\PP{q}$ can be calculated in PTIME with one single plan) if we
consider DRs. Thus, we would still like an improved algorithm that returns one
single plan if a query with DRs is safe. The following lemma will help us achieve this goal:

\begin{lemma}[Dissociation and DRs]\label{lemma:DetDissociation}
Dissociating a deterministic relation does not change the probability.
\end{lemma}

We thus define a new \emph{probabilistic dissociation preorder} $\preceq^p$ that only focuses on probabilistic relations:
\begin{align*}
  \Delta &\preceq^p \Delta' \Leftrightarrow \forall i \textrm{ with } R_i \mbox{ probabilistic}: {\vec y}_i \subseteq {\vec y}_i'
\end{align*}

\noindent 
In other words, $\Delta \preceq^p \Delta'$ still implies $\PP{q^\Delta} \leq
\PP{q^{\Delta'}}$, but $\preceq^p$ is defined on probabilistic relations only. Notice, that
for queries without DRs, the relations $\preceq^p$ and $\preceq$ coincide. However, for queries
with DRs, $\preceq^p$ is a preorder, not an order. Therefore, there exist distinct
dissociations $\Delta$, $\Delta'$ that are equivalent under $\preceq^p$ (written as $\Delta
\equiv^p \Delta'$), and thus have the same probability: $\PP{q^\Delta} = \PP{q^{\Delta'}}$.
As a consequence, using $\preceq^p$ instead of $\preceq$, allows us to further reduce the
number of minimal hierarchical dissociations we need to evaluate.

\begin{figure}[t]
	\vspace{-3mm}
    \centering
	\subfloat[No DRs]{
		\label{Fig_PartialDissociationOrderSimple_a}
		\includegraphics[scale=0.5]{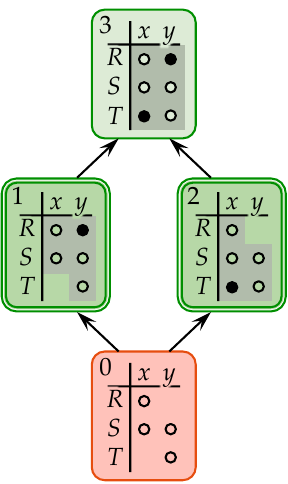}}
	\hspace{2mm}	
	\subfloat[{$T^d$}]{
		\includegraphics[scale=0.5]{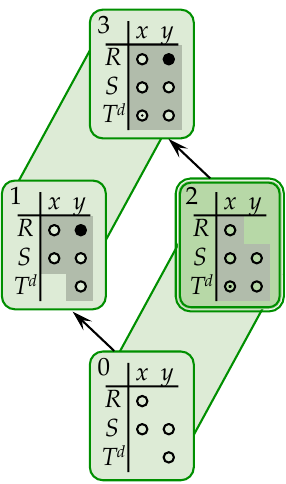}
		\label{Fig_PartialDissociationOrderSimple_b}}
	\hspace{2mm}	
	\subfloat[$R^d$ and $T^d$]{
		\includegraphics[scale=0.5]{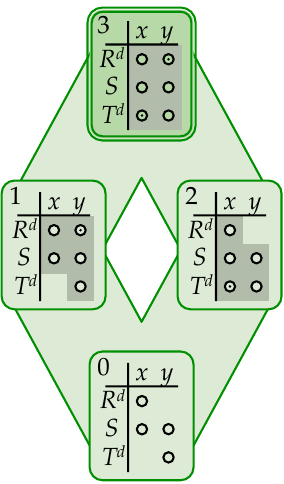}
		\label{Fig_PartialDissociationOrderSimple_c}}		
	\caption{\Autoref{ex:DRsSimple}:
The presence of DRs $R^d$ or $T^d$ changes the \emph{probabilistic dissociation preorder} for $q\datarule R(x),S(x,y),T(y)$:
Several dissociations now have the same probability (shown with shaded areas instead of arrows). Our modified algorithm now returns, for \emph{each minimal safe equivalence class}, the query plan for the top most hierarchical dissociation (shown in dark green and double-lined).}
	\label{Fig_PartialDissociationOrderSimple}
	\vspace{-1mm}
\end{figure}

\begin{example}[RST query with DRs]\label{ex:DRsSimple}
Consider the query $q \datarule $ $R(x),S(x,y),T^d(y)$. 
This query is known to be safe.
We thus expect our
definition of $\rho(q)$ to find that $\rho(q) = \PP{q}$.  
Ignoring that $T^d$ is deterministic, $\preceq$ has two minimal plans corresponding to dissociations
$q^{\Delta_1} \datarule R^{\{y\}}(x,y),S(x,y),T^d(y)$, and 
$q^{\Delta_2} \datarule R(x),$ $S(x,y),T^{d \{x\}}(x,y)$.
Since $\Delta_2$ dissociates only $T^d$, we now know from \autoref{lemma:DetDissociation} that $\PP{q} = \PP{q^{\Delta_2}}$. 
Thus, by using $\preceq$ as before,
and ignoring information about DRs,
we still get the correct answer.
However, evaluating the plan $P_{\Delta_1}$ 
is always unnecessary since
$\Delta_2 \preceq^p \Delta_1$. 
In contrast, without information about DRs, $\Delta_2 \not \preceq^p \Delta_1$, and we would thus have to evaluate both plans. 

\indent
\Autoref{Fig_PartialDissociationOrderSimple} illustrates these ideas with incidence matrices that are augmented in a 5th way:
dissociated variables in DRs do not affect the probability and are now marked with dotted circles ($\dotcirc$) instead of full circles ($\bullet$).
Thus, the preorder $\preceq^p$ is determined entirely by full circles (representing dissociated variables in probabilistic relations). 
However, as before, the correspondence to plans (as implied by the hierarchy between all variables) is still determined by all circles.
\Autoref{Fig_PartialDissociationOrderSimple_b} shows that 
$\rho(q) = \PP{q^{\Delta_2}} = \PP{q}$
since
$\Delta_0 \equiv^p \Delta_2 \preceq^p \Delta_1 \equiv \Delta_3$ 
(equivalence under $\preceq^p$ is shown with former arrows being replaced by broad connectors).
Thus, the query is safe, and it suffices to evaluate only $P_{\Delta_2}$.
Notice that $q$ is \emph{not hierarchical}, but still \emph{safe} since it is in an \emph{equivalence class} with a query that is hierarchical: $\Delta_0 \equiv^p \Delta_2$.

\Autoref{Fig_PartialDissociationOrderSimple_c} shows that for $R^d$ and $T^d$ being deterministic, all three possible query plans (corresponding to $\Delta_1$,  $\Delta_2$, and  $\Delta_3$) form an equivalence class in $\preceq^p$ with $\Delta^0$, and thus give the exact probability.
In other words, the number of hierarchical dissociations ``minimal in $\preceq^p$'' has increased to 3, 
but all of them are now in the same equivalence class and thus have the same probability.
We, therefore, want to modify our algorithm to return just one plan from each ``\emph{minimal safe equivalence class}''. Ideally, we prefer the plan corresponding to $\Delta_3$ (or more generally, the plan for the top hierarchical dissociation in $\preceq$ 
for each minimum safe equivalence class)
since $P_{\Delta_3}$ least constrains the join order between tables: compare
$P_{\Delta_3} = \projp{x,y} \joinp{}{ R(x), S(x,y), T^d(y)}$ with
$P_{\Delta_2} = \projp{x} \joinp{}{ R(x), \projp{y} \joinp{}{S(x,y), T^d(y)}}$.
\markend
\end{example}

We now explain two simple modifications to \autoref{alg:basicAlgorithm} that achieve our desired optimizations described above:
\begin{enumerate}[nolistsep,label=(\arabic*)]

\item 
Let a ``\emph{p-cut}'' be a set of existential variables $\vec x \in \EVar(q)$ s.t.\ $q - {\vec x}$ has at least two query components, each of which has at least one probabilistic table.
Denote by $\MinPCuts(q)$ the set of all 
``\emph{minimal p-cuts}'' and replace $\MinCuts(q)$ with $\MinPCuts(q)$ in \autoref{alg1:line10}.

\item
Denote with $m_p$ the number of probabilistic relations in a query, and
w.l.o.g.\ order the relations in a query $q$ as to first list the probabilistic relations, followed by DRs.
Replace the stopping condition
in \autoref{alg1:line1} with:
\textbf{if~}{$m^p \leq 1$} \textbf{then~}{$\mathcal{P} \leftarrow \{\projpd{\vec x} \joinp{}{R_1(\vec x_1), R_2^d(\vec x_2), \ldots R_m^d(\vec x_m)}\}$}. In other words, if a query has maximal one probabilistic relation, then first join all relations, then project on the head variables.
\end{enumerate}

\begin{theorem}[\Autoref{alg:basicAlgorithm} with DRs]\label{prop:algorithm2}
Above two modifications to \autoref{alg:basicAlgorithm} return one plan 
for each minimal safe equivalence class in ${\preceq^p}$, i.e.\
it returns a minimum number of plans to calculate $\rho(q)$ given schema knowledge about deterministic relations.
\end{theorem}

\begin{example}[\autoref{ex:DRsSimple} continued]
For our simple query
$q\datarule R(x),S(x,y),T^d(y)$, $\MinCuts(q) = \{\{x\}, \{y\}\}$, while $\MinPCuts(q) = \{\{x\}\}$. 
Therefore, the modified algorithm returns $P_{\Delta_2}$ as single plan.
For $q\datarule R^d(x),S(x,y),T^d(y)$, the stopping condition is reached (also, $\MinPCuts(q) = \emptyset $) and the algorithm returns $P_{\Delta_3}$ as single plan (\autoref{Fig_PartialDissociationOrderSimple_c}).
\markend
\end{example}

Note that as before, if the query is safe, then the algorithm produces one single query plan.
Furthermore, if all relations are deterministic, then the returned query plan consists of one multi-join between all relations followed by a single projection:
$\projpd{\vec x} \joinp{}{R(\vec x_i), \ldots, R(\vec x_m)}$. The translation into SQL is thus one single standard deterministic SQL query and the query optimizer is unconstrained to determine the optimal join order between the relations. 
Therefore, \emph{our algorithm conservatively extends deterministic SQL queries to probabilistic SQL queries} in that fully deterministic queries are evaluated exactly as deterministic SQL.

\begin{figure}[t]
	\vspace{-3mm}
    \centering
	\subfloat[$U^d$]{
		\includegraphics[scale=0.5]{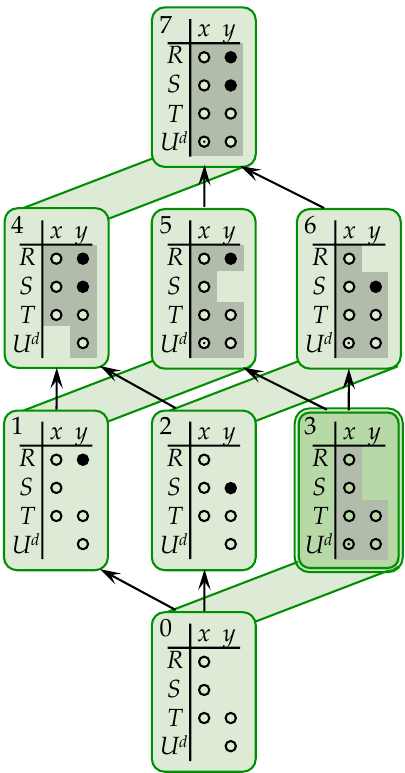}
		\label{Fig_PartialDissociationOrderExample_c}}		
	\hspace{8mm}	
	\subfloat[{$S^d$}]{
		\includegraphics[scale=0.5]{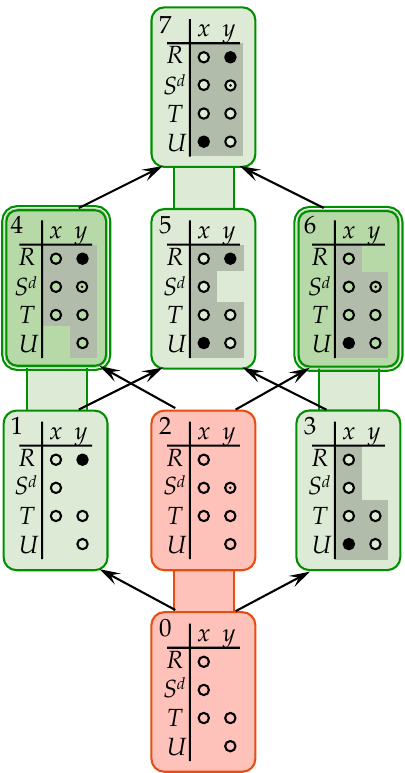}
		\label{Fig_PartialDissociationOrderExample_d}}
	\caption{\Autoref{ex:DRsComplicated} (\autoref{ex:queryPlans} continued):
The presence of DRs (either $U^d$ or $S^d$) changes the probabilistic dissociation preorder and thus the minimal plans returned by our algorithm: $P_{\Delta_3}$, or $P_{\Delta_4}$ and $P_{\Delta_6}$.}
	\label{Fig_PartialDissociationOrderExampleDT}
	\vspace{-1mm}
\end{figure}

\begin{example}[\autoref{ex:queryPlans} continued]\label{ex:DRsComplicated}
First consider relation $U$ to be deterministic: 
$q_1 \datarule R(x), S(x), T(x,y), U^d(y)$. 
\Autoref{Fig_PartialDissociationOrderExample_c}
shows that 
$\rho(q) = \PP{q^{\Delta_3}} = \PP{q}$
and thus the query is safe.
Put differently, there is only one \emph{minimal safe equivalence class} 
$\Delta_0 \equiv^p \Delta_3$, and 
our modified algorithm returns $P_{\Delta_3}$ as single minimal plan. 
Next consider $S$ to be deterministic: 
$q_2 \datarule R(x), S^d(x), T(x,y), U(y)$. 
The query is now in an equivalence class with $\Delta_2$ (\autoref{Fig_PartialDissociationOrderExample_d}). However, neither $\Delta_0$ nor $\Delta_2$ is hierarchical and thus the query is hard. Also, $\Delta_3$ is in a minimal safe equivalence class with $\Delta_6$, the latter of which has fewer constraints on the joins. Thus, the algorithm returns $P_{\Delta_4}$ and $P_{\Delta_6}$ as the least constrained plans, one from each minimum safe equivalence class.
\markend
\end{example}

We end with a short commend on actual implementation in SQL: In practice, deterministic relations do not have a probabilistic attribute, which simplifies the calculations. 
Consider a plan $P=\joinp{x,z}{T^d(z), R^d(x,z), M(x,y,z,u)}$. This query is specified as join with standard semantics $T(x,y,z,u,p) \datarule T(z), R(x,z), M(x,y,z,u,p)$ over the input relations with $p$ as the probability attribute.

\subsection{Functional dependencies (FDs)}\label{sec:FDOptimization}

Knowledge of functional dependencies (FDs), such as keys, can also restrict
the number of necessary minimal plans.
A well known example is the query $q \datarule R(x), S(x,y), T(y)$ from \autoref{ex:DRsSimple}: It becomes safe 
if we know that $S$ satisfies the FD $\Gamma: x \rightarrow y$ and has a unique safe plan that corresponds to dissociation $\Delta_2$. In other words, we would like our modified algorithm to take $\Gamma$ into account and to not return the plan corresponding to dissociation $\Delta_1$.

Let $\bm{\upGamma}$ be the set of FDs on $\Var(q)$ 
consisting of the union 
of FDs on every atom
$R_i$ in $q$.  
As usual, denote ${\vec x}_i^+$ the ``\emph{closure}'' of a set of
attributes ${\vec x}_i$ under $\bm{\upGamma}$,
i.e.\ ${\vec x}_i^+$  is the smallest set of variables that contains $\vec x_i$, and contains $z$ whenever it contains $y$ and 
$y \leftarrow z$ is a FD in $\bm{\upGamma}$.\footnote{E.g., if $\vec x = \{y \}$ and $\bm{\upGamma} = \{ x \rightarrow y, y \rightarrow z, z \rightarrow u \}$, then ${\vec x}^+ = \{y, z, u \}$.}
Then we show:

\begin{lemma}[Dissociation and FDs]\label{lemma:FDs1}
Dissociating an atom $R_i (\vec x_i)$ 
on any variable 
$y \in ({\vec x}_i^+  \!- {\vec x}_i)$ does not change the probability of the query.
\end{lemma}

\noindent
In other words, dissociating a table on a variable that is functionally dependent on the existing variables does not change the probability. 
This lemma is similar to \autoref{lemma:DetDissociation} for DRs. 
Next, for any query $q$, denote $q^+$ the query where each atom $R_i(\vec x_i)$ 
is replaced with $R_i(\vec x_i^+)$, 
and call $q^+$ the closure of $q$.
Call a query ``\emph{closed}'' if $q^+ = q$, and
call a dissociation $\Delta$ closed if $q^\Delta$ is closed,  
i.e.\ $(\vec x_i \cup \vec y_i)^+ = (\vec x_i \cup \vec y_i)$.  
We can then further refine our \emph{probabilistic dissociation preorder} ${{\preceq^p}}$ by:
\begin{align*}
  	\Delta {\preceq^p} \Delta' \Leftrightarrow \forall i \textrm{ with } R_i \textrm{ probabilistic}: 
	({\vec x}_i \cup \vec y_i)^+ \subseteq ({\vec x}_i \cup {\vec y}_i')^+
\end{align*}

\noindent
In other words, we only need to consider closed dissociations.
As a consequence, using our refined definition of ${\preceq^p}$ allows us to further reduce the number of \emph{minimal safe equivalence classes}. 
We next state a result by \cite{DBLP:conf/icde/OlteanuHK09} in our notation:
\begin{proposition}[Safety and FDs~\protect{\cite[Prop.~IV.5]{DBLP:conf/icde/OlteanuHK09}}] A query $q$
  is safe 
  under FDs $\bm{\upGamma}$
  iff $q^+$ is hierarchical.
\end{proposition}

This justifies our third modification to \autoref{alg:basicAlgorithm} for enumerating the minimum number of plans for computing 
$\rho(q)$ over a database that
satisfies FDs $\bm{\upGamma}$ and has DRs: 

\begin{enumerate}[nolistsep,label=(\arabic*)]
\setcounter{enumi}{2}		

\item 
At each recursive call of \autoref{alg:basicAlgorithm}, just before \autoref{alg1:line1}, replace $q$ with its closure $q^+$.
\end{enumerate}

\begin{theorem}[\autoref{alg:basicAlgorithm} with DRs and FDs]\label{prop:allQueryPlansFDs}
Above three modifications to \autoref{alg:basicAlgorithm} return one plan for each minimal safe equivalence class in ${\preceq^p}$, i.e.\
it returns a minimum number of plans to calculate $\rho(q)$ in the presence of deterministic relations and functional dependencies.
\end{theorem}

It is easy to see that our modified algorithm returns one single
plan iff the query is safe (taking into account its structure,
DRs and FDs). 
It is thus a
\emph{strict generalization of all known safe self-join-free
  conjunctive
  queries}~\cite{DBLP:journals/vldb/DalviS07,DBLP:conf/icde/OlteanuHK09}.
In particular, we can reformulate the known safe query dichotomy~\cite{DBLP:journals/vldb/DalviS07} in our notation very succinctly:

\begin{corollary}[Dichotomy]\label{cor:dichotomy}
	$\PP{q}$ is in PTIME iff there exists a dissociation of $q^+$ 
	that 
	is hierarchical and that 
	dissociates only deterministic relations.
	In particular, if all relations are probabilistic then $\PP{q}$ is in PTIME iff $q^+$ is hierarchical.	
\end{corollary}

\begin{corollary}[Dichotomy in plans]
$\PP{q}$ can be calculated in PTIME iff our modified algorithm returns one single plan.
\end{corollary}

\noindent
To see what \autoref{cor:dichotomy} says, assume first that there are no FDs: 
Then $q$ is in PTIME iff there exists a hierarchical dissociation $\Delta$ that dissociates only DRs.
If there are FDs, then we first compute the closure $q^+$ (called ``\emph{full chase}'' in \cite{DBLP:conf/icde/OlteanuHK09}), then apply the same criterion to $q^+$.

\begin{figure}[t]
	\vspace{-3mm}
    \centering
	\subfloat[$\Gamma: y \rightarrow x$]{
		\includegraphics[scale=0.5]{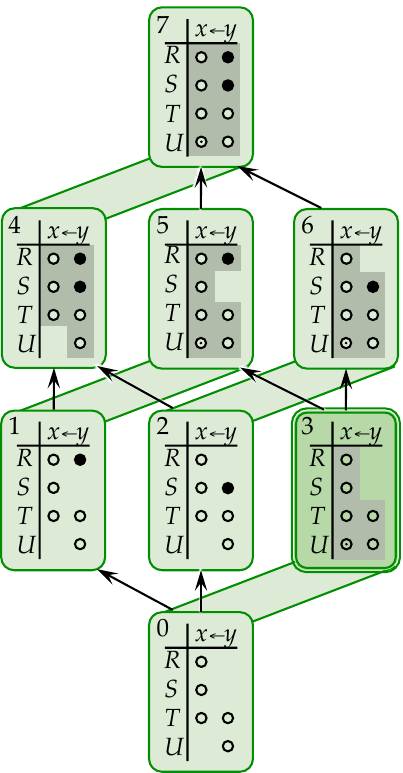}
		\label{Fig_PartialDissociationOrderExample_e}}		
	\hspace{8mm}	
	\subfloat[{$\Gamma: x \rightarrow y$}]{
		\includegraphics[scale=0.5]{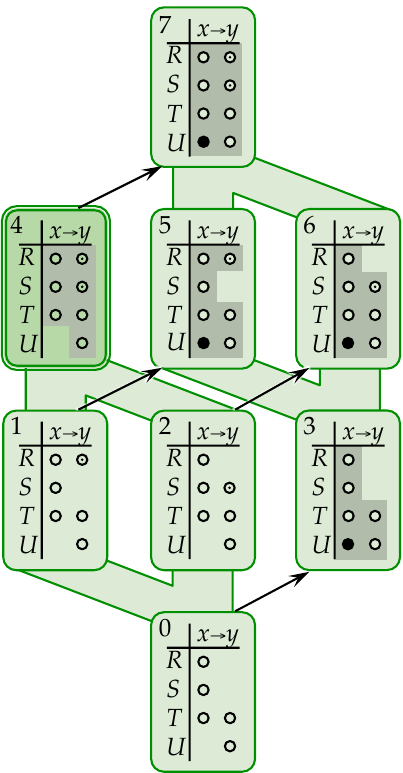}
		\label{Fig_PartialDissociationOrderExample_f}}
	\caption{\Autoref{ex:FDsComplicated}: (\autoref{ex:DRsComplicated} continued):
The presence of FDs also changes the probabilistic dissociation preorder and thus the minimal plans returned by our algorithm: either $P_{\Delta_3}$ (as in \autoref{Fig_PartialDissociationOrderExample_c}) or $P_{\Delta_4}$.}
	\label{Fig_PartialDissociationOrderExampleFD}
	\vspace{-1mm}
\end{figure}

\begin{example}[\autoref{ex:DRsComplicated} continued]\label{ex:FDsComplicated}
We illustrate here how FDs can change the ``\emph{probabilistic dissociation preorder}''. 
Analogously to DRs, we mark variables in the incidence matrix that are dissociated as result of an FD and do not affect the probability with a dotted circle ($\dotcirc$) instead of a bullet ($\bullet$).
As before, the preorder ${\preceq^p}$ is determined entirely by full circles (representing dissociated variables in probabilistic relations that are not implied by FDs on the other variables). 
However, as before, the correspondence to plans (as implied by the hierarchy between all variables) is still determined by all circles.

First consider $\Gamma: y \rightarrow x$:
\Autoref{Fig_PartialDissociationOrderExample_e} shows that this FD
leads to the same preorder as for DR $U^d$ from \autoref{Fig_PartialDissociationOrderExample_c}. Thus, the minimal plan is also $P_{\Delta_3}$.
Next consider $\Gamma: x \rightarrow y$:
\Autoref{Fig_PartialDissociationOrderExample_f} shows that there are now only two equivalence classes, both of which are safe, and one of which is minimal: 
$\Delta_0 \equiv^p \Delta_1 \equiv^p \Delta_2 \equiv^p \Delta_4$.
Among those, only $\Delta_4$ is hierarchical and is thus the one returned by the algorithm.
\markend
\end{example}

\section{Multi-query Optimizations}\label{sec:optimizations}
So far, we enumerate all minimal query plans, then take the minimum score of those plans in order to calculate the propagation score $\rho(q)$.
In this section, we develop three optimizations that can considerably reduce 
the necessary calculations for evaluating all minimal query plans. 
Notice that these three optimizations and the previous optimizations using schema knowledge are orthogonal and can be arbitrarily combined in the obvious way.
We use the following example to illustrate the first two optimizations:

\begin{example}[Multi-query optimizations]\label{ex:optimizationExample}
Consider the query
$q \datarule R(x,z),S(y,u),$ $T(z), U(u), M(x,y,z,u)$.
Our default is to evaluate all 6 minimal plans returned by \specificref{Algorithm}{alg:basicAlgorithm}, 
then take the minimum score (shown in \figref{Fig_AllOptimizedQueryPlans}{a}).
\Figref{Fig_AllOptimizedQueryPlans}{b}
and \figref{Fig_AllOptimizedQueryPlans}{c}
illustrate the optimized evaluations after applying Opt.~1, or Opt.~1 and Opt.~2, respectively.
\markend
\end{example}

\subsection{Opt.\ 1: One single query plan}

Our first optimization creates one single query plan by \emph{pushing the min-operator down into the leaves}. It thus avoids calculations when it is clear that other calculations must have lower bounds. 
The idea is simple: Instead of creating one query subplan for each min-cut  $\vec y \in \MinCuts(q)$ in \autoref{alg1:line12} of \autoref{alg:basicAlgorithm}, 
the adapted \autoref{alg:opt1Algorithm} takes the minimum score over those min-cuts, for each tuple of the head variables in \autoref{algDet:P:connected}. It thus creates one single query plan.
\Figref{Fig_AllOptimizedQueryPlans}{b} shows this single plan for our running example.

\begin{algorithm2e}[t]
\caption{Optimization~1 recursively pushes the min operator into the leaves and generates one single query plan.}\label{alg:opt1Algorithm}
\SetKwInput{Algorithm}{Recursive algorithm}
\SetKwFunction{SP}{SP}
\SetKwFor{ForAll}{forall}{do}{endfch}	
\LinesNumbered

\Algorithm{\FuncSty{SP (SinglePlan)}}
\KwIn{Query $q(\vec z) \datarule R_1(\vec x_1), \ldots, R_{m_p}(\vec x_{m_p}), \ldots, R_m^d(\vec x_m)$}
\KwOut{Single query plan $P$}
\BlankLine
		
\lIf{$m^p \leq 1$}{$\mathcal{P} \leftarrow \{\projpd{\vec x} 
	\joinp{}{R_1(\vec x_1), R_2(\vec x_2), \ldots, R_m^d(\vec x_m)}\}$}
\Else{
	\uIf{$q$ is disconnected}{
		Let $q = q_1, \ldots, q_k$ be the query components of $q$ \;
		\lForEach{$q_i$}{$\HVar(q_i) \leftarrow \HVar(q) \cap \Var(q_i)$}
		$P \leftarrow \joinp{}{\SP (q_1 ), \ldots, \SP (q_k )}$ \;\label{algDet:P:disconnected}
	}
	\Else{
		Let $\MinPCuts(q) = \{\vec y_1, \ldots, \vec y_j\}$ \;
		\lForEach{$\vec y_i$}{$q'_i \leftarrow q_i$ 
			with $\HVar(q'_i) \leftarrow \HVar(q) \cup \vec y_i$} 
		$P \leftarrow  \minp{\projp{\vec y_1} \SP (q_1'), \ldots, 
			\projp{\vec y_j} \SP (q_j')}$ \;\label{algDet:P:connected}
	}
}
\end{algorithm2e}

\subsection{Opt.~2: Re-using common subplans}\label{sec:reusingCommonSubplans}
Our second optimization calculates only once, then \emph{re-uses common subplans shared between the minimal plans}.
Thus, whereas our first optimization reduces computation by combining plans at their roots, the second optimization stores and re-uses common results in the branches by re-using views.
The adapted algorithm works as follows: It first traverses the whole single query plan and remembers each subplan by the atoms used and its head variables in a HashSet. If it sees a subplan twice, it creates a new view for this subplan, mapping the subplan to a new view definition. The actual plan then uses these views whenever possible. The order in which the views are created assures that the algorithm also discovers and exploits \emph{nested common subexpressions}.
\Figref{Fig_AllOptimizedQueryPlans}{c} shows the generated views and plans for our running example: Notice that the main plan and the view $V_3$ both re-use views $V_1$ and $V_2$.

\begin{figure}[t]
\centering
{\includegraphics[scale=0.82]{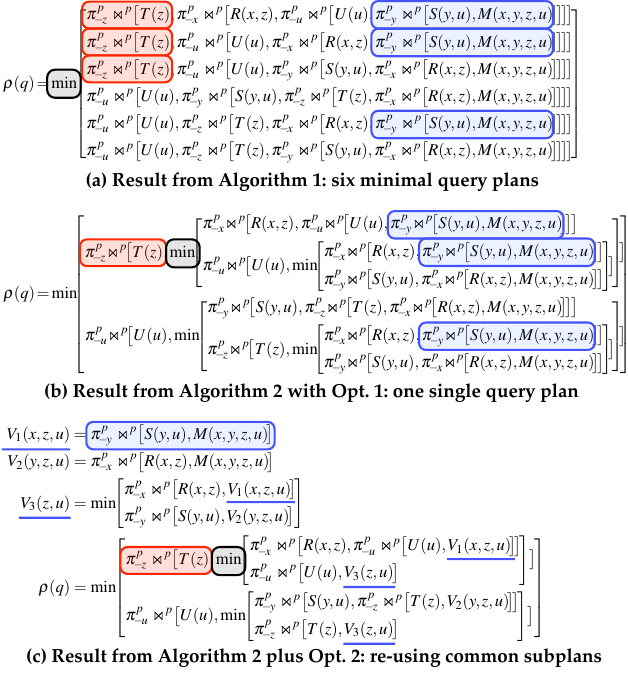}}
\vspace{-4mm}
\caption{
\Autoref{ex:optimizationExample} before and after applying optimizations 1 and 2.
  }
\label{Fig_AllOptimizedQueryPlans}
\end{figure}

\subsection{Opt.~3: Deterministic semi-join reduction}\label{sec:opt3}
The most expensive operations in probabilistic query plans are the group-bys for the probabilistic project operations. These are often applied early in the plans to tuples which are later pruned and do not contribute to the final query result. 
Our third optimization is to first apply a \emph{full semi-join reduction on the input relations} before starting the probabilistic evaluation from these \emph{reduced input relations}. 
 
We like to draw here an important connection to \cite{DBLP:conf/icde/OlteanuHK09}, which introduces the idea of ``lazy plans'' and shows orders of magnitude performance improvements for safe plans by computing confidences not after each join and projection, but rather at the very end of the plan.
We note that our semi-join reduction \emph{serves the same purpose} with similar performance improvements and also apply for safe queries.
The advantage of semi-join reductions, however, is that we do not require any modifications to the query engine.

\section{Experiments}\label{sec:experiments}

\noindent
We are interested in the efficiency (``how fast?'') and the quality (``how good?'') of ranking by dissociation as compared to 
exact probabilistic inference, 
Monte Carlo simulation (MC), 
and standard deterministic query evaluation (``deterministic SQL''). 
Our experiments, thus, investigate the following questions:
\emph{
How much can our three optimizations improve dissociation?
How fast is dissociation as compared to exact probabilistic inference, MC, and deterministic query evaluation? 
How good is the ranking from dissociation as compared to MC and ranking by lineage size? 
What are the most important parameters determining the ranking quality for each of the three methods?}

\introparagraph{Ranking quality}
We use \textit{mean average precision} (MAP) to evaluate the quality of a ranking by comparing it against the ranking from exact probabilistic inference as ground truth (GT). 
MAP rewards rankings that place relevant items earlier;
the best possible value is 1, and the worst possible 0.
We use a variant of ``Average Precision at 10'' defined as $\mathrm{AP}@10 := \frac{\sum_{k=1}^{10} \mathit{\mathrm{P}@k}}{10}$. Here, P@$k$ is the precision at the $k$th answer, i.e., the fraction of top $k$ answers according to GT that are also in the top $k$ answers returned.
Averaging over several experiments yields MAP~\cite{Manning:2008:IIR:1394399}.
We use a variant of the analytic method proposed in~\cite{McSherryN08} to  calculate $\mathrm{AP}$ in the presence of ties.
As baseline for no ranking, we 
use ``\emph{random average precision}''~\cite{DetwilerGLST2009:ICDE}, i.e.\ 
we assume all tuples have the same score and are thus tied for the same position.

\introparagraph{Exact probabilistic inference} 
Whenever possible, we calculate GT rankings with a tool called  
SampleSearch~\cite{gogate&dechter11,DBLP:conf/uai/GogateD10},
which also serves to evaluate the cost of exact probabilistic inference.
We describe the method of evaluating the lineage DNF with SampleSearch  in~\cite{DBLP:journals/tods/GatterbauerS14}.

\introparagraph{Monte Carlo (MC)}
We evaluate the MC simulations for different numbers of samples and write MC($x$) for $x$ samples. For example, AP for MC(10k) is the result of sampling the individual tuple scores 10\,000 times from their lineages and then evaluating AP once over the sampled scores. The MAP scores together with the standard deviations are then the average over several repetitions.

\introparagraph{Ranking by lineage size}
To evaluate the potential 
of non-probabi\-listic methods for ranking answers,
we also rank the answer tuples by decreasing size of their lineages; i.e.\ number of clauses in their DNFs.
Intuitively, a larger lineage size indicates that an answer has more ``support'' and should thus be more important. 
Notice that, in contrast to other methods, we ignore here the weight of support and correlations.

\begin{figure}[t]
\centering
\vspace{-3mm}
\renewcommand{\tabcolsep}{0.9mm}
\renewcommand{\arraystretch}{0.95}
\subfloat[$Q(a)$]{
\label{Fig_TPCHhierarchies_a}
\begin{tabular}[t]{@{\hspace{1pt}} >{$}r<{$}|>{$}c<{$} | >{$}c<{$} >{$}c<{$} >{$}c<{$} >{$}c<{$} @{\hspace{1pt}}}
				& a    	& s		& p 	& n 		\\
				\hline
	S			&\circ	&\circ	&		& 	 		\\
	\mathit{PS}	&		&\circ	& \circ & 		 	\\
	P			&		& 		& \circ	& \circ
\end{tabular}
}
\hspace{4mm}
\subfloat[$Q^S(a)$]{
\label{Fig_TPCHhierarchies_b}
\begin{tabular}[t]{@{\hspace{1pt}} >{$}r<{$}|>{$}c<{$} | >{$}c<{$} >{$}c<{$} >{$}c<{$} >{$}c<{$} @{\hspace{1pt}}}
				& a    	& s		& p 	& n 		\\
				\hline
	S			&\circ	&\graycell\circ	&\graycell\bullet& 	 		\\
	\mathit{PS}	&		&\graycell\circ	&\graycell \circ & 		 	\\
	P			&		&				&\graycell \circ & \graycell\circ
\end{tabular}
}
\hspace{4mm}
\subfloat[$Q^P(a)$]{
\label{Fig_TPCHhierarchies_c}
\begin{tabular}[t]{@{\hspace{1pt}} >{$}r<{$}|>{$}c<{$} | >{$}c<{$} >{$}c<{$} >{$}c<{$} >{$}c<{$} @{\hspace{1pt}}}
				& a    	& s		& p 	& n 		\\
				\hline
	S			&\circ	&\graycell\circ	&				& 	 		\\
	\mathit{PS}	&		&\graycell\circ	&\graycell \circ & 		 	\\
	P			&		&\graycell\bullet&\graycell \circ & \graycell\circ
\end{tabular}
}
\caption{
Parameterized Deterministic SQL query $Q(a)$ over TPC-H.
Incidence matrices for TPC-H query $Q(a)$ and its two minimal hierarchical dissociations from either dissociating table $S$ or table $P$.}
\label{Fig_TPCHhierarchies}
\end{figure}

\introparagraph{Setup 1}
We use the {TPC-H} DBGEN data generator~\cite{tpc-h}
to generate a 1GB database to which we add a column \sql{P} for each table and store it in PostgreSQL 9.2~\cite{postgresql}.
We assign to each input tuple $i$ 
a random probability $p_i$ uniformly chosen from the interval $[0, p_{i\max}]$, resulting in an expected average input probability $\avg[p_{i}] = p_{i\max}/2$. 
By using databases with $\avg[p_{i}] < 0.5$, we can avoid output probabilities close to $1$ for queries with very large lineages. 
We use the following intuitive parameterized hard query:
\begin{align*}
	&Q(a)  \datarule S(\underline{s},a), \mathit{PS}(s,u), P(\underline{u},n), s \leq \$1, n \sql{ like } \$2 \\
&
\parbox{30mm}{
\vspace{-1mm}
	\sql{
		\begin{tabbing}
		\hspace{3mm}\=\hspace{3mm}\=\hspace{3mm}\=\hspace{0cm}\=\hspace{0cm}\=\kill
		select distinct s\_nationkey	
		from Supplier, Partsupp, Part	\\
		where s\_suppkey = ps\_suppkey	
		and ps\_partkey = p\_partkey	\\
		and s\_suppkey $<=$ \$1			
		and p\_name like \$2	
		\end{tabbing}
\vspace{-3mm}
}}
\notag
\end{align*}
Relations $S$, $\mathit{PS}$ and $P$ represent tables \sql{Supplier}, \sql{PartSupp} and \sql{Part}, respectively. Variable $a$ stands for attribute \sql{nationkey} (``answer tuple''), $s$ for \sql{suppkey}, $u$ for \sql{partkey} (``unit''), and $n$ for \sql{name}.
The probabilistic version of this query is:
``\emph{Which nations (as determined by the attribute \sql{nationkey}) are most likely to have suppliers with \sql{suppkey} $\leq {\$1}$ that supply parts with a $\sql{name like } \$2$?}''
Parameters {\$1} and {\$2} allow us to 
change the lineage size.
Tables \sql{Supplier}, \sql{Partsupp} and \sql{Part} have 10k, 800k and 200k tuples, respectively.
There are 25 different numeric attributes for \sql{nationkey} and our goal is to efficiently rank these 25 nations. 
As baseline for not ranking, we use random average precision for 25 answers, which leads to MAP@10 $\approx 0.220$.
This query has the following two minimal query plans (\autoref{Fig_TPCHhierarchies}):
{\fontsize{9}{20}	
\begin{align*}
	& P_S(a)  = \projpd{a}\joinp{u}{\projpd{a,u}\joinp{s}{
		S(\underline{s},a),\mathit{PS}(s,u), s \leq \$1},P(\underline{u},n), n \sql{ like } \$2}	\\
	& P_P(a) = \projpd{a} \joinp{s}{
		S(\underline{s},a), \projpd{s} \joinp{u}{\mathit{PS}(s,u), s \leq \$1,
			P(\underline{u},n), n \sql{ like } \$2 }}
\end{align*}}
Here, $P_S$ and $P_P$ stand for the plans that dissociate tables \sql{Supplier} or \sql{Part}, respectively. 
We take the minimum of the two bounds to determine the propagation score for each answer tuple $a$.
We will also evaluate the speed-up from applying the following deterministic semi-join reduction (Optimization 3) on the input tables
and then reusing intermediate query results across both query plans:
\begin{align*}		
	\mathit{PS}^*(s,u) & \datarule 
		\mathit{PS}(s,u), S(\underline{s},a), P(\underline{u},n),
		s \leq \$1, n \sql{ like } \$2\\
	\mathit{P}^*(u,n) & \datarule 
		P(\underline{u},n),\mathit{PS}^*(s,u)	
\end{align*}

\begin{figure*}[t]
    \centering
	\vspace{-4mm}
	\subfloat[$4$-chain query]
       	{\includegraphics[scale=0.425]{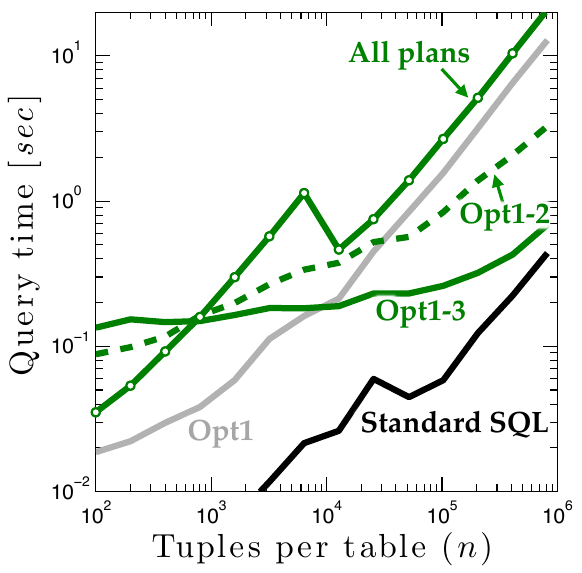}
    	\label{Fig_SyntheticChain4}}
	\hspace{-0.1mm}
	\subfloat[$7$-chain query]
		{\includegraphics[scale=0.425]{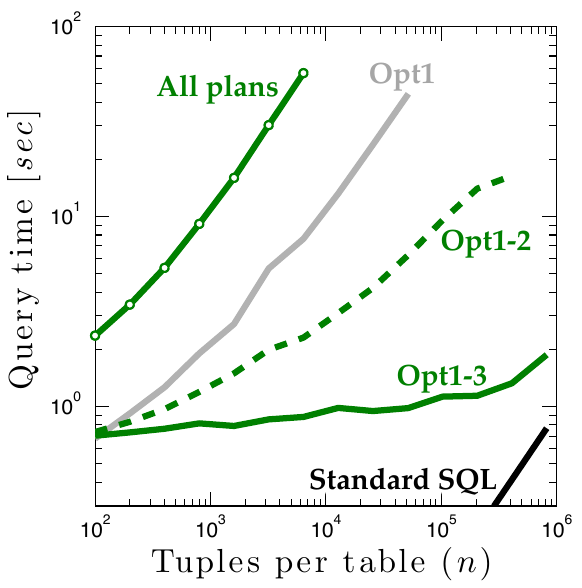}
		\label{Fig_SyntheticChain7}}
	\hspace{-0.1mm}
	\subfloat[$2$-star query]
		{\includegraphics[scale=0.425]{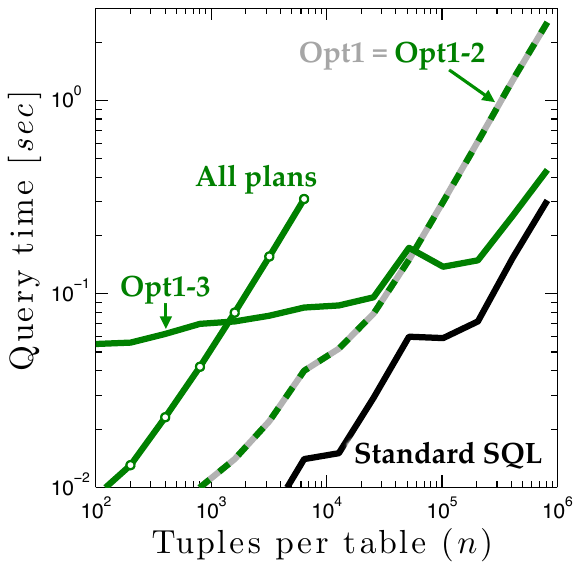}
		\label{Fig_SyntheticStar2}}
    \hspace{-0.1mm}
	\subfloat[$5$-star query]
       	{\includegraphics[scale=0.425]{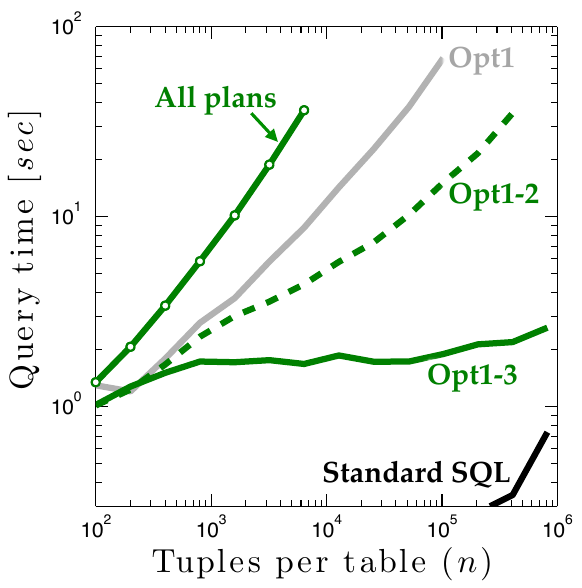}
    	\label{Fig_SyntheticStar5}}
	\hspace{5mm}
    \subfloat[\$2 = \%red\%green\%]
		{\includegraphics[scale=0.32]{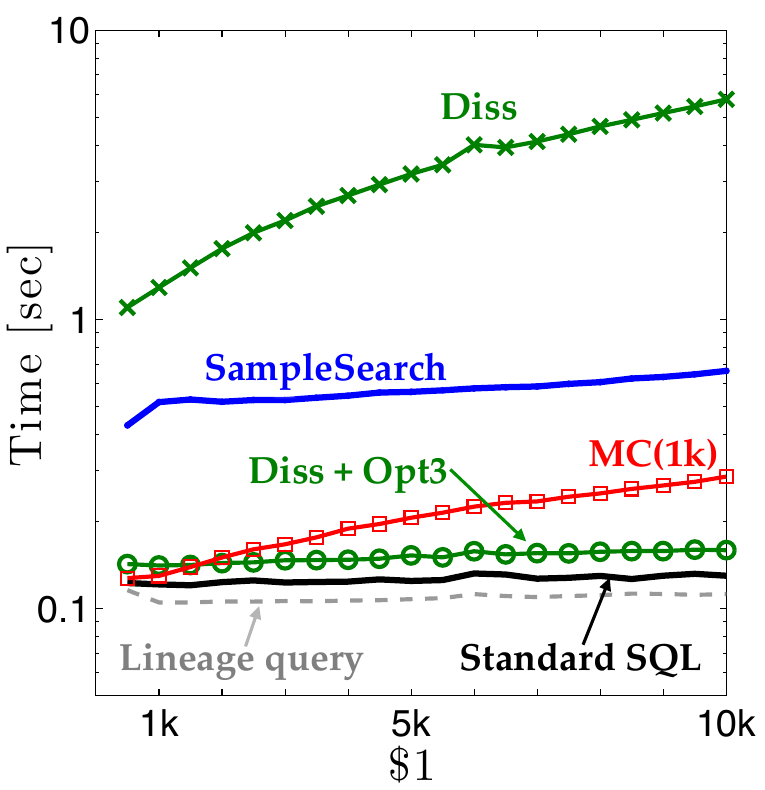}
		\label{Fig_VLDBJ_TPCH_timing1}}					
	\hspace{-0.2mm}
	\subfloat[\$2 = \%red\%]
		{\includegraphics[scale=0.32]{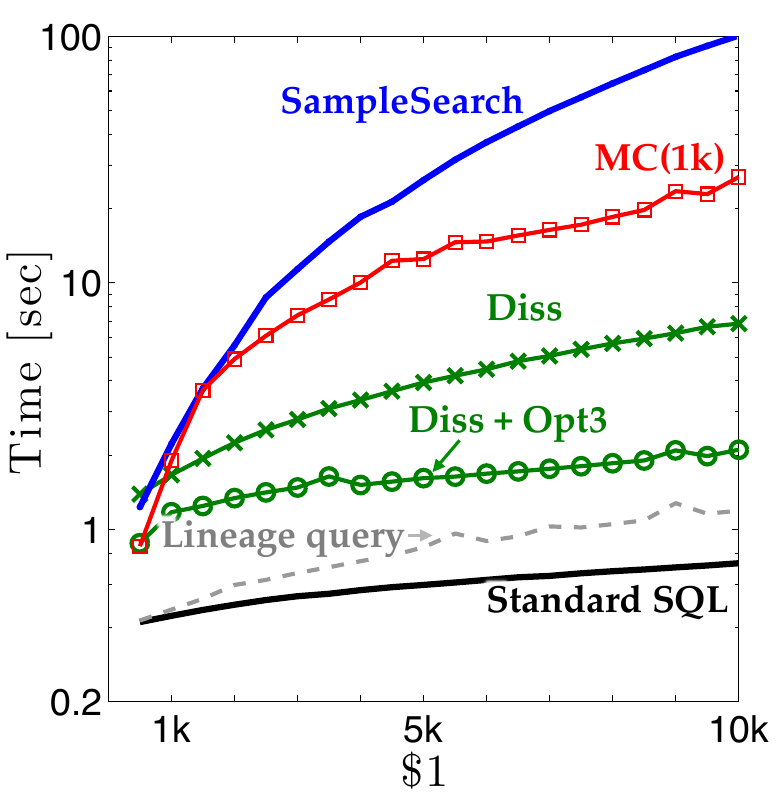}
		\label{Fig_VLDBJ_TPCH_timing2}}
	\hspace{-0.2mm}
    \subfloat[\$2 = \%]
		{\includegraphics[scale=0.32]{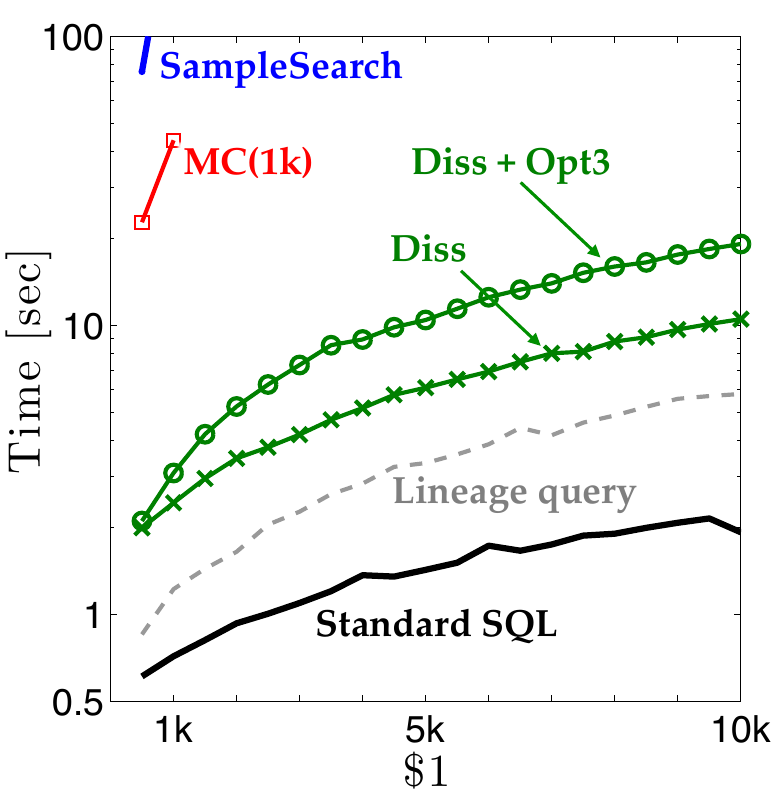}
		\label{Fig_VLDBJ_TPCH_timing3}}	
	\hspace{-0.2mm}
    \subfloat[Combining (a)-(c)]
		{\includegraphics[scale=0.32]{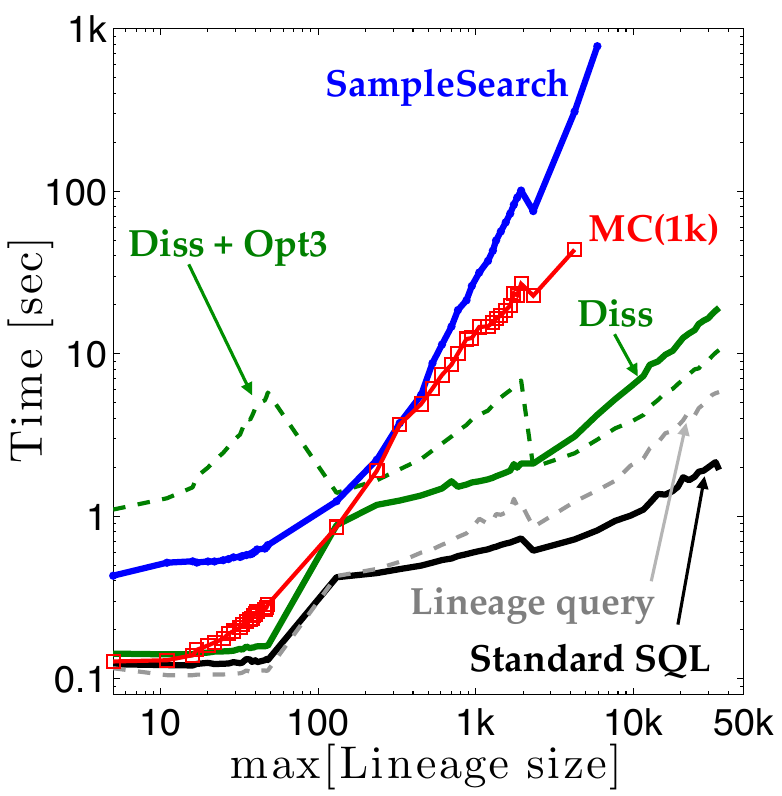}
		\label{Fig_VLDBJ_TPCH_timing4}}		
	\caption{Timing results:
	(a)-(d) For increasing database sizes and constant cardinalities, our optimizations approach deterministic SQL performance.
	(e)-(h) For the TPC-H query, the best evaluation for dissociation is within a factor of $6$ of that for deterministic query evaluation.}
    \label{Fig_VLDBJ_TPCH_timing}
\end{figure*}

\introparagraph{Setup 2}
We compare the runtimes for our three optimizations against evaluation of all plans 
for $k$-chain queries and $k$-star queries
over varying database sizes (data complexities) and varying query sizes (query complexities). 
The $k$-chain queries return many results, whereas the $k$-star queries return one tuple, 
thus representing a Boolean query:
\begin{align*}
	&{\textrm{$k$-chain: }}q(x_0, x_k)  \datarule R_1(x_0, x_1), R_2(x_1,x_2), \ldots, R_k(x_{k-1},x_k) \\
	&{\textrm{$k$-star: }}q(a)	 \datarule R_1(a, x_1), R_2(x_2), \ldots, R_k(x_k), R_0(x_1, \ldots, x_k)
\end{align*}

\noindent
We denote the length of the query with $k$, the number of tuples per table with $n$, and the domain size with $N$. We use integer values which we uniformly draw from the range $\{1,2, \ldots N\}$. 
Thus, the parameter $N$ determines the \emph{selectivity} and is varied as to 
keep the answer cardinality constant around 20-50 for chain queries, or the
answer probability between 0.90 and 0.95 for star queries.
For the data complexity experiments, we vary the number of tuples $n$ per table between $100$ and $10^6$. 
For the query complexity experiments, we vary $k$ between $2$ and $8$ for chain queries.
For these experiments, the optimized (and often \emph{extremely long}) SQL statements are ``calculated'' in JAVA and then sent to Microsoft SQL server 2012~\cite{sqlServer}.
To illustrate with numbers, we have to issue 429 query plans in order to evaluate the $8$-chain query (see \autoref{table:numberMinimalQueryPlans}). Each of these plans joins 8 tables in a different order. Optimization 1 then merges those plans together into one truly gigantic single query plan.

\subsection{Runtime experiments}\label{sec:exOptimizations}

\begin{tcolorbox}[top=-2mm]
\begin{questionW}\label{question11}
When and how much do our three query optimizations speed up query evaluation?
\end{questionW}

\begin{resultW}\label{lesson11}
Combining plans (Opt.~1) and using intermediate views (Opt.~2) almost always speeds up query times.
The semi-join reduction (Opt.~3) 
slows down queries with high selectivities, but 
considerably speeds up queries with small selectivities,
bringing probabilistic query evaluation close to deterministic evaluation.
\end{resultW}
\end{tcolorbox}

\introparagraph{Setup 2}
\Autoref{Fig_SyntheticChain4} 
to
\autoref{Fig_SyntheticStar5}
show the results for increasing database sizes,
and \autoref{Fig_SyntheticQuerySize} for increasing query sizes.
For example, 
\autoref{Fig_SyntheticChain7} shows the performance of computing a 
7-chain query
which has 132 hierarchical dissociations.  
Evaluating each of these queries separately takes a long time,
while our optimization techniques 
bring evaluation time close to deterministic query evaluation.  
Especially on larger databases, where the running time is I/O bound, the
penalty of the probabilistic inference is only a factor of 2-3 in this example.
Notice here the trade-off between optimization~1,2 and optimization 1,2,3: Optimization 3 applies a full semi-join reduction on the input relations before starting the probabilistic plan evaluation from these reduced input relations. This operation imposes a rather large constant overhead, both at the query optimizer and at query execution. For larger databases (but constant selectivity), this overhead is amortized. 
Without self-join reductions, opimization~1,2 would not execute on the $6$-star query with 720 minimal query plans at all (``The query processor ran out of internal resources and could not produce a query plan'').
In practice, this suggests that dissociation allows us a large space of optimizations depending on the query and particular database that can conservatively extend the space of optimizations performed today in deterministic query optimizers.

\introparagraph{Setup 1} \Autoref{Fig_VLDBJ_TPCH_timing1} to \autoref{Fig_VLDBJ_TPCH_timing3} 
compare the running times for
dissociation with two minimal query plans  (``Diss''), 
dissociation with semi-join reduction (``Diss + Opt3''), exact probabilistic inference (``SampleSearch''), 
Monte Carlo with 1000 samples (``MC(1k)''),
retrieving the lineage only (``Lineage query''), 
and deterministic query evaluation without ranking (``Standard SQL'').
As experimental platform, we use PostgreSQL 9.2 on a 2.5 Ghz Intel Core i5 with 16G of main memory. We run each query 5 times and take the average execution time.
We fixed $\$2$ to $\sql{'\%red\%green\%'}$, $\sql{'\%red\%'}$ or $\sql{'\%'}$ and varied 
$\$1 \in \{500, 1000, $ $\ldots 10k\}$.
\Autoref{Fig_VLDBJ_TPCH_timing4} combines all three previous plots and shows the times as function of the maximum lineage size (i.e. the size of the lineage for the tuple with the maximum lineage) of a query.
We see here again that the semi-join reduction speeds up evaluation considerably for small lineage sizes
(\autoref{Fig_VLDBJ_TPCH_timing1} shows speedups of up to 36).
For large lineages, however, the semi-join reduction is an unnecessary overhead, as most tuples are participating in the join anyway (\autoref{Fig_VLDBJ_TPCH_timing2} shows overhead of up to 2).

\begin{figure}[t]
    \centering
	\vspace{-3mm}
	\subfloat[$k$-chain queries]
       	{\includegraphics[scale=0.41]{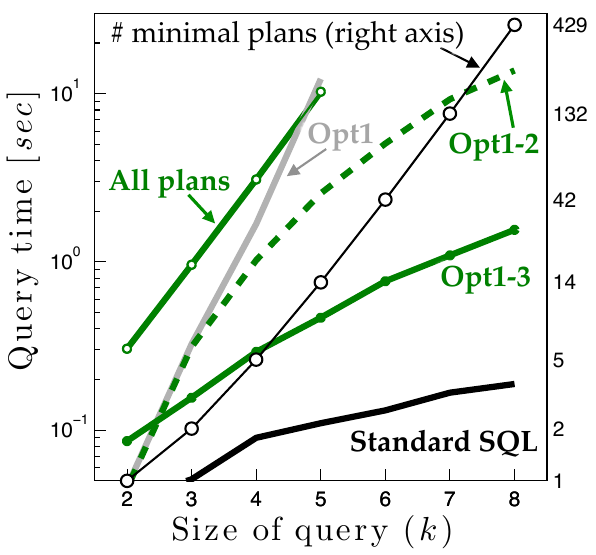}
    	\label{Fig_SyntheticChainQuerySize}}
	\hspace{-2mm}
	\subfloat[$k$-star queries]
		{\includegraphics[scale=0.41]{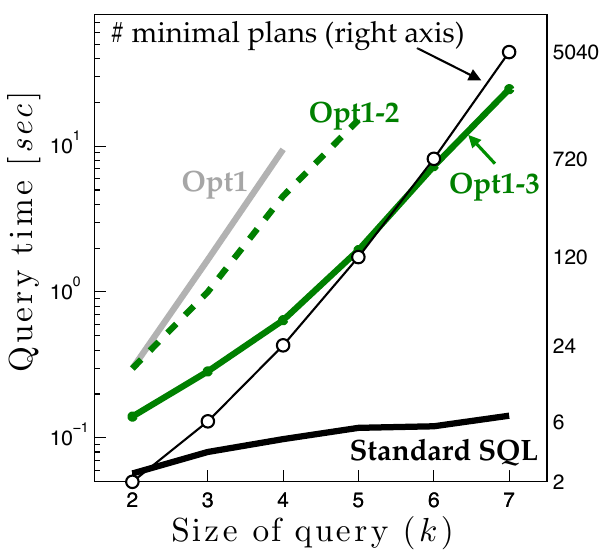}
		\label{Fig_SyntheticStarQuerySize}}
	\caption{While the query complexity is exponential (number of minimal plans are shown on the right side),
	our optimizations can even evaluate a very large number of minimal plans (here shown up to 429 for a 8-chain query and 5040 (!) for a 7-star query).
	}
    \label{Fig_SyntheticQuerySize}
\end{figure}

\begin{tcolorbox}[top=-2mm]
\begin{questionW}\label{question8}
How does dissociation compare against other probabilistic methods and standard query evaluation?
\end{questionW}

\begin{resultW}\label{lesson8}
The best evaluation strategy for dissociation takes only a small overhead over standard SQL evaluation and is considerably faster than other probabilistic methods for large lineages.
\end{resultW}
\end{tcolorbox}

\Autoref{Fig_VLDBJ_TPCH_timing1} to \autoref{Fig_VLDBJ_TPCH_timing4} 
show that SampleSearch does not scale to larger lineages as 
the performance of exact probabilistic inference
depends on the
tree-width of the Boolean lineage formula, which generally increases with the size of the data. In contrast, dissociation is \emph{independent of the treewidth}.
For example, SampleSearch needed 780 sec for calculating the ground truth for a query with $\max[\textrm{lin}]= 5.9$k 
for which dissociation took 3.0 sec, and MC(1k) took 42 sec for a query with $\max[\textrm{lin}]= 4.2$k 
for which dissociation took 2.4 sec.
Dissociation takes only 10.5 sec for our largest query $\$2 = \sql{'\%'}$ and  $\$1 = 10k$ with  $\max[\textrm{lin}]= 35$k.
Retrieving the lineage for that query alone takes 5.8 sec, which implies that any probabilistic method 
that evaluates the probabilities outside of the database engine needs to issue this query to retrieve the DNF for each answer and would thus have to evaluate lineages of sizes around 35k in only 4.7 (= 10.5 - 5.8) sec to be faster than dissociation.\footnote{The time needed for the lineage query thus serves as minimum benchmark for \emph{any} probabilistic approximation.
The reported times for SampleSearch and MC are the sum of time for retrieving the lineage plus the actual calculations, without the time for reading and writing the input and output files for SampleSearch.}

\begin{figure}[t]
\centering
\vspace{-0mm}
\begin{minipage}[b]{86mm}
\setlength{\tabcolsep}{1.5mm}
\scriptsize
\begin{tabular}[b]{@{\hspace{0pt}} r | r r r r r r @{\hspace{0pt}}  }
	\$2	&\multicolumn{2}{c}{\!\!\!\scriptsize{\%red\%green\%}\!\!\!\!}
		&\multicolumn{1}{c}{\!\!\!\!\scriptsize\%red\%\!\!\!\!}
		&\multicolumn{1}{c}{\!\!\!\!\scriptsize\%red\%\!\!\!\!}
		&\multicolumn{1}{c}{\scriptsize\%} &\multicolumn{1}{c}{\scriptsize\%} \\
	\$1						&500		&10\,000	&500		&10\,000	&500	&10\,000	\\
	\hline     
	max[lineage size]       		&$5$	 &$48$		&$131$	&$1,941$		&$2\,320$&$35\,040$ \\
	total lineage size				&$42$	 &$1\,004$ &$2\,218$	&$44\,152$	&$40\,000$	&$800\,000$	\\			  	    	        	
	\textbf{SampleSearch} [sec]		&$0.43$  &$0.66$  	&$1.23$	&$100.71$		&$75.47$&$-$ 	\\	
	\textbf{MC(1k)} [sec]			&$0.13$  &$0.29$  	&$0.86$	&$26.87$		&$22.75$&$-$ 	\\		
	\textbf{Dissociation \& SJ} [sec]&$0.14$  &$0.16$  	&$0.88$	&$2.11$			&$2.11$	&$19.14$ 	\\
	\textbf{Dissociation} [sec]		&$1.10$  &$5.76$  	&$1.39$	&$6.83$			&$2.00$	&$10.52$	\\
	\textbf{Lineage SQL}	[sec]	&$0.12$  &$0.11$  	&$0.43$	&$1.19$			&$0.86$	&$5.80$		\\
	\textbf{Deterministic SQL} [sec]&$0.12$  &$0.13$  	&$0.42$	&$0.73$			&$0.61$	&$1.93$		\\
\end{tabular}
\end{minipage}
\caption{Overview timing results TPC-H.}
\label{Fig_VLDBJ_TPCH_timingTable}
\end{figure}

\introparagraph{Further optimizations}: We found that materialized views performed better than just views. For example, the query $\$1 = 500$ and $\$2 = \sql{'\%red\%green\%'}$ takes over 3 sec with common views instead of our reported 0.88 sec for materialized views.
We also found that using standard database-provided aggregates (which requires us to use the logarithm for products) instead of user-defined aggregates notably speeds up query evaluation for large lineages. 
Concretely, instead of every occurrence of \sql{'ior(T.P) as P'} in our queries, we used the following nested PostgreSQL expression: 
\sql{'case when (sum(case T.P when 1 then -746 else ln(1-T.P) end)) $<$ -745 then 1
else 1-exp(sum(case T.P when 1 then -746 else ln(1-T.P) end)) end as P'}.
The outer case statement prevents errors for deterministic tuples (i.e.\ with $p_i = 1$), and the inner case statement prevents errors due to underflows. 
As illustration of the improvements, the query $\$1 = 10k$ and $\$2 = \sql{'\%'}$ would take 42.2 sec instead of 20.7 with semi-join reduction, and 32.5 sec instead of 11.3 for the two individual query plans when using a UDF instead of the above expression. We also found that removing the outer case statement would reduce the time by 5\% (which could be used if there were no deterministic tuples in a table), and removing the inner case by another 1\% (which could be used if there was no risk of underflows). 
An important by-product of using standard database-defined aggregates is that dissociated queries (and their optimized versions) can be executed with the help of \emph{any standard relational database}, even cloud-based databases that commonly do not allow users to define their own UDAs, e.g.\ Microsoft SQL Azure.
To our best knowledge, this is the currently only technique to approximate rankings of probabilistic queries \emph{without any modifications to the database engine nor performing any calculations outside the database}.

\begin{figure*}[t]
    \centering
	\vspace{-3mm}
	\hspace{0.5mm}
	\subfloat[\autoref{lesson1}]
		{\includegraphics[scale=0.325]{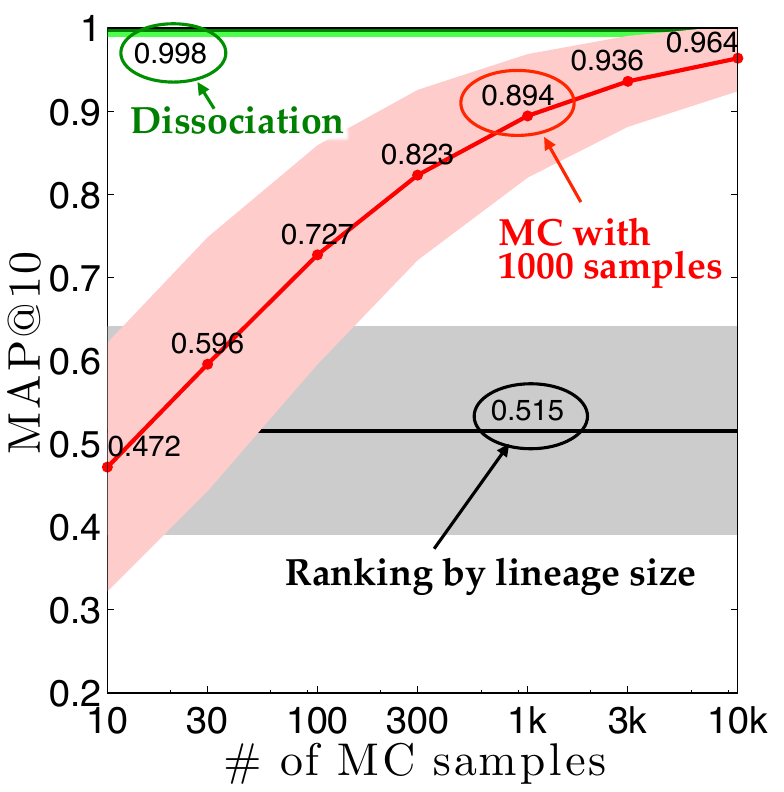}
		\label{Fig_VLDBJ_TPCH_AP_MC_aggregated_09}}
	\hspace{-1mm}
    \subfloat[\autoref{lesson2}]
		{\includegraphics[scale=0.325]{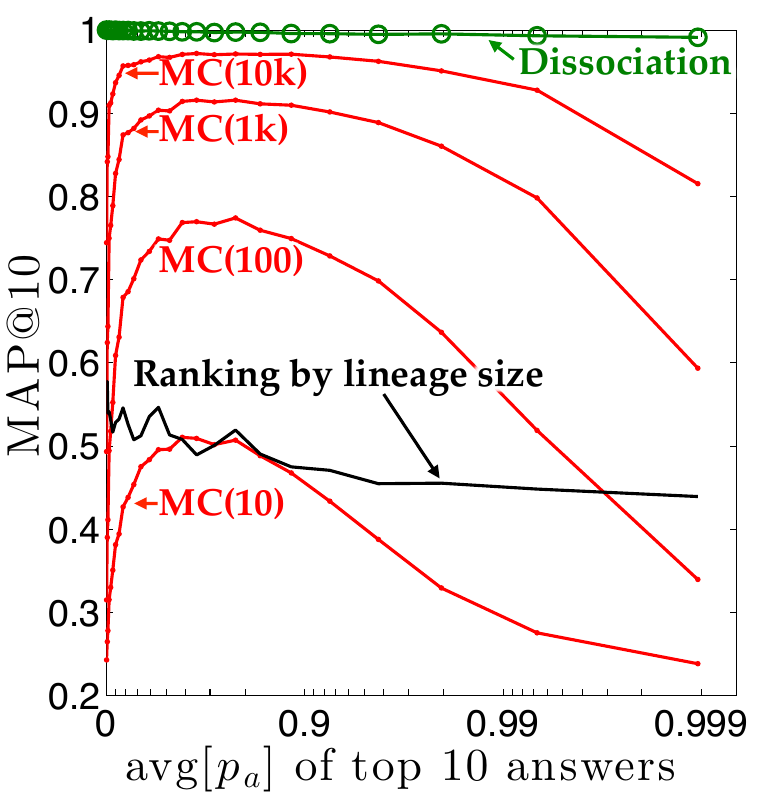}
		\label{Fig_VLDBJ_TPCH_ap_mc_random_redgreen_equalSizedBins_log}}	
	\hspace{-1mm}
    \subfloat[\autoref{lesson3}]
		{\includegraphics[scale=0.325]{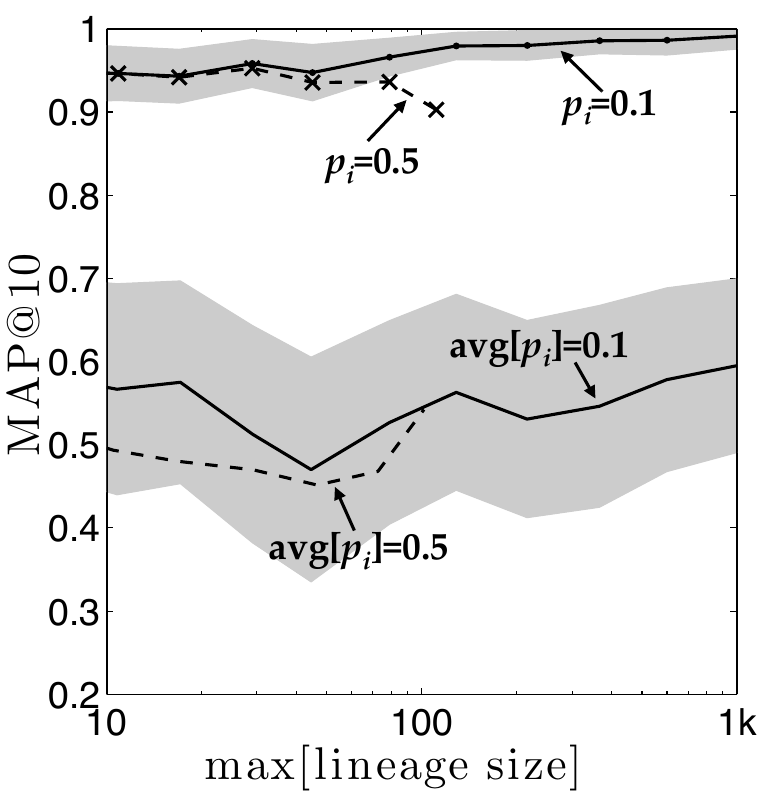}
		\label{Fig_VLDBJ_TPCH_AP_LineageRanking}}
	\hspace{-1mm}
	\subfloat[\autoref{lesson4}]
		{\includegraphics[scale=0.325]{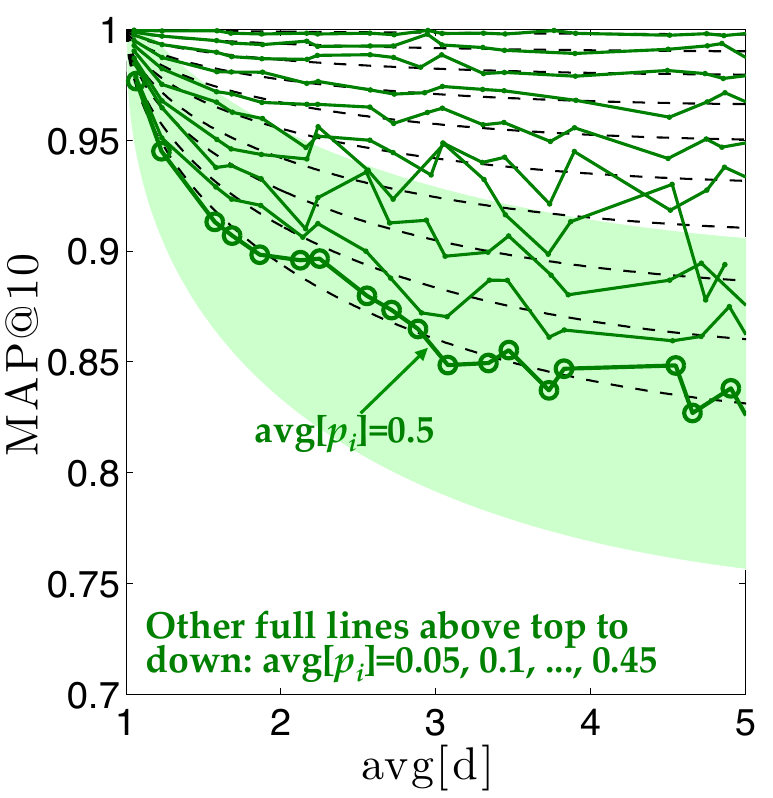}
		\label{Fig_VLDBJ_TPCH_AP_Diss_MC_Tradeoff1}}
	\hspace{-1mm}
	\subfloat[\autoref{lesson4}]
		{
		\includegraphics[scale=0.325]{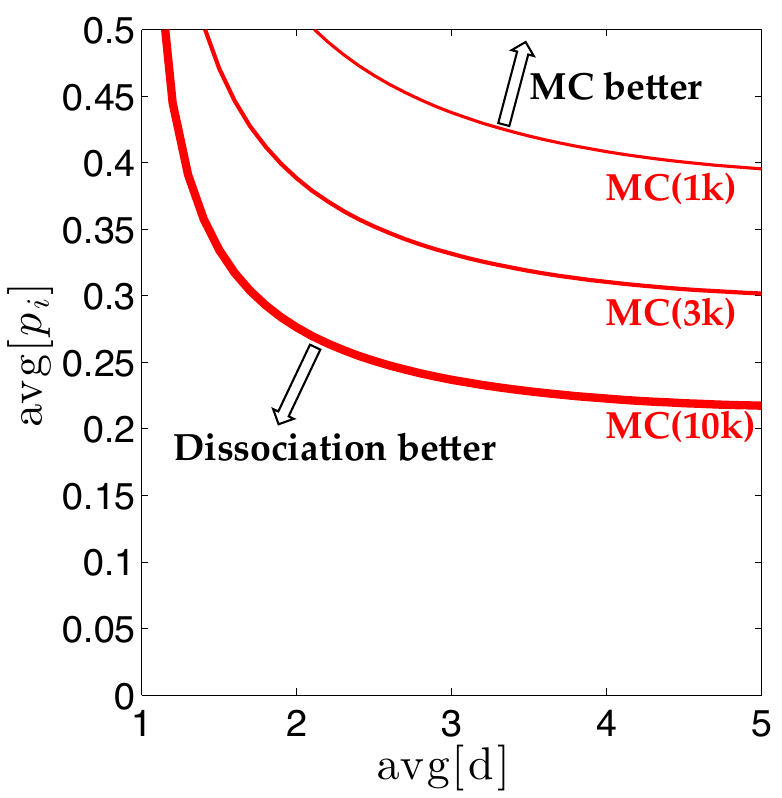}
		\label{Fig_VLDBJ_TPCH_AP_Diss_MC_Tradeoff2}}
	\hspace{-1mm} 
	\subfloat[\autoref{lesson6}]
		{\includegraphics[scale=0.325]{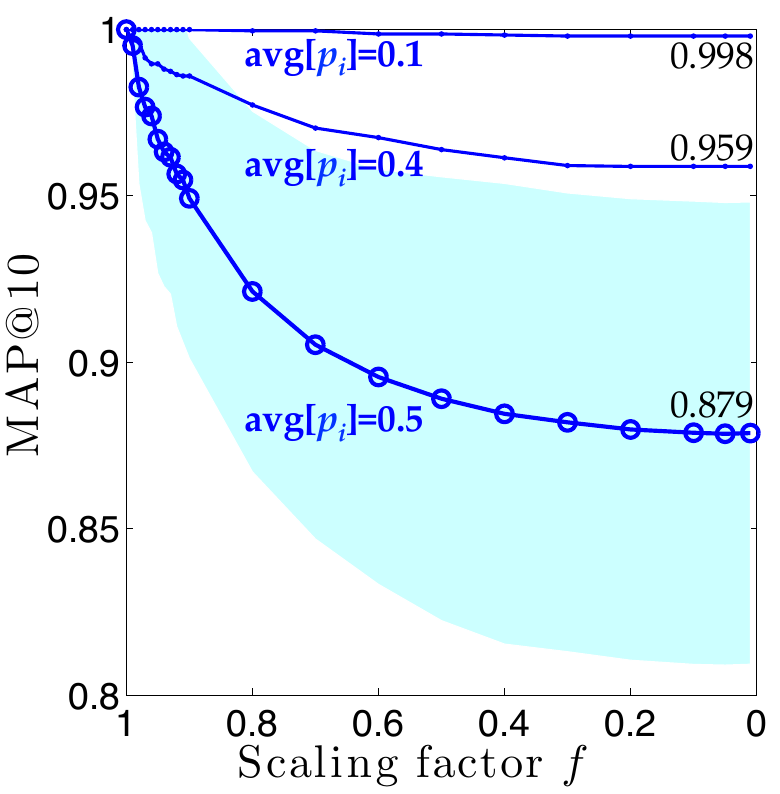}
		\label{Fig_VLDBJ_TPCH_AP_ScaledGT}}	
	\hspace{1mm} 
	\subfloat[\autoref{lesson6}]
		{\includegraphics[scale=0.38]{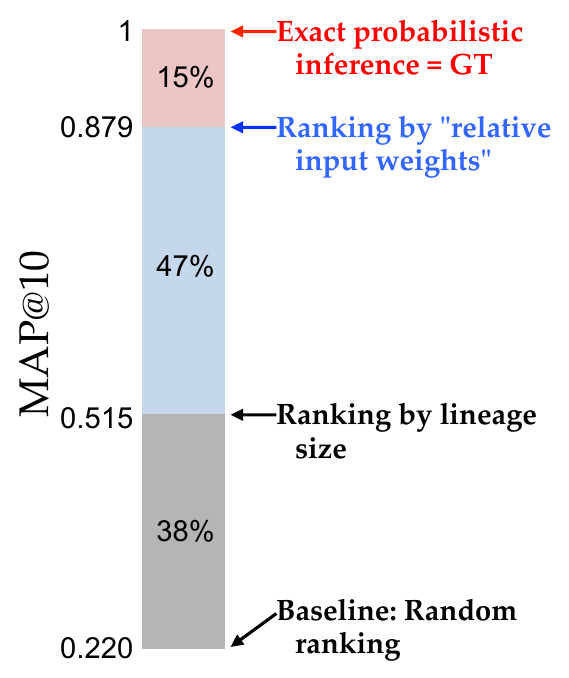}
		\label{Fig_VLDBJ_TPCH_AP_overviewBars}}	
	\hspace{3mm} 
	\subfloat[\autoref{lesson7}]
		{\includegraphics[scale=0.325]{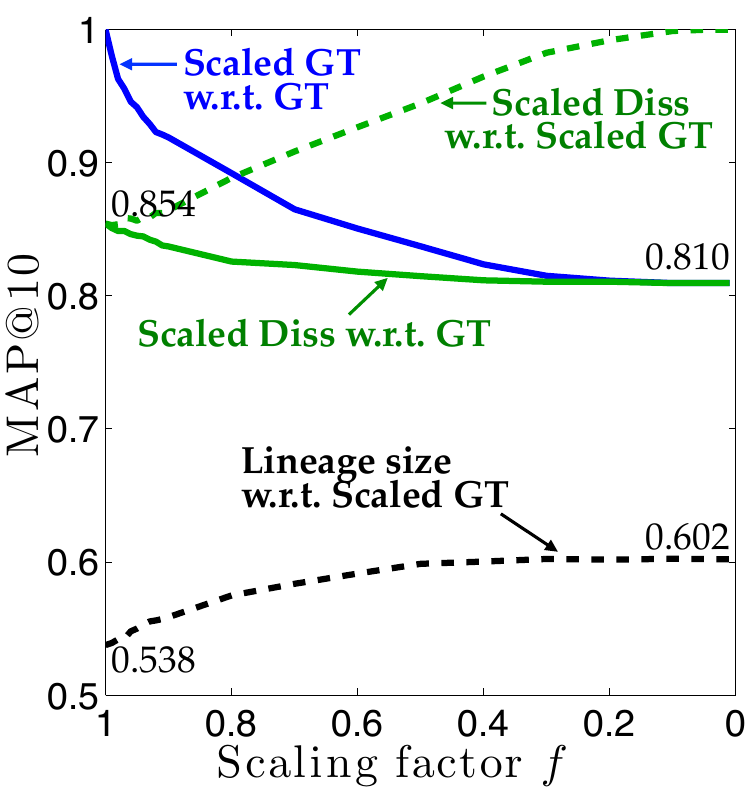}
		\label{Fig_VLDBJ_TPCH_AP_ScaledDissociation}}	
	\caption{
Ranking experiments on TPC-H: Assumptions for from each plot and conclusions are described below each respective result in the text.}
    \label{Fig_VLDBJ_TPCH_AP_Diss_MC_Tradeoff}
\end{figure*}

\subsection{Ranking experiments}\label{sec:exRanking}

\noindent
For the following experiments, we are limited to those query parameters \$1 and \$2 for which we can get the ground truth (and results from MC) in acceptable time.
We systematically vary $p_{i\max}$ between $0.1$ and $1$ (and thus $\avg[p_{i}]$ between $0.05$ and $0.5$) and evaluate the rankings several times over randomly assigned input tuple probabilities.
We only keep data points (i.e.\ results of individual ranking experiments) for which the output probabilities are not too close to 1 to be meaningful ($\max[p_{a}]<0.999\,999$).

\begin{tcolorbox}[top=-2mm]
\begin{questionW}\label{question1}
How does ranking quality compare for our three ranking methods and which are the most important factors that determine the quality for each method?
\end{questionW}

\begin{resultW}\label{lesson1}
Dissociation performs better than MC which performs better than ranking by lineage size.
\end{resultW}
\end{tcolorbox}

\Autoref{Fig_VLDBJ_TPCH_AP_MC_aggregated_09} shows averaged results of our probabilistic methods for $\$2 = \sql{'\%red\%green\%'}$.\footnote{Results for MC with other parameters of \$2 are similar. However, the evaluation time for the experiments becomes quickly infeasible.} 
Shaded areas indicate standard deviations and the x-axis shows varying numbers of MC samples. 
We only used those data points for which $\avg[p_{a}]$ of the top 10 ranked tuples is between $0.1$ and $0.9$ according to ground truth 
($\approx 6$k data points for dissociation and lineage, $\approx 60$k data points for MC, as we repeated each MC simulation 10 times), as this is the best regime for MC, according to \autoref{lesson2}.
\Autoref{Fig_TODS_TPCH_Result_AP_tableOverview} 
gives an overview of the performance of each method, depending on the most important parameters which we will explain next.

We also evaluated quality for dissociation and ranking by lineage for more queries by choosing parameter values for $\$2$ from a set of 28 strings, such as \sql{'\%r\%g\%r\%a\%n\%d\%'} 
and \sql{'\%re\%re\%'}.
The average MAP over all 28 choices for parameters \$2 is 0.997 for ranking by dissociation and 0.520 for ranking by lineage size ($\approx 100$k data points). Most of those queries have too large of a lineage to evaluate MC. Note that ranking by lineage always returns the same ranking for given parameters \$1 and \$2, but the GT ranking would change with different input probabilities.

\begin{tcolorbox}[top=-2mm]
\begin{resultW}\label{lesson2}
Ranking quality of MC increases with the number of samples and decreases when the average probability of the answer tuples $\avg[p_a]$ is close to $0$ or $1$.  
\end{resultW}	
\end{tcolorbox}

\Autoref{Fig_VLDBJ_TPCH_ap_mc_random_redgreen_equalSizedBins_log}
shows the AP as a function of $\avg[p_{a}]$ of the top 10 ranked tuples according to ground truth by  logarithmic scaling of the x-axis (each point in the plot averages AP over $\approx 450$ experiments for dissociation and lineage and over $\approx 4.5$k experiments for MC). We see that MC performs increasingly poor for ranking answer tuples with probabilities close to $0$ or $1$ and even approach the quality of random ranking (MAP@10 = 0.22). 
This is so because, for these parameters, 
the probabilities of the top 10 answers are very close, and MC needs many iterations to distinguish them.
Therefore, MC performs increasingly poorly for increasing size of lineage but fixed average input probability $\avg [p_{i}] \approx 0.5$, as the average answer probabilities $\avg[p_{a}]$ will be close to $1$. 
In order not to ``bias against our competitor,'' we compared against MC
in its best regime with $0.1< \avg[p_a] <0.9$ in \autoref{Fig_VLDBJ_TPCH_AP_MC_aggregated_09}.

\begin{tcolorbox}[top=-2mm]
\begin{resultW}\label{lesson3}
Ranking by lineage size has good quality only when all input tuples have the same probability.
\end{resultW}
\end{tcolorbox}

\Autoref{Fig_VLDBJ_TPCH_AP_LineageRanking} shows that ranking by lineage is good only when all tuples in the database have the \emph{same} probability
(labeled by $p_i = \textrm{const}$ as compared to $\avg[p_i] = \textrm{const}$). 
This is a consequence of the output probabilities depending mostly on the size of the lineages
if all probabilities are equal.
Dependence on other parameters, such as overall lineage size and magnitude of input probabilities (here shown for $p_i = 0.1$ and $p_i = 0.5$), seem to matter only slightly.

\begin{tcolorbox}[top=-2mm]
\begin{resultW}\label{lesson4}
The quality of dissociation decreases with the average number of dissociations per tuple $\avg[d]$ and with the average input probabilities $\avg[p_{i}]$.
Dissociation performs very well and notably better then MC(10k) if either $\avg[d]$ or $\avg[p_i]$ are small.
\end{resultW}
\end{tcolorbox}

Each answer tuple $a$ gets its score $p_a$ from one of two query plans $P_S$ and $P_P$ that dissociate tuples in tables $S$ and $P$, respectively. For example, if the lineage size for tuple $a$ is 100 and the lineage contains 20 unique suppliers from table $S$ and 50 unique parts from table $P$, then $P_S$ dissociates each tuple from $S$  into 5 tuples and $P_P$ each tuple from $P$ into $2$ tuples, on average. Most often, $P_P$ will then give the better bounds as it has fewer average dissociations. 
Let $\avg[d]$ be the mean number of dissociations for each tuple in the dissociated table of its respective optimal query plan, averaged across all top 10 ranked answer tuples. 
For all our queries (even those with $\$1 = 10k$ and $\$2 = \sql{'\%'}$), $\avg[d]$ stays below $1.1$ as, for each tuple, there is usually one plan that dissociates few variables. 
In order to \emph{understand the impact of higher numbers of dissociations} (increasing $\avg[d]$), we also measured AP for the ranking for \emph{each query plan individually}. 
Hence, for each choice of random parameters, we record two new data points -- one for ranking all answer tuples by using only $P_S$ and one for using only $P_P$ -- together with the values of $\avg[d]$ in the respective table that gets dissociated. This allows us to draw conclusions for a larger set of parameters.
\Autoref{Fig_VLDBJ_TPCH_AP_Diss_MC_Tradeoff1} plots MAP values 
as a function of $\avg[d]$ of the top 10 ranked tuples on the horizontal axis, and various values of $\avg[p_i]$ ($\avg[p_i] = 0.05, 0.10, \ldots, 0.5$).
Each plotted point averages over at least 10 data points (some have 10, other several 1000s). 
Dashed lines show a fitted parameterized curve to the data points on $\avg[p_i]$ and $\avg[d]$.
The figure also shows the standard deviations as shaded areas for $\avg[p_i] = 0.5$.
We see that the quality is very dependent on $\avg[p_i]$, as predicted by \autoref{prop:smallProbabilities}.

\Autoref{Fig_VLDBJ_TPCH_AP_Diss_MC_Tradeoff2} maps the trade-off between dissociation and MC for the two important parameters for the quality of dissociation ($\avg[d]$ and $\avg[p_i]$) and the number of samples for MC. For example, MC(1k) gives a better expected ranking than dissociation only for the small area above the thick red curve marked MC(1k). 
For MC, we used the test results from \autoref{Fig_VLDBJ_TPCH_AP_MC_aggregated_09}; i.e.\ assuming $0.1< \avg[p_a] <0.9$ for MC.
Also recall that for large lineages, having an input probability with $\avg[p_{i}] = 0.5$ will often lead to answer probabilities close to $1$ for which ranking is not possible anymore (recall \autoref{Fig_VLDBJ_TPCH_AP_LineageRanking}). Thus, for large lineages, we need small input probabilities to have meaningful interpretations. And for small input probabilities, dissociation considerably outperforms any other method.

Notice that one should not confuse ($i$) the AP score of each of the two plans taken separately with ($ii$) the AP score of the min between the two plans; the ranking produced by the latter can be much better.  
For example, one experiment ($\$1 = 10k$, and $\$2 = \sql{'\%re\%bl\%re\%'}$) with maximal lineage size 106 has $\avg[d]$ equal $1.053$ and $1.099$ for $P_P$ and $P_S$, respectively. None of the two plans gets perfect AP@10. However, using the minimum score of both plans \emph{for each tuple individually} has $\avg[d] = 1.049$ and perfect AP@10 = 1.
We also evaluated MAP for ranking all tuples by the plan that has the minimal mean $\avg[d]$ as compared to ranking by the minimum scores for each tuple individually. MAP over all 100k data points would then drop from 0.997 (\autoref{Fig_VLDBJ_TPCH_AP_overviewBars}) to only 0.995, which shows the value of taking the minimum score for each tuple \emph{individually}.

\begin{figure}[t]
    \centering
	\begin{minipage}[b]{60mm}
	\setlength{\tabcolsep}{1.5mm}
	\scriptsize
	\begin{tabular}[b]{@{\hspace{1pt}} l | c r | c c @{\hspace{1pt}}  }
		\hline
		\textbf{Dissociation}	&$\textrm{avg}[p_i]$& $\textrm{avg}[d]$ 		&MAP@10	& stdv	\\
						&0.05		& 5 			&0.997	& 0.011	\\
						&0.25		& 2 			&0.967	& 0.036	\\
						&0.50		& 1.1 			&0.968	& 0.035	\\
						&0.50		& 2 			&0.894	& 0.061	\\
						&0.50		& 5 			&0.833	& 0.074	\\																								
		\hline                                          
		\textbf{MC}		& $\textrm{avg}[p_a]$ &trials		&MAP@10	& stdv	\\
						& $0.1 - 0.9$	&10k		&0.964	& 0.040	\\
						& $0.1 - 0.9$	 	&3k			&0.936	& 0.055	\\
						& $0.1 - 0.9$	 	&1k			&0.894	& 0.074	\\
						& $\approx 0.99$	&10k		&0.945	& 0.046	\\
						& $\approx 0.99$ 	&3k			&0.897	& 0.059	\\
						& $\approx 0.99$	&1k			&0.827	& 0.076	\\
		\hline                           		                      	
		\textbf{Lineage size}	& $p_i$ 			&	&MAP@10	& stdv	\\
						& random 			&	& 0.520	& 0.130 \\
						& all equal	 		&	& 0.949	& 0.033 \\						
		\hline                           		                      	
		\textbf{Random ranking}	& 			&	&MAP@10	& stdv			\\
						&			&	& 0.220	& 0.112
	\end{tabular}
	\end{minipage}			
\caption{Quality results TPC-H: Our three methods, their respectively most important parameters, and their average ranking qualities.}\label{Fig_TODS_TPCH_Result_AP_tableOverview}
\end{figure}

\begin{tcolorbox}[top=-2mm]
\begin{questionW}\label{question6}
How much would the ranking change according to exact probabilistic inference if we scale down all input tuples?
\end{questionW}

\begin{resultW}\label{lesson6}
If the probabilities of all input tuples are already small, then scaling them further down does not affect the ranking much.
\end{resultW}
\end{tcolorbox}

This result is a more general statement about the applicability of ranking over probabilistic databases, and motivated by the observation that dissociation works surprisingly well for small input probabilities.
Here, we repeatedly evaluated the exact ranking for 7 different parameterized queries over randomly generated databases with one query plan that has 
$\avg[d] \approx 3$, for two conditions:
first on a probabilistic database with $avg[p_{i}]$ input probabilities (we defined the resulting ranking as GT);
then again on a scaled version, where all input probabilities in the database are multiplied by the same scaling factor $f \in (0,1)$. We then compared the new ranking against GT. 
\Autoref{Fig_VLDBJ_TPCH_AP_ScaledGT} shows that if all input probabilities are already small (and dissociation already works well), then scaling has little 
effect on the ranking. 
However, for $avg[p_{i}]=0.5$ (and thus many tuples with $p_i$ close to $1$), we have a few tuples with $p_i$ close to $1$. These tuples are very influential for the final ranking, but their relative influence decreases if scaled down even slightly. 
Also note that even for $avg[p_{i}]=0.5$, scaling a database by a factor $f = 0.01$ instead of $f = 0.2$ does not make a big difference. However, the quality remains well above ranking by lineage size (!). 
This suggests that the difference between ranking by lineage size (MAP $\approx  0.529$) and the ranking on a scaled database for $f \rightarrow 0$ (MAP $\approx  0.879$) can be attributed to the relative weights of the input tuples (we thus refer to this as ``\emph{ranking by relative input weights}''). The remaining difference in quality then comes from the \emph{actual probabilities} assigned to each tuple.
Using MAP $\approx  0.220$ as baseline for random ranking, 
38\% of the ranking quality can be found by the lineage size alone vs.\ 85\% by the lineage size plus the relative weights of input tuples. The remaining 15\% come from the actual probabilities (\autoref{Fig_VLDBJ_TPCH_AP_overviewBars}).
While these exact numbers only hold for this particular choice of queries and while the implicit assumption that the quality of ranking were a linear scale of MA is debatable, we think that this ``thought experiment'' provides an interesting way to think about ``the value'' of exact probabilistic inference.

\begin{tcolorbox}[top=-2mm]
\begin{questionW}\label{question7}
Does the expected ranking quality of dissociation decrease to random ranking for increasing fractions of dissociation (just like MC does for decreasing number of samples)?
\end{questionW}

\begin{resultW}\label{lesson7}
The expected performance of dissociation for increasing $\avg[d]$ for a particular query is lower bounded by the quality of ranking by relative input weights.
\end{resultW}
\end{tcolorbox}

Here, we 
use a similar setup as before and now compare various rankings against each other:
SampleSearch on the original database (``GT'');
SampleSearch on the scaled database (``Scaled GT'');
dissociation on the scaled database (``Scaled Diss''); and
ranking by lineage size (which is unaffected by scaling).
From \autoref{Fig_VLDBJ_TPCH_AP_ScaledDissociation}, we see
that the quality of Scaled Diss w.r.t.\ Scaled GT $\rightarrow 1$ for $f \rightarrow 0$ since dissociation works increasingly well for small $\avg[p_i]$ (recall \autoref{prop:smallProbabilities}). We also see that Scaled Diss w.r.t.\ GT decreases towards Scaled GT w.r.t.\ GT for $f \rightarrow 0$.
Since dissociation can always reproduce the ranking quality of ranking by relative input weights by first downscaling the database (though losing information about the actual probabilities) the expected quality of dissociation for smaller scales does not decrease to random ranking, but rather to ranking by relative weights. 
Note this result only holds for the expected MAP; any particular ranking can still be very much off.

\section{Related Work}\label{sec:relatedWork}

\introparagraph{Probabilistic databases} 
Query evaluation over probabilistic databases corresponds to solving the weighted model counting problem, 
and current approaches 
can be classified into three
categories (\autoref{TableWorkComparison}): 
(1) \emph{incomplete approaches} identify tractable
cases either at the
query-level~\cite{DBLP:journals/vldb/DalviS07,DBLP:journals/jacm/DalviS12,FinkO:PODS2014dichotomy,DBLP:conf/icde/OlteanuHK09}
or the
data-level~\cite{DBLP:conf/sum/OlteanuH08,DBLP:conf/icdt/RoyPT11,SenDeshpandeGetoor2010:ReadOnce} and ignore the rest;
(2) \emph{exact approaches}
\cite{DBLP:conf/icde/AntovaJKO08,JhaOS2010:EDBT,DBLP:conf/icde/SenD07}
are based on variants and extensions of a complete search based on the DPLL procedure~\cite{DBLP:series/faia/GomesSS09}
and work well for queries over databases with simple lineage expressions, but perform
poorly on complex lineage expressions; and
(3) \emph{approximate approaches} 
usually first compute the lineage of the query on the given database to obtain a Boolean formula, 
then  
either apply variants of Monte Carlo 
sampling
methods~\cite{DBLP:conf/sigmod/JampaniXWPJH08,DBLP:journals/vldb/JoshiJ09,DBLP:conf/icde/KennedyK10,DBLP:conf/icde/ReDS07},
or approximate the number of models of the Boolean lineage
expression~\cite{DBLP:dblp_conf/icdt/FinkO11,OlteanuHK2010:ICDE,DBLP:journals/pvldb/ReS08}.
A recent approach combines safe plans with Monte Carlo simulation~\cite{DBLP:journals/pvldb/GribkoffS16}.
An approximate ``anytime method'' based on DPLL search is developed in \cite{DBLP:journals/vldb/FinkHO13} 
that stops the full search whenever a given confidence bound can be guaranteed.
This approach allows evaluating a query to a precision determined by a given computational budget.
A variant of this method with confidence bounds over first-order lineage formulas is developed in \cite{DBLP:conf/icde/DyllaMT13}.
Our work can be seen as a generalization of some of
of these techniques: our algorithm returns the exact probability if the
query is
safe~\cite{DBLP:journals/vldb/DalviS07,OlteanuHK2010:ICDE} or
data-safe~\cite{JhaOS2010:EDBT} and
gives a unique and well-defined value for every query.
This property can be useful when learning the probabilities from queries.
In addition, our method can be used together with any existing relational database 
without any modifications to the engine.
On the other side, our query-centric approach currently works only for self-join-free conjunctive queries and does not allow an iterative refinement or a trade-off between computational complexity and precision for applications where the exact probability scores are required.

\introparagraph{Lifted and approximate inference} Lifted inference was
introduced in the AI literature as an approach to probabilistic
inference that uses the first-order formula to exploit symmetries at
the grounded level~\cite{DBLP:conf/ijcai/Poole03}.  This research
evolved independently of that on probabilistic databases, and the two
have many analogies: A formula is called \emph{domain liftable} iff
its data complexity is in polynomial time~\cite{jaeger-broeck-2012},
which is the same as a \emph{safe query} in probabilistic databases,
and the FO-d-DNNF circuits described
in~\cite{DBLP:conf/ijcai/BroeckTMDR11} correspond to the safe plans
discussed in this paper.  See~\cite{BroeckSuciu:UAI2014tutorial} for a
recent discussion on the similarities and differences.

\introparagraph{Representing correlations} The most popular approach
to represent correlations between  tuples in a probabilistic
database is by a Markov Logic network (MLN) which is a set of
\emph{soft constraints}~\cite{DBLP:series/synthesis/2009Domingos}.
Quite remarkably, all complex correlations introduced by an MLN can be
rewritten into a query over a tuple-independent probabilistic
database~\cite{DBLP:conf/uai/GogateD11a,BroeckMD:KR2014,DBLP:journals/pvldb/JhaS12}.
In combination with such rewritings, our techniques can be also
applied to MLNs if their rewritings results in conjunctive
queries without self-joins.

\introparagraph{Dissociation} 
In a graph-based scenario \cite{DetwilerGLST2009:ICDE} that basically corresponds to our abstracted \autoref{ex:1}, 
we observed that propagation-based methods often perform as well as reliability-based methods for predicting protein functions from integrated uncertain biological databases.
We then first introduced dissociation in the
workshop paper~\cite{GatterbauerJS2010:MUD} as an attempt to generalize the success of propagation methods from graphs to hypergraphs.  
\cite{DBLP:journals/tods/GatterbauerS14} provides a general framework for approximating the probability of Boolean functions with both upper and lower bounds. We also illustrate how upper bounds to hard queries can be complemented by lower bounds (those lower bounds, however are not as tight, which is why we only use upper bounds for ranking in this paper).
Dissociation is closely related to a number of recent approaches in the graphical model and constraint satisfaction literature which approximate an intractable problem with a tractable relaxed version after \emph{treating multiple occurrences of variables or nodes as independent} or ignoring some equivalence constraints. Those approaches are usually referred to as \emph{relaxation}~\cite{DBLP:conf/uai/BroeckCD12}
(see~\cite{DBLP:journals/tods/GatterbauerS14} for a detailed discussion on 
similarities and differences).

\begin{figure}[t]
\setlength{\tabcolsep}{1.8mm}
\renewcommand{\arraystretch}{1.1}	
\begin{tabular}[t]{@{\hspace{0pt}} r @{\hspace{1mm}} | >{$}c<{$} >{$}c<{$} >{$}c<{$} >{$}c<{$} >{$}c<{$}
	@{\hspace{2mm}}    l    >{$}c<{$} >{$}c<{$} >{$}c<{$} >{$}c<{$} >{$}c<{$} }
		\multicolumn{1}{l|}{}&\multicolumn{8}{l}{
		\!\!\!\!\!\;
		\begin{turn}{55}all queries\end{turn}
		\!\!\!\!\!\!\!\!\!\!\!
		\begin{turn}{55}all data\end{turn}	
		\!\!\!\!\!\!
		\begin{turn}{55}unique score\end{turn}	
		\!\!\!\!\!\!\!\!\!\!\!\!
		\begin{turn}{55}performance\end{turn}
		\!\!\!\!\!\!\!\!\!\!\!\!\!\!
		\begin{turn}{55}rel.\ algebra\end{turn}
		}\\						
        \cline{1-6}
	Safe query plans \cite{DBLP:journals/vldb/DalviS07,DBLP:journals/jacm/DalviS12}
		&	&\bullet 		&\bullet &\bullet &\bullet &
		\multirow{2}*{\Big\}\! \scriptsize(1) incomplete}\\
	Read-once formulas \cite{DBLP:conf/sum/OlteanuH08,DBLP:conf/icde/OlteanuHK09,SenDeshpandeGetoor2010:ReadOnce}		
		&\bullet	& 		&\bullet &\bullet 	 	\\		
	Exact prob.\ inference \cite{JhaOS2010:EDBT}
		&\bullet 	&\bullet 	&\bullet 	& &	
		&\} \scriptsize(2) slow		\\			
	Monte Carlo \cite{DBLP:conf/sigmod/JampaniXWPJH08,DBLP:conf/icde/KennedyK10,DBLP:conf/icde/ReDS07} 
		&\bullet 	&\bullet 	& 	&\circ 	&
		&\multirow{2}*{\Big\}\! \scriptsize(3) approximate}	  		\\		
	Approx.\ mod.\ count.\ \cite{DBLP:conf/icde/DyllaMT13,DBLP:journals/vldb/FinkHO13,OlteanuHK2010:ICDE} 
		&\bullet 	&\bullet 	& 	&\circ	\\				
\end{tabular}
\caption{Current techniques for evaluating probabilistic queries are either (1) \emph{incomplete} and work only on a subset of queries and data instances, or (2) always work but may become arbitrarily \emph{slow} on general data instances, or (3) only \emph{approximate} the actual score.}
\label{TableWorkComparison}
\end{figure}

\section{Conclusions and Outlook}\label{sec:conclusion}

This paper developed a new scoring function called \emph{propagation} for ranking query results over probabilistic databases. Our semantics is based on a sound and principled theory of \emph{query dissociation}, and can be evaluated efficiently in an off-the-shelf relational database engine for \emph{any type of self-join-free conjunctive query}. 
We proved that the propagation score is an upper bound to query reliability, 
that both scores coincide for safe queries, 
and that propagation naturally extends the case of safe queries to unsafe queries.
We further showed that the scores for chain queries before and after dissociation correspond to two well-known scoring functions on graphs, namely network reliability (which is \#P-hard) and propagation (which is related to \mbox{PageRank} and in PTIME), and that our dissociation scores are thus generalizations of the propagation score from graphs to hypergraphs.
We calculated the propagation score
by evaluating a fixed number of safe queries, each providing an upper bound on the true probability, then taking their minimum. 
We provided algorithms that takes into account schema information to enumerate only the minimal necessary plans among all possible plans, and prove our method to be a strict generalization of all known results of PTIME self-join free conjunctive queries. 
We described relational query optimization techniques that 
allow us to evaluate all minimal queries in a single query and very fast.
Our evaluations show that the optimizations of our approach bring probabilistic query evaluation close to standard query evaluation while providing high ranking quality.
In future work, we plan to generalize the approach to full first-order queries.

\begin{acknowledgements}
This work was supported in part by 
NSF grants IIS-0513877, IIS-0713576, IIS-0915054, IIS-1115188, IIS-1247469,
and CAREER IIS-1553547.
We like to thank Abhay Jha for help with the experiments in the workshop version of this paper, Alexandra Meliou for suggesting the name ``dissociation'', and Vibhav Gogate for guidance in using his tool SampleSearch.
WG would also like to thank Manfred Hauswirth for a small comment in 2007 that was crucial for the development of the ideas in this paper.
\end{acknowledgements}

\bibliography{\bibpath}

\clearpage
\appendix
\section{Nomenclature}\label{sec:nomenclature}

\vspace{-2mm}

\capstartfalse			
\begin{table}[h!]
    \centering
    \small
    \begin{tabularx}{\linewidth}{  @{\hspace{0pt}} >{$}l<{$}  @{\hspace{1mm}}  X @{}}  
    \hline
	R,S,T,U		& relational tables \\
	r_i, s_i, t_i, u_i & tuple identifiers \\
	A,B,C		& attribute names \\
	a,\ldots, f,s,t	& constants \\
	s,t			& source and target nodes \\  
	t			& a tuple \\	
	x,y,z		& variables \\	
	q			& query	\\
	a_i			& atom \\	
	\at(x_j)	& set of atoms that contain variable $x_j$	\\
	\Var(a_i) 	& set of variables of a query $q$ or atom $a_i$ \\
	\HVar(P) 	& set of head variables of a query $q$ or a plan $P$ \\	
	\EVar(q) 	& set of existential variables: $\EVar(q) \!=\! \Var(q) \!-\! \HVar(q)$ \\		
	\MinCuts(q)	& set of minimal subsets of $\EVar(q)$ that disconnects $q$  \\
	\MinPCuts(q)& set of minimal p-cuts \\	
	\SepVar(q)	& existential variables that appear in all atoms  \\		
	\SepPVar(q)	& separator variables that appear in all probabilistic atoms \\				
	p			& probability function	\\
	r(q)		& reliability score of $q$\\
	\rho(q)		& propagation score of $q$\\
	\phi, \psi	& Boolean expression 	\\
	\PP{\phi}	& probability of a Boolean expression		\\
	m			& number of subgoals \\		
	D			& database  \\
	\ADom_x		& active domain of variable $x$ \\
	\Delta		& collection of sets of variables $\Delta = (\vec y_1, \ldots, \vec y_m)$ \\
	R_i^{\vec y_i} & dissociated relation $R_i(\vec x_i)$ on variables $\vec y_i$: $R_i(\vec x_i, \vec y_i)$ \\	
	q^{\Delta}	& dissociated query\\
	P			& query plan	\\
	\mathcal{P}	& set of plans	\\	
	\joinp{}{\ldots}	& probabilistic join operator in prefix notation	\\
	\projpd{\vec x},\projp{\vec y}	& probabilistic project operators: onto $\vec x$, or project $\vec y$ away \\
	\textit{score}(P)	& score of a query plan	\\
	\JVar		& join variables for a join operator \\
	\vec x		& unordered set or ordered tuple \\
	{[a/x]}		& substitute value $a$ for variable $x$ \\
    \hline	
    \end{tabularx}
\end{table}
\capstarttrue

\vspace{-3mm}

\section{\autoref{sec:backgroundPDBs}: 
Proof
Safety}

\begin{proof}[\autoref{prop:uniqueSafePlan}: Safety]
(1) is proven in \cite{DBLP:journals/vldb/DalviS07}; we prove here only (2):

(a): 
Hierarchical query $\Rightarrow$ unique safe plan:
	We prove the following statement by induction:  
	Let $\vec x = \SepVar(q)$ be the set of \emph{separator variables} for a
  query $q(\vec x)$, i.e.\ every variable in $\vec x$ occurs in every
  atom in $q$; then $q(\vec x)$ admits a unique safe plan either as
  $\joinp{}{P_1, \ldots, P_k}$ or as $\projpd{\vec x} P$: Define a graph where the nodes are
  the atoms of $q$ and any two nodes are connected by an edge iff they
  share an existential variable, i.e.\ a variable not occurring in $\vec x$.  If the graph has $k$
  query components represented by the queries $q_1, \ldots, q_k$,
  then $q \equiv \,\,\joinp{}{q_1, \ldots, q_k}$, and we apply induction
  hypothesis to each $q_i(\vec x)$.  If the graph has a single
  query component with additional variables  $\vec y \neq \emptyset$,
then $q \equiv \projpd{\vec x} q(\vec x, \vec y)$, and we apply
  induction hypothesis to $q(\vec x, \vec y)$.  Finally, if the graph has a
  single component and only the variables $\vec x$, then $q$ has a
  single atom, hence $q \equiv R(\vec x)$. 

(b):
Safe plan $\Rightarrow$ hierarchical query: We construct inductively $q$ from its  derivation by noting that $\joinp{}{P_1, \ldots, P_k}  \equiv  q_1 \wedge \ldots \wedge q_k$, 
  where $q_1, \ldots, q_k$ are the hierarchical queries obtained
  inductively from $P_1, \ldots, P_k$, 
	and $\projpd{\vec x} P  \equiv  \exists \vec x. q$, where $q$ is obtained inductively
  from $P$.  It is easy to
  check inductively that all resulting queries are hierarchical.\qed
\end{proof}

\section{\autoref{sec:graphConnection}: 
Proof
\Autoref{prop:connectionPropagationScore}:
Connection to networks}\label{app:connectionGraphQueryPropagation}

\begin{proof}[\autoref{prop:connectionPropagationScore}: connection to networks]
Notice that we use digraphs to enforce that each path from $s$ to $t$ has exactly $k$ edges.

(a) We first establish the connection between \emph{graph reliability} and \emph{query reliability} $r(q) \define \PP{q}$.
  The first claim is obvious: a possible world contains a path from
  $s$ to $t$ iff the query is true on that world:
The chain query is true exactly if, in a randomly chosen world, there is a set of tuples of each relation that forms at least one output tuple. This corresponds exactly to the graph reliability, i.e.\ the probability that the nodes $s$ and $t$ are connected in a randomly chosen subgraph.

(b) We next establish the connection between \emph{propagation in networks} and \emph{dissociation in databases}.
Consider the unique safe query plan for the dissociated query $q^{\Delta}$:
	\begin{align*}
		P	& =
			\projp{x_{k}} 	\!	\joinp{}{R_k^{\emptyset}(x_{k},t),
	        \projp{x_{k\!-\!1}}		\!	\joinp{}{R_{k-1}^{\emptyset}  (\vec x_{[k-1,k]}) \ldots , \\
	& \quad\; \projp{x_{2}} 	\!	\joinp{}	
			{R_1^{\vec x_{[3,k]}}  (s,\vec x_{[2,k]}), 
			R_2^{\vec x_{[4,k]}}  (\vec x_{[2,k]})   }\ldots}}
	\end{align*}
This plan is evaluated from the inside out. The table $R_1$ is dissociated on all variables except $x_2$, i.e.\ each consequent project on previous join results from $R_1$ will treat each tuple as independent. The independent project corresponds exactly to the way propagation is calculated at each node iteratively from the probabilities of its parents and incoming edges. 
  We prove this by induction on $k$:  When $k=1$ then $R_1$ contains a single
  edge $(s,t)$, whose probability is equal to $r(q)$, to $r(q^\Delta)$,
  and to the network propagation score.   
  To prove for $k\geq 1$, let
  $V_{k} = \set{a_1, \ldots, a_n}$ be the nodes in the before-last
  partition (the last partition is $V_{k+1} = \set{t}$).  

{	\begin{center}
		\includegraphics[scale=0.5]{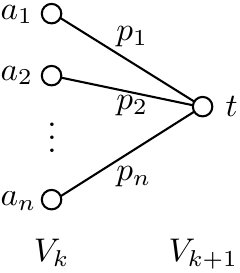}
	\end{center}
}

\noindent
We defined
  in \autoref{ex:1} the network propagation score to be:
  \begin{align*}
    \rho(t) & =  \bigotimes_i \rho(a_i) \cdot p_i
  \end{align*}
  where $p_i$ is the probability of the edge $(a_i,t)$.  On the other hand, the reliability of the dissociated query is given by the following formula which represents 
a probabilistic join with $R_k(x_{k}, t)$, 
followed by 
a probabilistic projection on the variable $x_{k}$, 
and where an expression $[a/x]$ stands for substitution of a variable $x$ by a constant $a$:
\begin{align*}
  r(q^\Delta) & =  \bigotimes_i r(q^\Delta[a_i/x_k]) \\
		&=  \bigotimes_i r \big(R_1(s,\vec x_{[2,k\!-\!1]}, a_i),\ldots,R_{k\!-\!1}(x_{k\!-\!1},a_i )) 
		    \cdot r(R_k(a_i,t)  \big) \\
   & = \bigotimes_i \rho(a_i) \cdot p_i    \qedhere
\end{align*}
\end{proof}

\begin{figure}[t]
\renewcommand{\tabcolsep}{0.65mm}
\renewcommand{\arraystretch}{0.9}
\centering
\hspace{0mm}
\subfloat[$G$]{
	\begin{minipage}[b]{49mm}
	\vspace{0mm}
	\begin{minipage}[t]{50mm}
	\vspace{0mm}
	\includegraphics[scale=0.72]{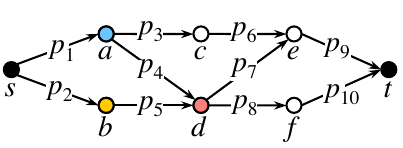}
	\vspace{0mm}
	\end{minipage}
	\vspace{0mm}
	\end{minipage}
	\label{Fig_AppendixExample_a}	
}
\hspace{10mm}		
\subfloat[$q$]{
	\begin{minipage}[b]{16mm}
	\vspace{0mm}
	\begin{minipage}[t]{15mm}
	\vspace{0mm}
	\begin{tabular}[b]{@{\hspace{1pt}} >{$}c<{$}|>{$}c<{$} >{$}c<{$} >{$}c<{$} >{$}c<{$} @{\hspace{1pt}}}
			& x	& y & z \\
		\hline
		R_1	& \circ	  \\		
		R_2	& \circ &  \circ \\
		R_3	& & \circ	& \circ \\
		R_4	& & & \circ 		
	\end{tabular}
	\vspace{0mm}
	\end{minipage}
	\vspace{0mm}
	\end{minipage}
	\label{Fig_AppendixExample_b}	
}
\hspace{0mm}
\subfloat[$G_1$]{
	\begin{minipage}[b]{49mm}
	\vspace{0mm}
	\begin{minipage}[t]{50mm}
	\vspace{0mm}
	\includegraphics[scale=0.72]{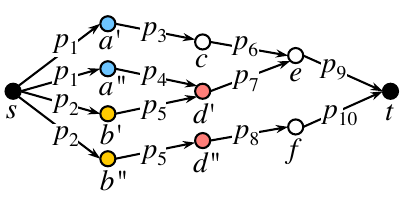}
	\vspace{0mm}
	\end{minipage}
	\vspace{-2mm}
	\end{minipage}
	\label{Fig_AppendixExample_c}	
}
\hspace{10mm}		
\subfloat[$q^{\Delta_1}$]{
	\begin{minipage}[b]{16mm}
	\vspace{0mm}
	\begin{minipage}[t]{15mm}
	\vspace{0mm}
	\begin{tabular}[b]{@{\hspace{1pt}} >{$}c<{$}|>{$}c<{$} >{$}c<{$} >{$}c<{$} >{$}c<{$} @{\hspace{1pt}}}
			& x	& y & z \\
		\hline
		R_1	& \graycell\circ	&  \graycell\bullet &  \graycell\bullet  \\		
		R_2	& \graycell\circ &  \graycell \circ & \graycell\bullet \\
		R_3	& & \graycell\circ	& \graycell\circ \\
		R_4	& & & \graycell\circ 		
	\end{tabular}
	\vspace{0mm}
	\end{minipage}
	\vspace{3mm}
	\end{minipage}
	\label{Fig_AppendixExample_d}
}
\hspace{0mm}
\subfloat[$G_2$]{
	\begin{minipage}[b]{49mm}
	\vspace{0mm}
	\begin{minipage}[t]{50mm}
	\vspace{0mm}
	\includegraphics[scale=0.72]{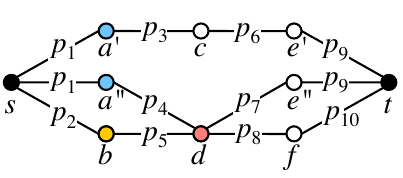}
	\vspace{0mm}
	\end{minipage}
	\vspace{-2mm}
	\end{minipage}
	\label{Fig_AppendixExample_e}
}
\hspace{10mm}		
\subfloat[$q^{\Delta_2}$]{
	\begin{minipage}[b]{16mm}
	\vspace{0mm}
	\begin{minipage}[t]{15mm}
	\vspace{0mm}
	\begin{tabular}[b]{@{\hspace{1pt}} >{$}c<{$}|>{$}c<{$} >{$}c<{$} >{$}c<{$} >{$}c<{$} @{\hspace{1pt}}}
			& x	& y & z \\
		\hline
		R_1	& \graycell\circ	&  \graycell\bullet     \\		
		R_2	& \graycell\circ &  \graycell \circ   \\
		R_3	& & \graycell\circ	& \graycell\circ \\
		R_4	& &\graycell\bullet & \graycell\circ 		
	\end{tabular}
	\vspace{0mm}
	\end{minipage}
	\vspace{3mm}
	\end{minipage}
	\label{Fig_AppendixExample_f}
}
\renewcommand{\tabcolsep}{0.4mm}
\hspace{5mm}		
\subfloat[$D$]{
	\begin{minipage}[b]{62mm}
	\vspace{-1mm}
	\begin{minipage}[t]{61mm}
	\hspace{2.7mm}
		\mbox{
				\begin{tabular}[t]{ >{$}c<{$} | >{$}c<{$} >{$}c<{$} }
	 			R_1	& E		& A	\\
				\hline
				p_1	& s		& a	\\
				p_2	& s		& b							
				\end{tabular}			
		}
		\hspace{5mm}
		\mbox{
				\begin{tabular}[t]{ >{$}c<{$} | >{$}c<{$} >{$}c<{$} >{$}c<{$} >{$}c<{$}}
	 			R_2	& A		& B\\
				\hline
				p_3	& a		& c		\\
				p_4	& a		& d		\\
				p_5	& b		& d				
				\end{tabular}
		}
		\hspace{0.5mm}
		\mbox{
				\begin{tabular}[t]{ >{$}c<{$} | >{$}c<{$} >{$}c<{$}}
	 			R_3	& B		& C	\\
				\hline
				p_6	& c		& e	\\	
				p_7	& d		& e	\\
				p_8	& d		& f		
				\end{tabular}			
		}
		\hspace{-0.5mm}
		\mbox{
				\begin{tabular}[t]{ >{$}r<{$} | >{$}c<{$} >{$}c<{$}}
	 			R_4	& C		& E	\\
				\hline
				p_9		& e		& t\\	
				p_{10}	& f		& t
				\end{tabular}			
		}
	\end{minipage}	
	\vspace{-1mm}
	\end{minipage}	
	\label{Fig_AppendixExample_g}
}		
\hspace{5mm}		
\subfloat[$D^{\Delta_1}$]{
	\begin{minipage}[b]{63mm}
	\vspace{-1mm}
	\begin{minipage}[t]{62mm}
		\mbox{
				\begin{tabular}[t]{ >{$}r<{$} | >{$}c<{$} >{$}c<{$}  >{$}c<{$} >{$}c<{$} }
	 			R_1^{\{y,z\}}	& E		& A 	& B		& C	\\
				\hline
				p_1	& s		& a		& c		& e \\
				p_1	& s		& a		& d		& e \\				
				p_2	& s		& b		& d		& e \\					
				p_2	& s		& b		& d		& f 				
				\end{tabular}			
		}
		\hspace{-2mm}
		\mbox{
				\begin{tabular}[t]{ >{$}c<{$} | >{$}c<{$} >{$}c<{$} >{$}c<{$} >{$}c<{$}}
	 			R_2^{\{z\}}	& A		& B		& C\\
				\hline
				p_3			& a		& c		& e	\\
				p_4			& a		& d		& e	\\
				p_5			& b		& d		& e \\			
				p_5			& b		& d		& f \\				
				\end{tabular}
		}
		\hspace{-2mm}
		\mbox{
				\begin{tabular}[t]{ >{$}c<{$} | >{$}c<{$} >{$}c<{$}}
	 			R_3	& B		& C	\\
				\hline
				p_6	& c		& e	\\	
				p_7	& d		& e	\\
				p_8	& d		& f		
				\end{tabular}			
		}
		\hspace{-0.5mm}
		\mbox{
				\begin{tabular}[t]{ >{$}r<{$} | >{$}c<{$} >{$}c<{$}}
	 			R_4	& C		& E	\\
				\hline
				p_9		& e		& t\\	
				p_{10}	& f		& t
				\end{tabular}			
		}
	\end{minipage}	
	\vspace{-1mm}
	\end{minipage}	
	\label{Fig_AppendixExample_h}
}	
\hspace{5mm}		
\subfloat[$D^{\Delta_2}$]{
	\begin{minipage}[b]{63mm}
	\vspace{-1mm}
	\begin{minipage}[t]{62mm}
	\hspace{0.7mm}		
		\mbox{
				\begin{tabular}[t]{ >{$}r<{$} | >{$}c<{$} >{$}c<{$}  >{$}c<{$} }
			 	R_1^{\{y\}}	& E		& A 	& B			\\
				\hline
				p_1	& s		& a		& c		 \\
				p_1	& s		& a		& d		 \\
				p_2	& s		& b		& d
				\end{tabular}
		}
		\hspace{1mm}
		\mbox{
				\begin{tabular}[t]{ >{$}c<{$} | >{$}c<{$} >{$}c<{$} >{$}c<{$} >{$}c<{$}}
			 	R_2^{\{z\}}	& A		& B		& C\\
				\hline
				p_3			& a		& c		& e	\\
				p_4			& a		& d		& e	\\
				p_5			& b		& d		& e \\
				p_5			& b		& d		& f \\
				\end{tabular}
		}
		\hspace{-2mm}
		\mbox{
				\begin{tabular}[t]{ >{$}c<{$} | >{$}c<{$} >{$}c<{$}}
			 	R_3	& B		& C	\\
				\hline
				p_6	& c		& e	\\
				p_7	& d		& e	\\
				p_8	& d		& f
				\end{tabular}
		}
		\hspace{-2mm}
		\mbox{
				\begin{tabular}[t]{ >{$}r<{$} | >{$}c<{$} >{$}c<{$} >{$}c<{$}}
			 	R_4^{\{y\}}	& B		&C		& E	\\
				\hline
				p_9			& c		& e		& t\\
				p_9			& d		& e		& t\\
				p_{10}		& d		& f		& t
				\end{tabular}
		}
	\end{minipage}	
	\vspace{-1mm}
	\end{minipage}	
	\label{Fig_AppendixExample_i}
}		
\caption{\specificref{Example}{ex:app1}. 
The \emph{propagation score} $\rho(t)$ of graph $G$ is identical to both
the \emph{reliability score} $r(t)$ 
and the \emph{propagation score} $\rho(t)$ of graph $G_1$, that has nodes $a$, $b$, and $d$ dissociated. 
The remaining correspondences are explained in the text.
}\label{Fig_AppendixExample}
\end{figure}

\begin{example}[Propagation in 5-partite digraphs]\label{ex:app1}
We illustrate here with a concrete example a number of correspondences (see \autoref{Fig_Correspondences}), in particular the correspondences between 
(1) the reliability of a $k+1$-partite digraph
and the probability of a conjunctive $k$-chain query, and between 
(2) the propagation score of a $k+1$-partite digraph as defined in \autoref{ex:1}
and the reliability of a dissociated $k+1$-partite digraph.

For correspondence (1), consider the graph $G$ in \autoref{Fig_AppendixExample_a}: 
The \emph{reliability score} $r(t)$ (i.e.\ the two-terminal network reliability between source node $s$ and target node $t$) 
corresponds to the \emph{query reliability} of 
4-chain query $q \datarule R_1(s,x),R_2(x,y),R_3(y,z),R_4(z,t)$ 
(shown in \autoref{Fig_AppendixExample_a} by its incidence matrix)
over the database $D$ of \autoref{Fig_AppendixExample_g}. 
Notice how each tuple in $D$ corresponds to one edge in $G$ and has the same probability.

For correspondence (2), consider the graph $G_1$ in \autoref{Fig_AppendixExample_c}, that results from $G$ after dissociating nodes $a$, $b$, and $d$. It is a serial-parallel graph and its \emph{reliability score} can easily be calculated as 
$r(t) = 
(p_1 p_3 p_6 \otimes (p_1 p_4 \otimes p_2 p_5)p_7)p_9 
\otimes p_2 p_5 p_8 p_{10}$.
This score is also identical to the propagation score $\rho(t)$ in both $G$ and $G_1$.

For other correspondences, notice that the \emph{reliability score} of graph $G_1$ corresponds to the probability of the dissociated query $q^{\Delta_1}$ from \autoref{Fig_AppendixExample_d} over the database $D^{\Delta_1}$ from \autoref{Fig_AppendixExample_h}.

Further notice that $q^{\Delta_2}$ from \autoref{Fig_AppendixExample_f} is yet another dissociated query 
for which the probability over database $D^{\Delta_2}$ from from \autoref{Fig_AppendixExample_i} 
corresponds to the reliability score $r(t)$ of graph $G_2$. 
Yet there is no intuitive interpretation of a propagation score over the same graph $G_2$.
\end{example}

\begin{figure}[t]
\centering
\includegraphics[scale=0.42]{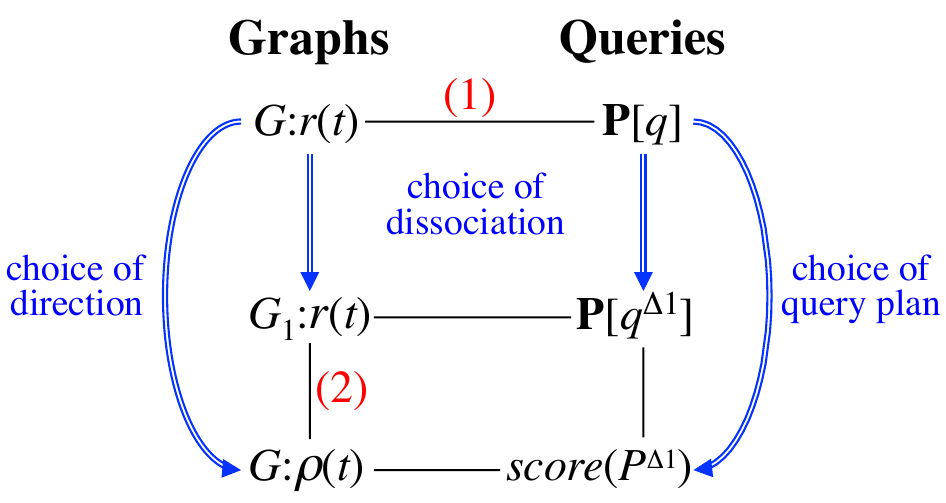}
\caption{\specificref{Example}{ex:app1}. 
Various correspondences between graphs and queries, and between reliability and propagation scores.
}
\label{Fig_Correspondences}
\end{figure}

\section{\autoref{sec:partialDissociationOrder}: 
Proof
Partial dissociation order}

\begin{proof}[\autoref{th:partialDissociationOrder}: Partial dissociation order] 
We already know that the following direction holds: 
\begin{equation*}
	\Delta \succeq \Delta'   \,\,\Rightarrow\,\,  \forall D: \PP{q^{\Delta}} \geq \PP{q^{\Delta'}}
\end{equation*}
i.e.\ whenever two dissociations are comparable, then we know which has a higher reliability scores on every database  $D$. This follows from 
\autoref{th:bool:dissoc} and the observation that there is a one-to-one correspondence between a query dissociations and a positive DNF dissociation. Concretely, at the level of its \emph{lineage expression}, a query dissociation is a dissociation of a $k$-partite DNF.

We next prove the other direction:
\begin{equation*}
	\Delta \succeq \Delta'   \,\,\Leftarrow\,\,  \forall D: \PP{q^{\Delta}} \geq \PP{q^{\Delta'}}
\end{equation*}
i.e.\ whenever  $\PP{q^{\Delta}} \geq \PP{q^{\Delta'}}$ holds for two dissociations over any database $D$, then $\Delta \succeq \Delta'$ in the partial dissociation order. In other words, for any two dissociations  $\Delta$ and $\Delta'$ of a query $q$ that are incomparable, i.e. $\Delta \not \succeq \Delta'$ and $\Delta \not \preceq \Delta'$, there exist two databases so that the dissociated probability of either dissociation becomes bigger. We prove this by showing the contrapositive
\begin{equation*}
	\Delta \not \succeq \Delta'   \,\,\Rightarrow\,\,  \exists D: \PP{q^{\Delta}} < \PP{q^{\Delta'}} 
\end{equation*}
i.e.\ for any two dissociations  $\Delta$ and $\Delta'$ of $q$,
with $\Delta \not \succeq \Delta'$,
there exists a database so that $\PP{q^{\Delta}} < \PP{q^{\Delta'}}$.
We will achieve this by constructing a database $D$ so that $\PP{q^{\Delta}} = \PP{q}$, but $\PP{q^{\Delta'}} > \PP{q}$:
Since $\Delta \not \succeq \Delta'$, there must exist one variable $x_i \in \vec y_j$ that is dissociated in relation $R_j$ of $\Delta'$ which is not dissociated in $\Delta$. W.l.o.g.\ let $x_1 \in \Var(q)$ be this variable, and $R_1$ be the relation.
W.l.o.g.\ consider the active domain $A$ of the database to be $\{a,b\}$. Then construct the following database $D$: 
\begin{enumerate}[label=\textup{\arabic*.}, itemsep=0pt, parsep=0pt, topsep = 0pt]
			
	\item For each relation $R_i \neq R_1$ with $x_1 \not \in \Var(R_i)$, insert one deterministic tuple ($p=1$) with '$a$' as value for each attribute $x_i \in \Var(R_i)$:
	\begin{align*}
		&R_i(a,a,\ldots, a), p = 1
	\end{align*}

	\item For each relation $R_i \neq R_1$ with $x_1\in \Var(R_i)$, insert two deterministic tuples ($p=1$) with '$a$' as value for each attribute $x_i \neq x_1$, and either '$a$' or '$b$' as value for attribute $x_1$:
	\begin{align*}
		&R_i(a,a,\ldots, a), p = 1 \\
		&R_i(b,a,\ldots, a), p = 1
	\end{align*}
		
	\item For relation $R_1$ with $x_1 \not \in \Var(R_i)$, insert one uncertain tuple ($p=0.5$) with '$a$' as value for each attribute:
	\begin{align*}
		&R_1(a,a,\ldots, a), p = 0.5 
	\end{align*}	
		
\end{enumerate}

Then for every dissociation $\Delta$ for which $x_1 \not \in \vec y_1$:  $\PP{q^{\Delta}} = \PP{q} = 0.5$. This follows from two facts: (\textit{i}) certain tuples that get dissociated do not change the query probability; (\textit{ii}) the only probabilistic tuple in $R_1$ only gets dissociated if the dissociation includes $x_1$ since all other variables only include one single value in the active domain. On the other hand, for every dissociation $\Delta'$ for which $x_1 \in \vec y_1$:  $\PP{q^{\Delta'}} = 0.75 > \PP{q}$.
Hence, we have shown that for $D$: $\PP{q^{\Delta}} < \PP{q^{\Delta'}}$.
\qed
\end{proof}

\section{\autoref{sec:queryPlans}: 
Proof
Hierarchical dissociation}

\begin{proof}[\autoref{th:safeDissociation}: Hierarchical dissociation]
Let $P = f(\Delta)$ be the plan corresponding to a hierarchical dissociation $\Delta$, and
let 
$\Delta = g(P)$ 
be the hierarchical dissociation corresponding to a plan $P$.
To prove the isomorphism, we have to show both directions:

(a) $g(f(\Delta)) = \Delta$: Consider a hierarchical dissociation $q^{\Delta}$ and denote its corresponding unique safe plan $P_{\Delta}$. 
This plan uses dissociated relations, hence each relation $R^{\vec y_i}_i(\vec x_i, \vec y_i)$ has extraneous variables $\vec y_i$. 
If we drop all variables $\vec y_i$ from the relations, then this transforms $P_{\Delta}$ into a regular (unsafe) plan $P$ for $q$. 
If we now consider all those variables $\vec y_i$ that we have thus dropped from a relation $R_i$, then these are exactly those variables that are added by recursively adding all existential variables $\JVar - \HVar(P_j)$ for each join operator.

(b) $f(g(P)) = P$: Consider a possibly unsafe plan $P$ for $q$ and recursively dissociate each relation $R_i$ occurring in a subplan $P_j$ of a join operation $\joinp{}{P_1,\ldots,P_k}$ on the missing existential variables $\JVar - \EVar(P_j)$. 
Then this defines a unique hierarchical dissociation $\Delta$ of $q$. Now consider the unique safe plan for $q^{\Delta}$. 
Since the safe plan $P_{\Delta}$ is the only plan for which all subplans share the same head variables, it must be the same plan as the original unsafe plan for $q$.
\qed
\end{proof}

\section{\autoref{sec:enumeratingMinimalQueryPlans}: 
Proof
\Autoref{alg:basicAlgorithm}
}

\begin{proof}[\autoref{prop:algorithmSound}: \Autoref{alg:basicAlgorithm}]
First, recall from the proof of \autoref{th:safeDissociation} that we go from a plan $P$ to a hierarchical dissociation by recursively dissociating each relation $R_i$ occurring in a subplan $P_j$ of a join operation $\joinp{}{P_1,\ldots,P_k}$ on the missing existential variables $\JVar - \HVar(P_j)$. Hence, in a project plan $P = \projp{\vec x} P'$, all relations occurring in $P'$ either already contain or are dissociated on each variable in $\vec x$. This is since $\vec x$ are the join variables if $P'$ is a join plan, and it holds trivially if $P'$ is a single relation.

(1) \emph{Soundness}: We show that every plan produced by \autoref{alg:basicAlgorithm} corresponds to a minimal hierarchical dissociation, i.e.\ it is not ``dominated'' by (that means has always bigger or equal probably than) any other plan. 
We show this by induction on the set of relations in a plan. At each step, any query must either project or join. 
Joins are only possible if the join variables are identical to the head variables of the plan, hence if the subplans share no existential variables. If at a step, there are disconnected components, then all minimal plans need to join on those components (as any project on a set of variables is clearly dominated by a join), then project on those variables only in the component that contains this set. 
If the relations are connected, then each plan needs to project on variables so that the query becomes disconnected. 
We need to show that projecting on a minimal set of variables that disconnects the query cannot be dominated by any other hierarchical dissociation. 
This follows immediately, as any projection on a subset does not disconnect the query, and every projection on a superset is dominated, and any overlapping set of variables cannot dominate this dissociation.

(2) \emph{Completeness}: We show that \autoref{alg:basicAlgorithm} returns a plan for each minimal hierarchical dissociation. We again show this by induction on the set of relations in a plan. 
Recall from before that if a query is disconnected, then the minimal query plan needs to join all components and we only need to focus on the case when the query is not disconnected. 
We need to consider all possible subsets of the variables that disconnect the query. We have shown that only those can be minimal that are not supersets of variables that alone disconnect. 
Since the algorithm iterates over all minimal subsets that disconnect the query, it is complete.
\qed
\end{proof}

\section{\autoref{sec:4:other}: 
Proof
Small probabilities}
In the following, we write $\mathcal{P}_{> k}([n])$ for the set of subsets of the powerset $\mathcal{P}([n])$ of cardinality bigger or equal to $k$. 

\begin{lemma}[DNF polynomial]\label{lem:multlinearDNF}
Let $\phi(\vec x)$ be a positive minimal $k$-uniform DNF of size $m$ containing $n$ total variables.
Then the multilinear polynomial for $\phi(\vec x)$ is
\begin{align*}
	\PP{\phi(\vec x)} &= \sum_{j=i}^m \prod_{i \in \vec c_j} p_i + 
		\sum_{S \in \mathcal{P}_{> k}([n])} c(S) \prod_{i \in S} p_i
\end{align*}
\end{lemma}

In other words, all of the terms of the multilinear polynomial for $\phi$ are of order bigger or equal to $k$. Furthermore, most are of order $>k$  except for $\sum_{j=i}^m \prod_{i \in \vec c_j} p_i$, which corresponds to summing up the probabilities of each conjunct.

\begin{proof}[\autoref{lem:multlinearDNF}]
Applying the inclusion-exclusion principle to the event probability for a \emph{positive DNF} leads to:
\begin{align}
	&\PP{\bigvee_{j=1}^m c_j} = \sum_{k=1}^m (-1)^{k-1} \!\!\!\!\!  \sum_{1\leq j_1 < \ldots < j_k \leq m} 
		\PP{\bigwedge_{i=1}^{k} c_{j_i}}
\label{inclusionExclusion2}
\end{align}
\Autoref{lem:multlinearDNF} then follows immediately by noting that any conjunction of terms $c_{j_1} \wedge c_{j_2}$ needs to consist of at least $k+1$ factors. This follows from the fact that any two terms $c_{j_1} \neq c_{j_2}$ need to have at least one different variable. 
For example, $x_1x_2x_3 \wedge x_1x_2x_4 = x_1 x_2 x_3 x_4$.
\qed
\end{proof}

\begin{proof}[\autoref{prop:smallProbabilities}: Small probabilities]
The lineage for a self-join-free conjunctive query $q$ is a positive $k$-uniform DNF of size $m$. Also, the lineage of any dissociation $q'$ of this query (in particular the one with the minimal score) is a positive $k$-uniform DNF of size $m$. It follows that 
\begin{align*}
	\PP{q} &= \sum_{j=i}^m \prod_{i \in \vec c_j} p_i + 
		\sum_{S \in \mathcal{P}_{> k}([n])} c(S) \prod_{i \in S} p_i	\\
	\PP{q'} &= \sum_{j=i}^m \prod_{i \in \vec c_j} p_i + 
		\sum_{S \in \mathcal{P}_{> k}([n])} c'(S) \prod_{i \in S} p_i		
\end{align*}
Next consider the operation of scaling down all probabilities by $f$, $p_i' = f p_{i}$, and the implication on
$\epsilon := \frac{\rho(q) - \PP{q}}{\PP{q}}$.
Observe that in the nominator, the terms of order $k$ cancel out and it becomes a sum of terms with minimum order $k+1$. In contrast, the denominator still has minimum order of $k$. After dividing both nominator and denominator by $f^k$, we have the each term in the denominator is multiplied with at least $f$, whereas the denominator has some terms without $f$. We thus have $\lim_{f \rightarrow 0^+} \epsilon = 0$.
\qed
\end{proof}

Finally, notice that the proof here was closely inspired by the proof for \cite[Theorem 19]{DBLP:journals/pvldb/GatterbauerGKF15} and v.v.

\section{\autoref{sec:4:other} (3): 
Number of minimal query plans
}\label{app:numberOfQueryPlans}

\introparagraph{Star queries}
Consider the $k$-star query $q \datarule R_1(x_1), \ldots, R_k(x_k), U(\vec x)$ and the $k!$ permutations on the order on $\vec x$. To simplify notation and w.l.o.g., we consider the permutation $\sigma = (x_{1}, \ldots, x_{k})$. We show the following query plan is minimal:
\begin{align*}
	P & = 
		\projp{x_{1}} \joinp{}{R_{1}(x_{1}),
        \projp{x_{2}} \joinp{}{\ldots, 
		\projp{x_{k}} \joinp{}{R_{k}(x_{k}), U(\vec x)
		}}}
\end{align*}
$P$ has $k$ projections and corresponds to a dissociation $\Delta$, where relation $R_{i}$ is dissociated on
$i-1$ variables $\vec x_{[1,i-1]}$.
Each query plan according to \autoref{def:queryPlans} must have either one projection and a single relation, or a join between more relations at the end of its branches. No query plan for $q$ can have a projection at the leaf since each variable appears in at least 2 relations. 
Since there is only one relation $U$ that joins with other relations, there can only be one leaf with a join between $U$ and at least one other relation $R_{i}$ and there cannot be several branches. Further any query plan for $q$ can have at maximum $k$ projections for each of the $k$ variables. As a consequence, each plan with $k$ projections is isomorph to a total order on $\vec x$.
Next consider a query plan that has less then $k$ projections while keeping the order of $\sigma$. This plan must have at least one $R_i$ that is dissociated on at least $i$ variables and is, hence, dominated by $P$. Hence, $P$ is minimal. Since there are $k!$ possible permutations, there are $k!$ such minimal plans.

The total number of plans (or hierarchical dissociations) is equal to the number of weak orderings definable on $k$ alternatives. This is a consequence of the relation $U(\vec x)$ which contains all variables. Hence, each query plan must define a hierarchy between the variables in which all variables are comparable, and in which ties are allowed. This is exactly the definition of a weak ordering (or total preorder), and the number of such is given by the OEIS sequence 
\href{http://oeis.org/classic/A000670}{A000670}\footnote{\url{http://oeis.org/classic/A000670}}: 
1, 3, 13, 75, 541, 4\;\!683, 47\;\!293, 545\;\!835, $\ldots$.

\introparagraph{Chain queries}
Next consider the Boolean $k\!+\!1$-chain query $q \datarule $ $R_1(x_1), R_2(x_1, x_2),\ldots, R_{k\!+\!1}(x_k)$.
Note that each variable is shared only by two subsequent relations. The leaves of each query plan must thus have exactly two relations. Furthermore, every minimal query plan must have exactly $k$ joins between subplan corresponding to combining one variable at a time. Hence, the number of minimal query plans for a $k\!+\!1$-chain query is equal to the number of ways to insert $k$ pairs of parentheses in a word of $k\!+\!1$ letters. This corresponds to OEIS sequence
\href{http://oeis.org/classic/A000108}{A000108}\footnote{\url{http://oeis.org/classic/A000108}}: 1, 2, 5, 14, 42, 132, 429, 1\;\!430, $\ldots$.

The number of total query plans corresponds to OEIS sequence
\href{http://oeis.org/classic/A001003}{A001003}\footnote{\url{http://oeis.org/classic/A001003}}: 1, 3, 11, 45, 197, 903, 4\;\!279, 20\;\!793, $\ldots$.
This is the number of ways to insert parentheses in a string of $k$ symbols. The parentheses must be balanced but there is no restriction on the number of pairs of parentheses. The number of letters inside a pair of parentheses must be at least 2. Parentheses enclosing the whole string are ignored.

The total number of dissociations for both, $k$-star query and $k\!+\!1$-chain query, is given by $2^K$, where $K = k(k-1)$ is the number of undissociated variables. \Autoref{table:numberMinimalQueryPlans} summarizes these numbers in one table.

\section{\autoref{sec:deterministicOptimization}: 
Proof 
Dissociation and DRs}

\begin{proof}[\autoref{lemma:DetDissociation}: Dissociation and DRs]
This lemma follows immediately from \specialref{Theorem}{th:bool:dissoc}{(2)} and noting that dissociating tuples in DRs corresponds exactly to dissociating variables with probabilities equal to 1 in the lineage formula.
\qed
\end{proof}

\section{\autoref{sec:deterministicOptimization}: 
Proof
\autoref{alg:basicAlgorithm} with DRs}

In order to prove \autoref{prop:algorithm2}, we need to add two more lemmas and to extend the definition of separator variables:

\begin{lemma}[Deleting variables from safe queries]\label{lemma:deletingVariables}
If $q$ is safe and $x \in \EVar(q)$, then deleting $x$ from $q$ leads to a query $q'$ that is still safe.
\end{lemma}

\begin{proof}[\autoref{lemma:deletingVariables}]
$q$ is safe iff for any two existential variables $y, z$, one of the following three conditions hold: $\at(y) \subseteq \at(z)$, $\at(y) \cap \at(z) = \emptyset$, or $\at(y) \supseteq \at(z)$. Deleting another variable $x$ from the query does not change the hierarchy between the other variables.
\qed
\end{proof}

\begin{definition}[p-separator]
A ``\emph{p-separator variable}'' is an existential variable that appears in every probabilistic atom of a connected query. We write $\SepPVar(q)$ for the set of all p-separators of $q$.
\end{definition}

\begin{lemma}[Separators and minimal plans]\label{cor:SeparatorVariables}
Let $x$ be a p-separator variable of $q$, then all query plans minimal in ${\preceq^p}$ have $x$ as ``root variable'', i.e.\ a variable that is projected away in the last projection.
\end{lemma}

\begin{proof}[\autoref{cor:SeparatorVariables}]
Let $q$ be connected by its existential variables and $x \in \SepPVar(q)$. If $x$ is contained in all atoms, then $x$ is a root variable in all plans produced by \autoref{alg:basicAlgorithm}.
If $x$ is not contained in all atoms 
(i.e.\ there is a deterministic relation that does not contain $x$),
then there exists a min-cut $\vec y$ considered by the original \autoref{alg:basicAlgorithm} that does not contain $x$. We need to show that, for any hierarchical dissociation $\Delta_{\vec y}$ with $\vec y$ as root variables, there is another hierarchical dissociation $\Delta_x$ with $x$ as root variable and $\PP{q^{\Delta_x}} \leq \PP{q^{\Delta_{\vec y}}}$. 

For that, consider the dissociation $\Delta_{\vec y x}$ obtained by further dissociating $\Delta_{\vec y}$ on $x$. 
According to \autoref{lemma:DetDissociation}, $\PP{q^{\Delta_{\vec y x}}} = \PP{q^{\Delta_{\vec y}}}$, and 
according to \autoref{lemma:deletingVariables}, $\Delta_{\vec y x}$ must be safe. As a consequence, $\Delta_{\vec y x}$ is a hierarchical dissociation that has $\vec y \cup x$ as root variables. Therefore, there exists a minimal hierarchical dissociation $\Delta_x$ (which may or may not be the same as $\Delta_{\vec y x}$) with $x$ as root variable and $\PP{q^{\Delta_x}} \leq \PP{q^{\Delta_{\vec y x}}} = \PP{q^{\Delta_{\vec y}}}$. 
As a consequence, query plans with $x$ as root variable will always have an equal or lower score than plans with $\vec y$ as root variables. Hence, all plans minimal in ${\preceq^p}$ have $x$ as root variable.
\qed
\end{proof}

\begin{proof}[\autoref{prop:algorithm2}: \autoref{alg:basicAlgorithm} with DRs]
Our proof shows that the modified algorithm correctly enumerates the minimal query plans in the presence of deterministic relations, i.e.\ it returns one plan for each minimum safe equivalence class of ${\preceq^p}$.
	
(1) \emph{Soundness}: We show that every plan produced by the modified algorithm corresponds to a minimal safe equivalence class, i.e.\ it produces no two plans 
where one plan dominates the other.
Note that the algorithm can only create two different plans whenever we need to iterate over at least 2 different minimal p-cuts in \autoref{alg1:line10} after modification of the algorithm. This happens whenever we don't have a p-separator set that alone disconnects the query ($\{\SepPVar(q)\} = \MinPCuts(q) $).
Thus, for each two different min-cuts, the algorithm needs to dissociate at least one non-deterministic relation on a variable that the other one does not dissociate on. 
Hence no query plan produced can dominate another one.

(2) \emph{Completeness}: We show that the modified algorithm returns a plan for each minimal hierarchical dissociation in ${\preceq^p}$, i.e.\ for every plan $P'$ from the original \autoref{alg:basicAlgorithm}, there is a plan $P$ from \autoref{alg:basicAlgorithm} after modification with lower or equal score. 
If $q$ is disconnected, then nothing changes from \autoref{alg:basicAlgorithm}, and we only need to focus on the case when the query is not disconnected. If $q$ is connected, then completeness follows immediately from \autoref{cor:SeparatorVariables}: all minimal plans need to have all p-separator variables as root variables. 
\qed
\end{proof}

\section{\autoref{sec:FDOptimization}: 
Proof
Dissociation and FDs}

\begin{proof}[\autoref{lemma:FDs1}: Dissociation and FDs]
Let $\bm\upGamma$ be the set of FDs on $\Var(q)$ consisting of the union of FDs on every atom in $q$.
First notice that $\bm\upGamma$ also holds for the natural join between all tables in the query.
Next consider a relation relation $R(\vec x_i)$ and let $\vec{x}_i^+$ be the closure of $\vec x_i$ under $\bm\upGamma$.
Then any join result with the same values $\vec a$ for $\vec x_i$ has exactly one tuple from $R_i$ in its lineage, which does not change after dissociating $R_i$ on $\vec x_i^+ \!- \vec x_i$. Therefore, also the lineage and the query reliability remain the same.
\qed
\end{proof}

\section{\autoref{sec:FDOptimization}: 
Proof
\autoref{alg:basicAlgorithm}
with DRs and FDs}

Before we prove the correctness of \autoref{prop:allQueryPlansFDs}, we introduce a helpful \autoref{lemma:FDs2} for FDs, which has no analogy for DRs. 
With DRs, eagerly dissociating all deterministic tables in a hierarchical dissociation $q^{\Delta}$ can lead to a non-hierarchical dissociation (see  \autoref{appendix:RemarksDRs} for an example).
With a functional dependency $\Gamma: \vec x \rightarrow y$, however, we can first eagerly dissociate the dependent variable $y$ in \emph{all relations} that include $\vec x$ among their variables, then apply our previous algorithm 
In other words, we can replace a query $q$ with its closure $q^+$.
In particular, if $q$ is hierarchical, then the closed query $q^+$ is still hierarchical.
Note the emphasis on \emph{all relations} as applying a FD only selectively 
(i.e.\ only to a subset of all relations) 
can turn a hierarchical query into a non-hierarchical one
(see  \autoref{appendix:RemarksFDs} for an  example).

\begin{lemma}(FD dissociation and hierarchies)\label{lemma:FDs2}
Given a hierarchical query $q$ and a FD $\Gamma: \vec x \rightarrow y$.
Dissociating all relations $R_i$ with $\vec x \subseteq \Var(R_i)$ on the dependent variable $y$ leads to a query that is still hierarchical.
\end{lemma}

\begin{proof}[\autoref{lemma:FDs2}: FD dissociation and hierarchies]
Assume that the FD $\Gamma: \vec x \rightarrow y$ holds and that query $q$ is hierarchical. 
We show that applying $\Gamma$ eagerly to \emph{all} atoms in $q$ results in a query $q^{\Delta}$ that is still hierarchical. 
Let $\at(\vec x)$ denote the atoms containing all variables of $\vec x$,  
$\at(y)$ the atoms containing $y$ in $q$, 
and $\at^\Delta(y)$ the atoms containing $y$ in $q^\Delta$.
First notice that $\at(\vec x) \cap \at(y) \neq \emptyset$ since we have one relation that enforces $\Gamma$.
If 	$\at(\vec x) \supseteq \at(y)$, then $\at(\vec x) = \at^\Delta(y)$ and $q^\Delta$ is still hierarchical, as removing $y$ from the original hierarchy cannot invalidate the hierarchy.
If, however, $\at(\vec x) \subset \at(y)$, then $\at(y) = \at^\Delta(y)$ and thus the hierarchy remains unchanged.
\qed
\end{proof}

Our third modification applies \autoref{lemma:FDs2} to dissociate eagerly on all functional dependencies (i.e.\ by replacing a query $q$ by its closed query $q^+$) at the beginning of each recursive call of our previously modified \autoref{alg:basicAlgorithm}. We next prove that the resulting algorithm returns a sound and complete enumeration of query plans minimal in 
${\preceq^{p}}$ 
in the presence of FDs and DRs, i.e.\ it returns exactly one from each minimal safe equivalence class in ${\preceq^{p}}$.

\begin{proof}[\autoref{prop:allQueryPlansFDs}: \autoref{alg:basicAlgorithm} with DRs and FDs]
In the following, let $q^+$ stand for a query $q$ that is closed after applying FDs $\bm{\upGamma}$,
``ALG'' for our previously modified \autoref{alg:basicAlgorithm}, and
``ALG$^+$'' for ALG after
 our additional modification that replaces a query $q$ with its closure $q^+$ at the start of each recursive call.
We start from the observation that ALG will return a superset of the plans we need (we need one plan from each minimal safe dissociation). In other words, ALG may return two plans $P$, $P'$ such that $P$ dominates $P'$ under functional dependencies.  ALG$^+$ returns a subset of the plans returned by ALG.  
We will prove two directions:

{(1) \emph{Completeness}}: 
We prove that every plan $P'$ returned by ALG is dominated (under functional dependencies) by some plan $P$ returned by ALG1$^+$ 
(recall that ``dominating'' means having smaller or equal scores over every database).
Suppose the contrary and let $P'$ be a plan that strongly dominates all plans produced by 
ALG$^+$
(i.e.\ $P'$ has a smaller score over some database).
Then let $q'$ be the dissociated hierarchical query corresponding to $P'$ 
and let $q''$ be the query after closing $q'$ under FDs $\bm\upGamma$.
Then we know from \autoref{lemma:FDs1} that $\PP{q''}= \PP{q'}$. 
Furthermore, since $q'$ is hierarchical, we know from \autoref{lemma:FDs2} that $q''$ must be hierarchical as well. 
Furthermore, we know that $q'' \succeq q^+$ in the partial dissociation order of $q^+$.
As a consequence, 
ALG$^+$
will enumerate a query plan $P''$ with $\score(P'') = \PP{q''}= \PP{q'}$ in one of the produced query plans, which violates our assumption.

{(2) \emph{Soundness}}: 
We prove that no two plans $P_1$, $P_2$ returned by ALG$^+$ can dominate each other under functional dependencies.  We already know that none dominates the other in the absence of FDs, by the soundness of ALG. 
Our claim now follows from the correctness of the previously modified algorithm and an additional induction step:
($i$) If $q^+$ is hierarchical, then there is just one query plan. 
($ii$) If $q^+$ is not hierarchical, then there are at least two different minimal p-cuts, either of which dissociates at least one probabilistic relation on a variable that the other one does not. 
Hence no query plan resulting from one minimal p-cut can dominate a query plan resulting from another minimal p-cut. 
Since repeated calls to our recursive algorithm always end up at a hierarchical query, 
soundness now follows by induction from the leaves of the query plans.
\qed
\end{proof}

\section{\autoref{sec:deterministicOptimization}: 
Remarks on deterministic relations}\label{appendix:RemarksDRs}
At a quick glance, \autoref{lemma:DetDissociation} seems to suggest that we can first dissociate all deterministic relations 
and then apply our standard algorithm to find the set of minimal plans of a query. 
However, this is not correct as we illustrate next with a counter-example:

\begin{figure}[t]
\centering
\renewcommand{\tabcolsep}{0.65mm}
\renewcommand{\arraystretch}{0.9}
\subfloat[$q$]{
	\begin{tabular}[t]{@{\hspace{1pt}} >{$}c<{$}|>{$}c<{$} >{$}c<{$} >{$}c<{$} >{$}c<{$} @{\hspace{1pt}}}
			& x	& y & z \\
		\hline
		R	& \circ	&  \\		
		S	& \circ	& \circ \\
		T^d	& & \circ	& \circ \\
		U	& & & \circ
	\end{tabular}\label{Fig_CounterExampleDet_a}}
\hspace{5mm}
\subfloat[$q^{\Delta_1}$]{
	\begin{tabular}[t]{@{\hspace{1pt}} >{$}c<{$}|>{$}c<{$} >{$}c<{$} >{$}c<{$} >{$}c<{$} @{\hspace{1pt}}}
			& x	& y & z \\
		\hline
		R	& \graycell\circ	& \graycell\bullet  \\		
		S	& \graycell\circ	& \graycell\circ \\
		T^d	& & \graycell\circ	& \graycell\circ \\
		U	& & \graycell\bullet & \graycell\circ 
	\end{tabular}\label{Fig_CounterExampleDet_b}}
\hspace{5mm}	
\subfloat[$q^{\Delta_2}$]{
	\begin{tabular}[t]{@{\hspace{1pt}} >{$}c<{$}|>{$}c<{$} >{$}c<{$} >{$}c<{$} >{$}c<{$} @{\hspace{1pt}}}
			& x	& y & z \\
		\hline
		R	& \circ	&  \\		
		S	& \circ	& \circ \\
		T^d	& \dotcirc & \circ	& \circ \\
		U	& & & \circ 
	\end{tabular}\label{Fig_CounterExampleDet_c}}	
\hspace{25mm}
\renewcommand{\tabcolsep}{0.5mm}
\subfloat[$D$]{
	\label{Fig_CounterExampleInstance}	
	\begin{minipage}[t]{46mm}
	\begin{minipage}[t]{45mm}
		\mbox{
			\begin{tabular}[b]{ >{$}c<{$} | >{$}c<{$}	>{$}c<{$}	>{$}c<{$}   }
 			R	& A		\\
			\hline
			r_1	& a	\\
			r_2	& b					
			\end{tabular}}
		\hspace{1mm}
		\mbox{
			\begin{tabular}[b]{ >{$}c<{$} | >{$}c<{$} >{$}c<{$} >{$}c<{$} >{$}c<{$}}
 			S	& A		& B \\
			\hline
			s_1	& a		& c	\\
			s_2	& b		& c		
			\end{tabular}}
		\hspace{1mm}
		\mbox{
			\begin{tabular}[b]{ >{$}c<{$} | >{$}c<{$} >{$}c<{$}}
			T^d	& B		& C	\\
			\hline
			t_1	& c		& e\\	
			t_2	& c		& f
			\end{tabular}}
		\hspace{1mm}
		\mbox{
			\begin{tabular}[b]{ >{$}c<{$} | >{$}c<{$} >{$}c<{$}}
			U	& C	\\
			\hline
			u_1	&  e\\	
			u_2	&  f
			\end{tabular}}		
	\end{minipage}
	\vspace{-1mm}
	\end{minipage}}						
\caption{\autoref{ex:IncorrectDetDiss}. 
(a): Chain query $q \datarule R(x), S(x,y), T^d(y,z), U(z)$.
(b): Minimal hierarchical dissociation $q^{\Delta_1}$ for which $\rho(q^{\Delta_1}) = \PP{q}$ happens to hold on database $D$ shown in (d). 
(c): Dissociation $q^{\Delta_2}$ replaces $T^d(y,z)$ with $T^{d\{x\}}(x, y,z)$, and thus 
$\PP{q^{\Delta_2}} = \PP{q}$. However, $q^{\Delta_2}$ is not yet hierarchical, and any further minimal hierarchical dissociation increases the probability, and thus 
$\rho(q^{\Delta_2}) > \rho(q^{\Delta_1})$ on $D$.}
	\label{Fig_CounterExampleDet}
\end{figure}

\begin{example}[Incorrect deterministic dissociation 1]\label{ex:IncorrectDetDiss}
Consider the query
$q \datarule R(x), S(x,y), T^d(y,z), U(z)$
with deterministic relation $T^d$ over the database $D$ from \autoref{Fig_CounterExampleInstance}. Here, indexed small letters refer to tuples in the respective relations, e.g., $r_1$ for tuple $R(a)$. The lineage of $q$ is
$\Lineage(q) = r_1 s_1 t_1 u_1 \vee r_1 s_1 t_2 u_2 \vee r_2 s_2 t_1 u_1 \vee r_2 s_2 t_2 u_2 $.
Replacing $t_1$ and $t_2$ with $1$, the lineage can be factored into
$\Lineage(q) = (r_1 s_1 \vee r_2 s_2 ) (u_1 \vee u_2)$,
which is a read-once formula. 
Assuming all non-deterministic tuples to have the same probability $0.5$, the query reliability is $\PP{q} = 21/64 \approx 0.328$. 
Query $q$ has three min-cuts  $\MinCuts(q) = \{\{x\},\{y\},\{z\}\}$,
as easily seen from its incidence matrix (\autoref{Fig_CounterExampleDet_a}).
Dissociating $q$ on $y$ (\autoref{Fig_CounterExampleDet_b}) results in a hierarchical dissociation $q^{\Delta_1}$ that turns out to have, on the particular database $D$, exactly the same lineage expression as the original query: $\Lineage(q^{\Delta_1}) = \Lineage(q)$. 
Since it is a hierarchical dissociation, there is a plan $P_1$ that calculates the probability of $q^{\Delta_1}$ and its score is exactly the probability of the original query: 
$\score(P_1) = \PP{q}\approx 0.328$. 
However, if we had first dissociated the deterministic relation $T^d$ into 
$T^{d \{x\}}(x,y,z)$, 
we would have a dissociated query $q^{\Delta_2}$ (\autoref{Fig_CounterExampleDet_c}) that,
despite having the same probability as the original query, $\PP{q^{\Delta_2}} = \PP{q}$, 
\emph{cannot be made hierarchical by only dissociating variable $y$};
any minimum hierarchical dissociations of $q^{\Delta_2}$ requires 
dissociating at least $x$ or $z$.
This turns out to increase the propagation score to 
$\rho(q^{\Delta_2}) = 87/256 \approx 0.340 > 0.328 \approx \rho(q^{\Delta_1}) = \PP{q}$.
\markend
\end{example}

We give another example to illustrate that the sound and complete enumeration in the presence of deterministic relations is non-trivial:
A reasonable idea seems be to first join all deterministic relations, thereby eagerly reducing variables that appear only in deterministic relations. 
However, by doing so, we may miss some minimal plans, as we will illustrate next:

\begin{example}[Incorrect deterministic dissociation 2]\label{ex:alternativeIncorrectDeterministic}
We again consider the chain query 
$q \datarule R(x), S^d(x,y), T^d(y,z), U(z)$ over the database from \autoref{ex:IncorrectDetDiss}, but now
with two deterministic relations $S^d$ and $T^d$. Then lineage of $q$ is again
$\Lineage(q) =  r_1 s_1 t_1 u_1 \vee r_1 s_1 t_2 u_2 \vee r_2 s_2 t_1 u_1 \vee r_2 s_2 t_2 u_2 $.
Replacing $s_1, s_2, t_1$ and $t_2$ with $1$, the lineage can be simplified and factored into
$\Lineage(q) = (r_1 \vee r_2) (u_1 \vee u_2)$,
which is a read-once formula. Assuming all non-deterministic tuples to have the same probability $0.5$, the query reliability can easily be calculated as $\PP{q} = 9/16 \approx 0.563$. 
Dissociating $q$ on $y$ (\autoref{Fig_CounterExampleDet2_b}) results in a hierarchical dissociation $q^{\Delta_1}$ that turns out to have, on the particular database $D$, exactly the same lineage expression as the original query: $\Lineage(q^{\Delta_1}) = \Lineage(q)$. 
Since it is a hierarchical dissociation, there is a unique plan $P_{\Delta}$ that calculates the reliability of $q^{\Delta_1}$, and its score is exactly the reliability of the original query: $\rho(q) = \score(P_{\Delta_1}) = \PP{q}\approx 0.563$. 
However, if we had first joined the two deterministic relation $S^d$ and $T^d$ into $N^d(x,y,z) = \joind{}{S^d(x,y), T^d(y,z)}$ (\autoref{Fig_CounterExampleDet2_c}), and then projected the variable $y$ away into $N^{d\prime\prime} = \projd{y} N^d(x,y,z)$, then the propagation score of $q''$ turns out to be bigger then the one of the original query: $\rho(q'') = 39/64 \approx 0.609 > r(q'') = \rho(q)$.
\markend
\end{example}

\begin{figure}[t]
\centering
\renewcommand{\tabcolsep}{0.65mm}
\renewcommand{\arraystretch}{0.9}
\subfloat[$q$]{
	\begin{tabular}[t]{@{\hspace{1pt}} >{$}c<{$}|>{$}c<{$} >{$}c<{$} >{$}c<{$} >{$}c<{$} @{\hspace{1pt}}}
			& x	& y & z \\
		\hline
		R	& \circ	  \\		
		S^d	& \circ &  \circ \\
		T^d	& & \circ	& \circ \\
		U	& & & \circ 		
	\end{tabular}\label{Fig_CounterExampleDet2_a}}
\hspace{5mm}
\subfloat[$q^{\Delta_1}$]{
	\begin{tabular}[t]{@{\hspace{1pt}} >{$}c<{$}|>{$}c<{$} >{$}c<{$} >{$}c<{$} >{$}c<{$} @{\hspace{1pt}}}
			& x	& y & z \\
		\hline
		R	& \graycell\circ	& \graycell\bullet \\		
		S^d	& \graycell\circ &  \graycell\circ \\
		T^d	& & \graycell\circ	& \graycell\circ \\
		U	& & \graycell\bullet & \graycell\circ\\		
	\end{tabular}\label{Fig_CounterExampleDet2_b}}
\hspace{5mm}	
\subfloat[$q'$]{
	\begin{tabular}[t]{@{\hspace{1pt}} >{$}c<{$}|>{$}c<{$} >{$}c<{$} >{$}c<{$} >{$}c<{$} @{\hspace{1pt}}}
			& x	& y & z \\
		\hline
		R	& \circ	  \\		
		N^d	& \circ & \circ & \circ \\
		U	& & & \circ 		
	\vspace{3.5mm}
	\end{tabular}\label{Fig_CounterExampleDet2_c}}
\hspace{5mm}	
\subfloat[$q''$]{
	\begin{tabular}[t]{@{\hspace{1pt}} >{$}c<{$}|>{$}c<{$} >{$}c<{$} >{$}c<{$} >{$}c<{$} @{\hspace{1pt}}}
			& x	& z  \\
		\hline
		R	& \circ	&  \\		
		N^{d\prime\prime}	& \circ & \circ \\
		U	& & \circ 		
	\vspace{3.5mm}
	\end{tabular}\label{Fig_CounterExampleDet2_d}}
\caption{\autoref{ex:alternativeIncorrectDeterministic}: 
(a): We again use the chain query
and database from \autoref{ex:IncorrectDetDiss}
and \autoref{Fig_CounterExampleDet} but now have also relation $S^d$ deterministic.
(b): Minimal hierarchical dissociation $q^{\Delta_1}$ that calculates $\PP{q}$ exactly on database $D$. (c,d): First joining both deterministic tables into $N^d = \joind{}{S^d,T^d}$ and projecting $y$ away does not change the query reliability. However, there is no minimal hierarchical dissociation anymore that allows to calculate the original probability exactly: $\rho(q'') > \rho(q)$.}
	\label{Fig_CounterExampleDet2}
\end{figure}

\begin{figure}[t]
\centering
\includegraphics[scale=0.5]{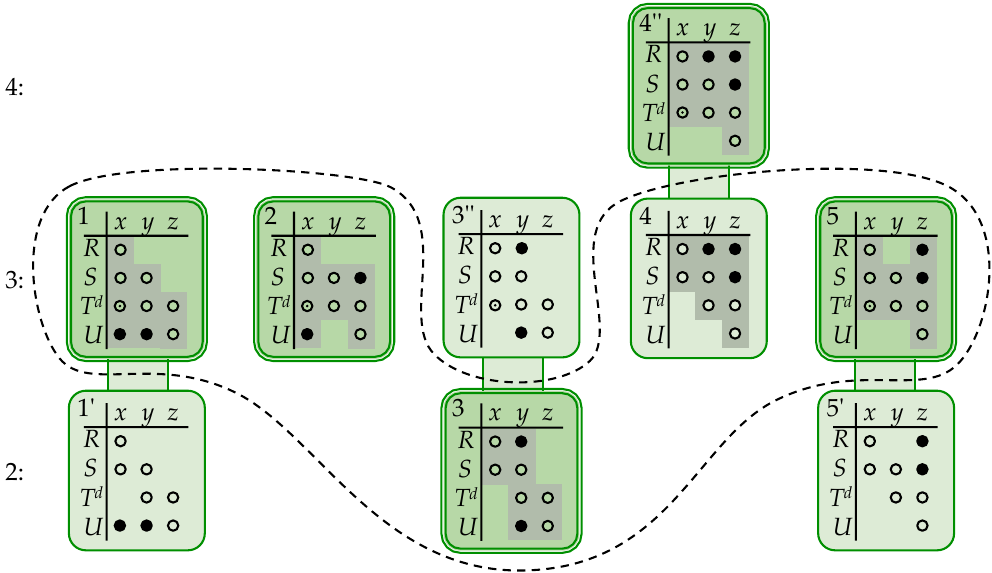}
\caption{\autoref{ex:PartialDissociationOrderExampleExcerpt}: 9 dissociations forming the 5 \emph{minimal safe equivalence classes} in the \emph{probabilistic dissociation preorder} among all 64 dissociations of chain query $q \datarule R(x), S(x,y), T^d(y,z), U(z)$.}
\label{Fig_PartialDissociationOrderExampleExcerpt}
\end{figure}

\begin{example}[Enumerating minimal plans]\label{ex:PartialDissociationOrderExampleExcerpt}
We illustrate here how the \emph{probabilistic dissociation preorder} changes in the presence of DRs for a slightly more complicated example: Consider again the chain query 
$q \datarule R(x), S(x,y), T^d(y,z), U(z)$
from \autoref{ex:alternativeIncorrectDeterministic}.
The query has $2^6= 64$ dissociations among which 5 are minimal and safe in the \emph{partial dissociation order} (thus ignoring DRs). Those 5 dissociations are returned by \autoref{alg:basicAlgorithm}  and shown in \autoref{Fig_PartialDissociationOrderExampleExcerpt} with numbers $1, \ldots, 5$ inside the dashed outline (the vertical height corresponds to the number of dissociated variables: 2, 3, or 4).
With $T^d$ deterministic, the probabilistic dissociation preorder has following equivalences for those 5 dissociations:
$\Delta_1' \equiv^p \Delta_{1}$,
$\Delta_3 \equiv^p \Delta_{3''}$,
$\Delta_4 \equiv^p \Delta_{4''}$,
and $\Delta_5' \equiv^p \Delta_{5}$.
The modified enumeration algorithm will thus return $4''$ instead of $4$ (it has fewer constraints on the join order).
Also notice from \autoref{Fig_PartialDissociationOrderExampleExcerpt} that $\Delta_{3''}$ is safe (because of $\Delta_3 \equiv^p \Delta_{3'}$) despite not being hierarchical.
\markend
\end{example}

\section{\autoref{sec:FDOptimization}:
Remark on completeness for functional dependencies}\label{appendix:RemarksFDs}
With a FD $\Gamma: \vec x \rightarrow y$, we can first eagerly dissociate the dependent variable $y$ in all relations that include $\vec x$ as variables, then apply our modified algorithm for enumerating all minimal plans. Note the emphasis on all (!) relations, as this property does not hold for arbitrary subsets as we illustrate with a small example:

\begin{figure}[t]
\centering
\renewcommand{\tabcolsep}{0.65mm}
\renewcommand{\arraystretch}{0.9}
\subfloat[$q$]{
	\begin{tabular}[t]{@{\hspace{1pt}} >{$}c<{$}|>{$}c<{$} >{$}c<{$} >{$}c<{$} >{$}c<{$} @{\hspace{1pt}}}
			& x	& y & z \\
		\hline
		R	& \graycell\circ	& \graycell\circ & \graycell\circ \\		
		S	& \graycell\circ	 \\
		U	& \graycell\circ & & \graycell\circ 
	\end{tabular}\label{Fig_CounterExampleFD_a}}
\hspace{3mm}
\subfloat[$q'$]{
	\begin{tabular}[t]{@{\hspace{1pt}} >{$}c<{$}|>{$}c<{$} >{$}c<{$} >{$}c<{$} >{$}c<{$} @{\hspace{1pt}}}
			& \multicolumn{2}{c}{$\, x \!\!\! \shortrightarrow \!\!\! y$\hspace{0.5mm}} & z \\
		\hline
		R	& \circ	& \circ & \circ \\		
		S	& \circ	& \dotcirc \\
		U	& \circ & & \circ
	\end{tabular}\label{Fig_CounterExampleFD_b}}
\hspace{3mm}	
\subfloat[$q''$]{
	\begin{tabular}[t]{@{\hspace{1pt}} >{$}c<{$}|>{$}c<{$} >{$}c<{$} >{$}c<{$} >{$}c<{$} @{\hspace{1pt}}}
			& \multicolumn{2}{c}{$\, x \!\!\! \shortrightarrow \!\!\! y$\hspace{0.5mm}} & z \\
		\hline
		R	& \graycell\circ & \graycell\circ & \graycell\circ \\		
		S	& \graycell\circ & \graycell\dotcirc	 \\
		U	& \graycell\circ & \graycell\dotcirc & \graycell\circ
	\end{tabular}\label{Fig_CounterExampleFD_c}}					
\caption{\autoref{ex:incorrectFDs}. Applying an FD $x \rightarrow y$ only to a subset of relations can turn the previously safe query $q$ into an unsafe $q'$. However, $q$ cannot become unsafe when applying $\Gamma$ to all relation (here $q''$).}
\label{Fig_CounterExampleFD}
\end{figure}

\begin{example}[Incorrect functional dependency dissociation]\label{ex:incorrectFDs}
Consider the query $q \datarule R(x,y,z), S(x), T(x,z)$ with FD $\Gamma: x \rightarrow y$ holding in relation $R$. The query is safe (the hierarchy is shown in gray in \autoref{Fig_CounterExampleFD_a})
and has  one single plan
$P = \projp{x} \joinp{}{S(x), \projp{z} \joinp{}{\projp{y}R(x,y,z),U(x,z)}}$.
Dissociating relation $S$ on $y$ does not change the reliability, however it makes the query unsafe ($q'$ in \autoref{Fig_CounterExampleFD_b}) with now two minimal plans. If we instead dissociate $y$ in both $S$ and $U$, the resulting query $q''$ is safe as well (\autoref{Fig_CounterExampleFD_c}) with one single minimal plan
$P'' = \projp{x,y} \joinp{}{S(x), \projp{z} \joinp{}{R(x,y,z),U(x,z)}}$.
Note that $P''$ has one projection less than $P'$.
\markend
\end{example}

\newsavebox\queryplanzero
\begin{lrbox}{\queryplanzero}
	\small
	\begin{minipage}[t]{100mm}
		\fontsize{8}{10}
		\vspace{-4mm}			
		\begin{align*}
			\rho(q) & = \minp{
			\begin{aligned}
				&\projp{z} \joinp{z}{T(z),
			            \projp{x} \joinp{x,z}{R(x,z),
			            \projp{u} \joinp{u}{U(u),
			            \projp{y} \joinp{y,u}{S(y,u), M(x,y,z,u)}}}}  \\
			            &\projp{z}\joinp{z}{T(z),
			            \projp{u} \joinp{u}{U(u),
			            \projp{x} \joinp{x,z}{R(x,z),
			            \projp{y} \joinp{y,u}{S(y,u), M(x,y,z,u)}}}}  \\
			            &\projp{z}\joinp{z}{T(z),
			            \projp{u} \joinp{u}{U(u),
			            \projp{y} \joinp{y,u}{S(y,u),
			            \projp{x} \joinp{x,z}{R(x,z), M(x,y,z,u)}}}}  \\
			            &\underline{
						\projp{u}\joinp{u}{U(u),
			            \projp{y} \joinp{y,u}{S(y,u),
			            \projp{z} \joinp{z}{T(z),
			            \projp{x} \joinp{x,z}{R(x,z), M(x,y,z,u)}}}}
						} \\
			            &\projp{u}\joinp{u}{U(u),
			            \projp{z} \joinp{z}{T(z),
			            \projp{x} \joinp{x,z}{R(x,z),
			            \projp{y} \joinp{y,u}{S(y,u), M(x,y,z,u)}}}}  \\
			            &\projp{u}\joinp{u}{U(u),
			            \projp{z} \joinp{z}{T(z),
			            \projp{y} \joinp{y,u}{S(y,u),
			            \projp{x} \joinp{x,z}{R(x,z), M(x,y,z,u)}}}}  \\
			\end{aligned}
			}
		\end{align*}
	\end{minipage}
\end{lrbox}	

\newsavebox\queryplanfirst
\begin{lrbox}{\queryplanfirst}
	\small
	\begin{minipage}[t]{80mm}
		\vspace{-2mm}
		\setlength\abovedisplayskip{0pt}	
		\begin{align*}	
			\rho(q_f) & = \minp{\!
			\begin{aligned}			
				&\,\, \projp{x,z} \joinp{x,z}{
				R(x,z),T(z),\projp{y,u} \joinp{y,u}{
				S(y,u), U(u),M(x,y,z,u)}}  \\
	            &\,\,
				\underline{
				\projp{y,u}\joinp{y,u}{
	            S(y,u),U(u),\projp{x,z} \joinp{x,z}{T(z),
	            R(x,z), M(x,y,z,u)}} 
				} \,\,\, \\
			\end{aligned}\!}\!	
		\end{align*}	
	\end{minipage}
\end{lrbox}

\newsavebox\queryplansecond
\begin{lrbox}{\queryplansecond}
	\small
	\begin{minipage}[t]{60mm}
		\setlength\abovedisplayskip{0pt}	
		\begin{align*}
			\rho(q_{fd}) = & \,  
				\underline{
				\projp{y,u} \joinp{y,u}{S(y,u), U(u), \projp{x,z} 
				\joinp{x,z}{T^d(z), R^d(x,z), M(x,y,z,u)}}
				}
		\end{align*}
	\end{minipage}
\end{lrbox}

\newsavebox\queryplanthird
\begin{lrbox}{\queryplanthird}
	\small
	\begin{minipage}[t]{60mm}
		\setlength\abovedisplayskip{0pt}	
		\vspace{4mm}			
		\begin{align*}
			\rho(q_d)	& = \,
				\underline{
				\projp{u} \joinp{u}{U(u),
				\projp{y} \joinp{y,u}{S(y,u), \projp{x,z} \joinp{x,z}{T^d(z), R^d(x,z), M(x,y,z,u)}}}
				}
		\end{align*}
		\vspace{-1mm}			
	\end{minipage}
\end{lrbox}

\begin{figure*}[t]
\centering
\renewcommand{\tabcolsep}{0.65mm}
\renewcommand{\arraystretch}{0.9}
\vspace{-1mm}
\begin{tabular}{	@{\hspace{3mm}}  c
					@{\hspace{3mm}}  c
					@{\hspace{3mm}}  c
					@{\hspace{5mm}}  c}
\subfloat[$q$]{\label{Fig_RuningExamplePlans_a}
	\begin{minipage}[t]{14mm}
		\vspace{4mm}
		\begin{tabular}[t]{@{\hspace{1pt}} >{$}c<{$}|>{$}c<{$} >{$}c<{$} >{$}c<{$} >{$}c<{$} @{\hspace{1pt}}}
				& x	& z & y & u \\
			\hline
			M	& \circ	& \circ & \circ & \circ \\								
			R	& \circ	& \circ \\
			T	& & \circ \\
			S	& & & \circ	& \circ \\
			U	& & & & \circ \\
		\end{tabular}
		\vspace{4mm}
	\end{minipage}		}
&
\subfloat[$q^{\Delta_4}$]{\label{Fig_RuningExamplePlans_a2}
	\begin{minipage}[t]{17mm}
		\vspace{4mm}
		\begin{tabular}[t]{@{\hspace{1pt}} >{$}c<{$}|>{$}c<{$} >{$}c<{$} >{$}c<{$} >{$}c<{$} @{\hspace{1pt}}}
				& x	& z & y & u \\
			\hline
			M	& \graycell\circ	& \graycell\circ & \graycell\circ & \graycell\circ \\ 
			R	& \graycell\circ	 & \graycell\circ & \graycell\bullet & \graycell\bullet \\
			T	& 					 & \graycell\circ & \graycell\bullet & \graycell\bullet \\
			S	& & & \graycell\circ & \graycell\circ \\
			U	& & & & \graycell\circ 
		\end{tabular}
		\vspace{4mm}		
	\end{minipage}		}
&&
\subfloat[]{\label{Fig_RuningExamplePlans_b}
	\usebox{\queryplanzero} }		
\\
\subfloat[$q_d$]{
	\begin{tabular}[t]{@{\hspace{1pt}} >{$}c<{$}|>{$}c<{$} >{$}c<{$} >{$}c<{$} >{$}c<{$} @{\hspace{1pt}}}
			& x	& z & y & u \\
		\hline
		M	& \circ	& \circ & \circ & \circ \\								
		R^d	& \circ	& \circ &  &  \\
		T^d	& 	 & \circ &  &  	\\
		S	& & & \circ & \circ \\
		U	& & & & \circ 
	\end{tabular}\label{Fig_ExampleIncidenceMatrixDet_b}}
&
\subfloat[$q_d^{\Delta}$]{
	\begin{tabular}[t]{@{\hspace{1pt}} >{$}c<{$}|>{$}c<{$} >{$}c<{$} >{$}c<{$} >{$}c<{$} @{\hspace{1pt}}}
			& x	& z & y & u \\
		\hline
		M	& \graycell\circ	& \graycell\circ & \graycell\circ & \graycell\circ \\
		R^d	& \graycell\circ	 & \graycell\circ & \graycell\dotcirc & \graycell\dotcirc \\
		T^d	& \graycell\dotcirc	 & \graycell\circ & \graycell\dotcirc & \graycell\dotcirc \\
		S	& & & \graycell\circ & \graycell\circ \\
		U	& & & & \graycell\circ 
	\end{tabular}\label{Fig_ExampleIncidenceMatrixDet_c}}
&
&
\subfloat[]{\label{Fig_ExampleDeterministicRelations}
	\usebox{\queryplanthird} } 
\\
\multirow{2}{*}{
\subfloat[$q_f$]{
	\begin{tabular}[t]{@{\hspace{1pt}} >{$}c<{$}|>{$}c<{$} >{$}c<{$} >{$}c<{$} >{$}c<{$} @{\hspace{1pt}}}
		& \multicolumn{2}{c}{$\,x\!\!\!\shortleftarrow\!\!\!z$\hspace{0.5mm}} 
		& \multicolumn{2}{c}{$\,y\!\!\shortleftarrow\!\!\!u$\hspace{0.5mm}} \\
		\hline
		M	& \circ	& \circ & \circ & \circ \\								
		R	& \circ	& \circ \\
		T	& \dotcirc	& \circ \\
		S	& & & \circ	& \circ \\
		U	& & & \dotcirc	& \circ
	\end{tabular}\label{Fig_ExampleIncidenceMatrix_b} } }
&
\multirow{2}{*}{
\subfloat[$q_f^{\Delta_2}$]{
	\begin{tabular}[t]{@{\hspace{1pt}} >{$}c<{$}|>{$}c<{$} >{$}c<{$} >{$}c<{$} >{$}c<{$} @{\hspace{1pt}}}
		& \multicolumn{2}{c}{$\,x\!\!\!\shortleftarrow\!\!\!z$\hspace{0.5mm}} 
		& \multicolumn{2}{c}{$\,y\!\!\shortleftarrow\!\!\!u$\hspace{0.5mm}} \\
		\hline
		M	& \graycell\circ		& \graycell\circ & \graycell\circ   & \graycell\circ   \\
		R	& \graycell\circ		& \graycell\circ & \graycell\dotcirc & \graycell\bullet \\
		T	& \graycell\dotcirc		& \graycell\circ & \graycell\dotcirc & \graycell\bullet\\
		S	& & & \graycell\circ	& \graycell\circ \\
		U	& & & \graycell\dotcirc	& \graycell\circ
	\end{tabular}\label{Fig_ExampleIncidenceMatrix_d} } }	
&
\multirow{2}{*}{
\subfloat[$q_{fd}$]{\label{Matrix:qfd}
	\begin{tabular}[t]{@{\hspace{1pt}} >{$}c<{$}|>{$}c<{$} >{$}c<{$} >{$}c<{$} >{$}c<{$} @{\hspace{1pt}}}
		& \multicolumn{2}{c}{$\,x\!\!\!\shortleftarrow\!\!\!z$\hspace{0.5mm}} 
		& \multicolumn{2}{c}{$\,y\!\!\shortleftarrow\!\!\!u$\hspace{0.5mm}} \\
		\hline
		M	& \graycell\circ	& \graycell\circ & \graycell\circ & \graycell\circ \\ 
		R^d	& \graycell\circ & \graycell\circ & \graycell\dotcirc & \graycell\dotcirc \\
		T^d	& \graycell\dotcirc & \graycell\circ & \graycell\dotcirc & \graycell\dotcirc \\
		S	& & & \graycell\circ & \graycell\circ \\
		U	& & & \graycell\dotcirc & \graycell\circ
	\end{tabular} } }
&
\subfloat[]{\label{Fig_ExampleKeys}
	\usebox{\queryplanfirst} }	
\\
&&&
\subfloat[]{\label{Fig_ExampleDeterministicRelationsAndKeys}
	\usebox{\queryplansecond} }		
\end{tabular}
\caption{
\Autoref{ex:runningExampleRevisited}: (a): Incidence matrix for our example $q \datarule R(x,z), S(y,u), T(z), U(u), M(x,y,z,u)$ from \autoref{ex:optimizationExample}. 
(b): one minimal hierarchical dissociation $q^{\Delta_4}$. (c): All 6 minimal query plans generated by \specificref{algorithm}{alg:basicAlgorithm} with the plan corresponding to $q^{\Delta_4}$ underlined. The propagation score is the minimum of the scores of those plans. 
(d): Knowing that relations $R^d$ and $T^d$ are deterministic 
makes the query $q_d$ safe and our algorithm returns one minimal plan (f) corresponding to $q_d^{\Delta}$ (e) and $\rho(q_d) = \PP{q_d} = \PP{q_d^{\Delta}}$.
(g): Knowing two FDs $z \rightarrow x$ and $u \rightarrow y$ 
results in two minimal hierarchical dissociations, one of which is shown in (h). The two corresponding minimal query plans are shown in (j).
If we know that $R^d$ and $T^d$ are deterministic and the previous FDs hold, then the query becomes safe (i) and has one single minimal query plan (k).}\label{Fig_ExampleIncidenceMatrix}
\end{figure*}

\section{\autoref{sec:optimizationsWithSchema}: 
A detailed example for schema knowledge}

We will re-use our \autoref{ex:optimizationExample} from \autoref{sec:optimizations} to illustrate changes in the minimal query plans in the presence of deterministic relations and functional dependencies.

\begin{example}[Optimization running example cont.]\label{ex:runningExampleRevisited}
We consider again the query from \autoref{ex:optimizationExample}:
$q \datarule R(x,z), S(y,u), T(z), U(u), M(x,y,z,u)$.
Its incidence matrix (\autoref{Fig_RuningExamplePlans_a}) has 10 variables in relations that can be dissociated; the query thus has $2^{10} = 1\,024$ possible dissociations, among which 6 are minimal hierarchical dissociations. One of those hierarchical dissociations is is shown in \autoref{Fig_RuningExamplePlans_a2}.
\specificref{Algorithm}{alg:basicAlgorithm} returns 6 minimal query plans (\autoref{Fig_RuningExamplePlans_b}) corresponding to these 6 minimal hierarchical dissociations. 

Next assume relations $R$ and $T$ to be deterministic (\autoref{Fig_ExampleIncidenceMatrixDet_b}):
Our modified algorithm then returns one single minimal plan (\autoref{Fig_ExampleDeterministicRelations}), which implies that the query is safe and our plan returns the exact probabilities.
The plan corresponds to the hierarchical dissociation from \autoref{Fig_ExampleIncidenceMatrixDet_c} which also shows that no probabilistic table needed to be dissociated, and thus $\PP{q_d} = \PP{q_d^{\Delta}}$.

Next assume keys $R(x, \underline{z})$ and $S(y, \underline{u})$, and hence the FDs $z \rightarrow x$ and $u \rightarrow y$ hold. Dissociating all dependent variables leads to a new query $q_f$ 
(\autoref{Fig_ExampleIncidenceMatrix_b})
that has only two minimal hierarchical dissociations 
(one of which is shown in \autoref{Fig_ExampleIncidenceMatrix_d}).
\autoref{Fig_ExampleKeys} shows the two minimal plans.

Next assume that relations $R$ and $T$ are deterministic and above FDs hold at the same time.  Then the query becomes $q_{fd}$ which is safe (the hierarchy is shown in \autoref{Matrix:qfd}) and has thus only one single plan (\autoref{Fig_ExampleDeterministicRelationsAndKeys}).

We encourage the reader to take a moment and compare the incidence matrices 
\autoref{Fig_RuningExamplePlans_a2},
\autoref{Fig_ExampleIncidenceMatrixDet_c}, 
\autoref{Fig_ExampleIncidenceMatrix_d}, and
\autoref{Matrix:qfd} and their corresponding underlined plans in
\autoref{Fig_RuningExamplePlans_b}, 
\autoref{Fig_ExampleDeterministicRelations}, 
\autoref{Fig_ExampleKeys}, and 
\autoref{Fig_ExampleDeterministicRelationsAndKeys}.
\markend
\end{example}

\section{\autoref{sec:dissociation}: Alternative proof for
 \autoref{th:upperBounds}}\label{app:proofs} 

We give here an alternative, self-contained proof of \autoref{th:upperBounds}. This is our original proof and has now been superseded with our more general results from \cite{DBLP:journals/tods/GatterbauerS14}. We list this original proof here in the online appendix as it is still interesting in its own right, self-contained and possibly more accessible. 
We start by developing some notions and lemmas for Boolean functions and their probabilistic interpretation.

\smallsection{Boolean notions~\textnormal{\cite{CramaHammer2010:BooleanFunctions}}}
Let $\vec x = \{x_1, \ldots, x_n\}$ be a set of $n$ Boolean variables. 
We use the bar sign (e.g.\ $\vec x$) to denote an ordered or unordered set, depending on the context.
A \emph{truth assignment} or valuation $\vec \theta$ for $\vec x$ is an $n$-vector in $\{0,1\}^{n}$, i.e.\ it is  a function $\theta: \vec x \rightarrow \{0,1\}^n$ where $\theta_i = \theta(x_i)$ denotes the value assigned to $x_i$ by $\theta$.
A \emph{positive term} or conjunct is $c = \bigwedge_{i\in \vec c} x_i$, where $\vec c \subseteq \vec x$.
Note that variable $c$ represents a conjunct, whereas the set $\vec c$ represents the variables of $c$.
A \emph{positive DNF} (Disjunctive Normal Form) of size $m$ is an expression of the form 
$\bigvee_{j=1}^m c_j = \bigvee_{j=1}^m \(\bigwedge_{i\in \vec c_j} x_i \)$,
where each $c_j$ $(j \in \{1, \ldots, m\})$ is a positive term of the DNF.
A \emph{positive $k$-uniform DNF} is one where each term contains exactly $k$ variables.
A \emph{positive $k$-partite $k$-uniform DNF} is one where the set of variables $\vec x$ can be partitioned into $k$ sets $(\vec x_1, \ldots, \vec x_k)$ so that each term consists of exactly $k$ variables with each variable coming from a different partition.

\smallsection{Event expressions}
We assign to each Boolean variable $x_i$ a \emph{primitive event} (we do not formally distinguish between the independent random variable $x_i$ and the event $x_i$ that it is true) which is true with probability $\PP{x_i} = p_i$. All primitive events are assumed to be independent, i.e.\ 
$\PP{x_i \wedge x_j} = \PP{x_i} \cdot \PP{x_j} = p_i \cdot p_j$,
$\forall i,j\in \{1, \ldots, n\}$ with $i \neq j$.
We are interested in the computation of probabilities of \emph{composed events}~\cite{DBLP:journals/tois/FuhrR97} and write $\phi(\vec x)$ to indicate that $\vec x = \Var(\phi)$ is the set of variables appearing in the expression $\phi$. Our focus is on calculating the probability of positive DNF event expressions.
Given independence of primitive events, the probability of a \emph{conjunct} $c$ is 
\begin{equation*}
	\PP{c} = \prod_{i \in \vec c} \PP{x_i} = \prod_{i \in \vec c} p_i 
\end{equation*}
Using the inclusion-exclusion principle, the event probability for a \emph{positive DNF} is then
\begin{align}
	&\PP{\bigvee_{j=1}^m c_j} = \sum_{k=1}^m (-1)^{k-1} \!\!\!\!\!  \sum_{1\leq j_1 < \ldots < j_k \leq m} 
		\PP{\bigwedge_{i=1}^{k} c_{j_i}}
	\label{inclusionExclusion}
\end{align}

\smallsection{Correlations and positive DNF dissociation}
Two events $\phi_1$ and $\phi_2$ are \emph{positively correlated} iff
\begin{equation*}
	\PP{\phi_1 \wedge \phi_2} > \PP{\phi_1} \cdot \PP{\phi_2} \ ,
\end{equation*}
By application of the inclusion-exclusion principle, it follows
\begin{equation*}
	\PP{\phi_1 \vee \phi_2} < 1 - \PP{\neg \phi_1} \cdot \PP{\neg \phi_2} 
\end{equation*}

\begin{lemma}[Positive correlations of terms]\label{lemma:PositiveCorrelations}
Two positive terms $c_1$ and $c_2$ over the random variables $\vec x$ are positively correlated if and only if neither of them has probability 0, and they have at least one variable $x_i$ in common with $p_i \neq 1$. 
\end{lemma}

\begin{proof}[\autoref{lemma:PositiveCorrelations}]
Assume $c_1$ and $c_2$ have a possibly empty set $\vec c_{1 \cap 2} := \vec c_1 \cap \vec c_2$ of variables in common.
We use $\vec c_{1 \setminus 2}$ for $\vec c_1 \setminus \vec c_{2}$ 
and can write
$\PP{c_1} = \prod_{i \in \vec c_1} p_i = \prod_{i \in \vec c_{1 \setminus 2}} p_i \cdot \prod_{i \in \vec c_{1 \cap 2}} p_i $. 
Hence, 
\begin{equation*}
	\PP{c_1} \cdot \PP{c_2} = \prod_{i \in \vec c_{1 \setminus 2}} p_i 
		\cdot \prod_{i \in \vec c_{2 \setminus 1}} p_i 
		\cdot \prod_{i \in \vec c_{1 \cap 2}} p_i^2 
\end{equation*}
whereas
\begin{equation*}
	\PP{c_1 \wedge c_2} = \prod_{i \in \vec c_{1 \setminus 2}} p_i 
		\cdot \prod_{i \in \vec c_{2 \setminus 1}} p_i 
		\cdot \prod_{i \in \vec c_{1 \cap 2}} p_i
\end{equation*}
Therefore,
\begin{equation*}
	\PP{c_1} \cdot \PP{c_2} = \PP{c_1 \wedge c_2}  \cdot \prod_{i \in \vec c_{1 \cap 2}} p_i
\end{equation*}
from which the proposition follows.
\qed
\end{proof}

\begin{corollary}[Positive correlations of terms]
For any two positive terms $c_1$ and $c_2$ over the random variables $\vec x$, the following hold:
\begin{align*}
	\PP{c_1 \wedge c_2} 	&\geq \PP{c_1} \cdot \PP{c_2} \\
	\PP{c_1 \vee c_2}		&\leq 1 - \PP{\neg c_1} \cdot \PP{\neg c_2} 
\end{align*}
\end{corollary}

\begin{definition}[Expression dissociation]
Assume a Boolean expression $\phi$ over variables $\vec x$. A dissociation of $\phi$
is a new expression $\phi'$ over variables $\vec x'$ so that there exists 
a substitution $\theta: \vec x' \rightarrow \vec x$ 
that transforms the new into the original expression:
$
	\phi'(\theta (\vec x')) = \phi(\vec x)
$.
The probability of the new event expression $\PP{\phi'}$ is evaluated by assigning each new variable independently the probability corresponding to its substitution in its event expression:
$
	\PP{x'} = \PP{\theta(x')}, \forall x' \in \vec x'
$.
\end{definition}

\begin{example}[Expression dissociation]
Take the two DNF expressions:
\begin{align*}
	\phi(x_1, x_2, x_3, x_4) 		&= x_1 x_3 \vee x_1 x_4 \vee x_2 x_4 \\
	\phi'(x_1, x_2, x_3, x_4, x_4') &= x_1 x_3 \vee x_1 x_4 \vee x_2 x_4'
\end{align*}
Then $\phi'(\vec x')$ is a dissociation of $\phi(\vec x)$ because $\phi(\vec x) = \phi'(\theta(\vec x'))$ for the substitution $\theta = \{(x_1, x_1), (x_2, x_2), (x_3, x_3), (x_4, x_4),$ $(x_4', x_4) \}$.
Further, $\PP{\phi} = 
	p_1 p_3 + p_1 p_4 + p_2 p_4 
	- p_1 p_3 p_4
	- p_1 p_3 p_2 p_4 
	- p_1 p_4 p_2
	+ p_1 p_3 p_4 p_2 
	$,
whereas 	
$\PP{\phi'} = 
	p_1 p_3 + p_1 p_4 + p_2 p_4 
	- p_1 p_3 p_4
	- p_1 p_3 p_2 p_4
	- p_1 p_4^2 p_2
	+ p_1 p_3 p_4^2 p_2 
	$.
Note that
\begin{align*}
\PP{\phi'} - \PP{\phi} 
	&= (p_1 p_2 p_3 p_4 - p_1 p_2 p_4) (p_4 - 1) \\
	&= (p_1 p_2 p_4) (1- p_3) (1-p_4) \geq 0 \ .
\end{align*}
\end{example}

The following lemma is now comparable to \autoref{th:bool:dissoc}:

\begin{lemma}[Positive DNF dissociation]\label{prop:positiveDNFdissociation}
For every dissociation $\phi'(\vec x')$ of a positive DNF $\phi(\vec x)$, the following holds: $\PP{\phi'} \geq \PP{\phi}$.
\end{lemma}

\begin{proof}[\autoref{prop:positiveDNFdissociation}]
We will proceed in two steps. We first show that (a) the proposition holds for any single-step dissociations. We then (b) infer by induction for multi-step dissociations. 

(a) \emph{Single-step dissociation}: A single step dissociation is one where $|\vec x'| = |\vec x| + 1$, i.e.\ there is exactly one more variable appearing in the dissociation $\phi'$ then $\phi$.
We know that the substitution $\theta$ must be surjective, i.e.\ each variable $\in \vec x$ must be mapped to at least once, since all variables must appear at least once in $\phi'(\theta (\vec x')) = \phi(\vec x)$.
It follows that the size $|\vec x'| \geq |\vec x|$. It follows that a single-step dissociation is the simplest dissociation for which $\phi$ and $\phi'$ are not trivially isomorphic.
From the pigeonhole principle, it also follows that there must be exactly two variables in $\vec x'$ that are mapped to the same variable in $\vec x$.
It also follows that the dissociation $\phi'$ must have the same structure, i.e.\ that there must be a one-to-one mapping between conjuncts in $\phi$ and $\phi'$ with corresponding conjuncts containing the same number of variables, and that two variables that $\theta$ maps to the same variable cannot appear in the same conjunct.

We assume w.l.o.g.\ that $\theta (x_1) = \theta (x_1') = x_1$ and $\theta (x_i) = x_i, \forall i \in \{2, \ldots, n\}$.
W.l.o.g., we further assume that $x_1$ appears in the terms $c_1, \ldots, c_k$ $(k\leq m)$ of the DNF $\phi$, but not in $c_{k+1}, \ldots, c_m$. 
W.l.o.g., we further consider a dissociation $\phi'$ where $x_1$ appears in the terms $c_1, \ldots, c_l$ $(1 \leq l < k)$ and
$x_1'$ in the terms $c_{l+1}, \ldots, c_k$.
We write $c_j^*$ for $c_j = x_1 c_j^*, (j \in \{1, \ldots, k\})$, i.e.\ a conjunct $c_1, \ldots, c_k$ without the variable $x_1$.
We can then write $\phi(\vec x)$ and $\phi'(\vec x')$, respectively, as
\begin{align*}
	\phi(\vec x) &=	\Big(x_1 \wedge \bigvee_{j=1}^l c_j^* \Big)
			\vee 	\Big(x_1 \wedge \!\! \bigvee_{j=l+1}^k \!\! c_j^* \Big)
			\vee 	\Big( \! \bigvee_{j=k+1}^m \!\! c_j \Big)		\\
	\phi'(\vec x') &=	\Big(x_1 \wedge \bigvee_{j=1}^l c_j^* \Big)
			\vee 	\Big(x_1' \wedge \!\! \bigvee_{j=l+1}^k \!\! c_j^* \Big)
			\vee 	\Big( \! \bigvee_{j=k+1}^m \!\! c_j \Big)
\end{align*}
with $(1 \leq l < k \leq m)$. Substituting the disjunctive expressions with $D_i$, we can write more compactly
\begin{align*}
	\phi(\vec x) &=	x_1  D_1 		\vee x_1  D_2 		\vee D_3		\\
	\phi'(\vec x') &=	x_1  D_1 	\vee x_1'  D_2 		\vee D_3		
\end{align*}
Using the inclusion-exclusion principle, we can write the event probabilities as
\begin{align*}
	\PP{\phi} 	&= 	p_1 \PP{D_1} + p_1 \PP{D_2} + \PP{D_3} \\
				& \quad - p_1 \PP{D_1 D_2} - p_1 \PP{D_1 D_3} - p_1 \PP{D_2 D_3} + p_1 \PP{D_1 D_2 D_3} \\
	\PP{\phi'} 	&= 	p_1 \PP{D_1} + p_1 \PP{D_2} + \PP{D_3} \\
				& \quad - p_1^2 \PP{D_1 D_2} - p_1 \PP{D_1 D_3} - p_1 \PP{D_2 D_3} + p_1^2 \PP{D_1 D_2 D_3} 
\end{align*}
Comparing the two expressions, we get
\begin{align}
	\PP{\phi'}& 	- \PP{\phi} =	p_1 (1-p_1)\( \PP{D_1 D_2} - \PP{D_1 D_2 D_3}	\) \geq 0 \label{eq:dissociationInequality}
\end{align}
since  $\PP{\psi_1} \geq \PP{\psi_1 \psi_2}$.

(b) \emph{Multi-step dissociation}: A $k$-step dissociation is one where single-step dissociations are consecutively applied ($\phi \rightarrow \phi' \rightarrow \phi' \rightarrow \ldots \rightarrow \phi^{(k)}$). It trivially follows from transitivity that the proposition also holds for a $k$-step dissociation. It also follows that a $k$-step dissociation $\phi^{(k)}$ has $|\vec x^{(k)}| = |\vec x| + k$ variables.

Vice versa, every dissociation with $|\vec x'| = |\vec x| + k$ can be constructed as a $k$-step dissociation as follows: Denote the number that a variable $x_i \in \vec x$ is mapped to in $\theta$ by $k(i)$. Then consider the following multi-step dissociation with $k = \sum_{i=1}^n k(i) - 1$ steps: Iterate over all variables $x_i$. For each variable with has $k(i) > 1$ dissociate the variable in $k(i) - 1$ steps so that afterwards the substitution $\theta(x_i) = \theta(x_i') = \ldots = \theta(x_i^{(k-1)}) = x_i$. This can be done by partitioning the appearances of $x_i$ in $\phi$ according to the appearances of the variables in $\phi'$ and dissociating these appearances one after the other. Hence, the proposition holds for any dissociation.
\qed
\end{proof}

Note that \autoref{eq:dissociationInequality} holds in any of 3 conditions:
\begin{enumerate}[label=\textup{(\roman*)}, itemsep=0pt, parsep=0pt, topsep = 0pt]
			
	\item $p_1$ is either 0 or 1.
	
	\item $D_1 = 0$ or $D_2 = 0$, which requires that in all conjuncts $c_j^*$ of either $D_1$ or $D_2$ (written as $D_{1/2}$) there must exist at least one variable with 0 probability: $\forall c_j^* \in D_{1/2}. \exists i. x_i \in c_j^*: p_i = 0$.

	\item $D_3 = 1$, which requires that there is at least one conjunct in which $x_1$ does not appear that is 1: $\exists c_j \in D_3. \forall i. x_i \in c_j: p_i = 1$.
\end{enumerate}

\end{document}